\documentclass[11pt,a4paper,reqno]{article}
\usepackage{amsmath,amssymb,amsthm}
\usepackage{graphicx}
\usepackage[english]{babel}
\usepackage[utf8]{inputenc}
\usepackage[T1]{fontenc}
\usepackage{mathtools}
\usepackage{hyperref}
\allowdisplaybreaks

\newcommand{\N}{\mathbb{N}}
\newcommand{\Z}{\mathbb{Z}}

\newcommand{\R}{\mathbb{R}}

\newcommand{\ud}{\,\mathrm{d}}
\DeclareMathOperator{\diverg}{div}

\DeclareMathOperator{\inte}{int}

\DeclareMathOperator{\dist}{dist}

\newtheorem{theorem}{Theorem}[section]
\newtheorem{proposition}[theorem]{Proposition}
\newtheorem{lemma}[theorem]{Lemma}
\newtheorem{corollary}[theorem]{Corollary}
\theoremstyle{definition}
\newtheorem{definition}[theorem]{Definition}
\theoremstyle{remark}
\newtheorem{remark}[theorem]{Remark}
\numberwithin{equation}{section}

\begin{document}

\title{Estimates for the Green's function of the discrete bilaplacian in dimensions 2 and 3}

\author{Stefan Müller\footnote{Hausdorff Center for Mathematics, Universität Bonn, Endenicher Allee 60, 53115 Bonn, Germany, E-Mail: \texttt{stefan.mueller@hcm.uni-bonn.de}}\and Florian Schweiger\footnote{Institut für angewandte Mathematik, Universität Bonn, Endenicher Allee 60, 53115 Bonn, Germany, E-Mail: \texttt{schweiger@iam.uni-bonn.de}}}

\date{\today}

\maketitle
\begin{center}
\textit{Dedicated to the memory of Eberhard Zeidler who has been an inspiration in so many ways. Who ever had the good fortune to meet him will never forget him.}
\end{center}
\begin{abstract} We prove estimates for the  Green's function of the discrete bilaplacian in squares and cubes in two and three dimensions which are optimal
except possibly near the corners of the square and the edges and corners of the cube. The main idea is to transfer estimates
for the continuous bilaplacian using a new discrete compactness argument and a discrete version of the Cacciopoli (or reverse
Poincaré) inequality. One application that we have in mind is the study of entropic repulsion for the membrane model from statistical physics.
\end{abstract}

\section{Introduction}
Let $V=[-1,1]^n$  and $V_N=NV\cap\Z^n$ with $n\in\N^+$ and $N\in\N^+$. Consider the Hamiltonian $H_N(\psi)=\frac12\sum_{x\in V_N}|\Delta_1\psi_x|^2$,
 where $\Delta_1$ is the discrete Laplacian and $\psi\in\R^{V_N}$ is a function on $V_N$, extended by 0 to all of $\Z^n$. The associated Gibbs measure
\[
	P_N(\mathrm{d}\psi)=\frac1{Z_N}\exp(-H_N(\psi))\prod_{x\in V_N}\ud\psi_x\prod_{x\in\Z^n\setminus V_N}\delta_0(\ud\psi_x)
\]
is then the distribution of a Gaussian random field on $\Z^n$ with 0 boundary data, the so-called membrane model. Its covariance matrix is given by the Green's function $G_N$ of the discrete bilaplacian on $V_N$ with zero boundary data outside $V_N$. We prove estimates for the  Green's function of the discrete bilaplacian for $n=2$ and $n=3$ which are optimal except possibly near the corners of the square and the edges and corners of the cube.

One motivation for our work is to understand entropic repulsion, i.e. the probability of the event $\Omega_{V_N,+}=\{\psi\colon\psi_x\ge0\ \forall x\in V_N\}$, and the behaviour of $P_N$ conditioned on $\Omega_{V_N,+}$. For this analysis a good understanding of the Green's function $G_N$ is crucial. We focus here on the subcritical dimensions $n=2$ and $n=3$, since entropic repulsion for the membrane model has already been studied 
by Kurt in the supercritical case $n\ge5$ \cite{Kurt2007} and in the critical case $n=4$ \cite{Kurt2009}.
For earlier results in the supercritical case see also Sakagawa \cite{Sakagawa2003}. 
In the case $n=1$ Hamiltonians of the form  $\sum_{x\in V_N} V(\Delta_1\psi_x)$ 
for convex $V$ have been studied by Caravenna and Deuschel  \cite{Caravenna2008,Caravenna2009} using renewal methods. 
Entropic repulsion in the gradient model corresponding to the Hamiltonian $\frac12\sum_{x\in V_N}|\nabla_1 \psi_x|^2$ with $0$ boundary condition 
was analysed
by Deuschel \cite{Deuschel1996}, see also Bolthausen-Deuschel-Giacomin  \cite{Bolthausen2001} and the survey by Velenik \cite{Velenik2006}.

More recently the membrane model with periodic boundary conditions  has also been discussed as
a scaling limit of the divisible sandpile model, see Levine et al.  \cite{Levine2016} for the expression of the odometer function 
as a shifted discrete bilaplacian field
 and Cipriani, Subhra Hazra and Ruszel \cite{Cipriani2016a,Cipriani2016b}
 for the convergence of the rescaled odometer to a continuum bilaplacian field on the unit  torus and further properties of the odometer function. 

The analysis of the discrete Green's function is very closely related to stability estimates for the inverse of the corresponding fourth order finite difference operator.
In numerical analysis such stability estimates and related convergence results estimates  go back to the seminal work of Courant, Friedrich and Lewy \cite{Courant1928},
who followed a variational approach for second and fourth order equations  and showed in particular
apriori estimates for the $L^2$ norm of the discrete derivatives and convergence of discretely biharmonic functions to
a continuous biharmonic function,
and Gerschgorin \cite{Gerschgorin1930}  who proved an error estimate of order $h^2$ for the Poisson equation on a grid of size $h$
if the continuous solution is in $C^4$. There has been a large amount of subsequent work in particular for second order
equations, including estimates under weak regularity assumptions. See, e.g., 
the recent monographs  by  Jovanovi\'c and Süli  \cite{Jovanovic2014}  and Hackbusch \cite{Hackbusch2017} for the state of the art and further references.
For special domains such as a rectangle or an orthant explicit formulae for the Green's function of the Laplace operator
are available \cite{McCrea1940,Chiarini2016}.
The biharmonic case has been less studied. Early references include \cite{Duffin1958,Mangad1967,Simpson1967}
and error estimates for low regularity solutions have been obtained in \cite{Hackbusch1981,Lazarov1981,Gavrilyuk1983,Gavrilyuk1983b,Ivanovich1986,Jovanovic2014}.

The main result of this paper is

\begin{theorem}\label{t:mainthmN}
Let $n=2$ or $n=3$, let $G_N$ be the Green's function of the discrete bilaplacian with zero boundary data outside $V_N$, and let $d(z)=\dist(z,\Z^n\setminus V_N)$. Then there exist $c,C>0$ independent of $N$ such that $G_N$ and its discrete derivatives satisfy the following estimates.
\begin{itemize}
	\item[i)] For any $x,y\in \Z^n$
\begin{align}
	|G_N(x,y)|&\le C\min\left(d(x)^{2-\frac n2}d(y)^{2-\frac n2},\frac{d(x)^2d(y)^2}{(|x-y|+1)^n}\right)\,,\label{e:estGuppN}\\
	|\nabla_xG_N(x,y)|&\le C\min\left(d(y)^{3-n},\frac{(d(x)+1)d(y)^2}{(|x-y|+1)^n}\right)\,,\label{e:estnablaxGuppN}\\
	|\nabla_x^2G_N(x,y)|&\le \begin{cases}C\log\left(1+\frac{d(y)^2}{(|x-y|+1)^2}\right)&\quad n=2\\ C\min\left(\frac{1}{|x-y|+1},\frac{d(y)^2}{(|x-y|+1)^3}\right)&\quad n=3\end{cases}\,,\label{e:estnablax2GuppN}\\
	|\nabla_x\nabla_yG_N(x,y)|&\le\begin{cases}C\log\left(1+\frac{(d(x)+1)(d(y)+1)}{(|x-y|+1)^2}\right)&\quad n=2\\ C\min\left(\frac{1}{|x-y|+1},\frac{(d(x)+1)(d(y)+1)}{(|x-y|+1)^3}\right)&\quad n=3\end{cases}\,.\label{e:estnablaxnablayGuppN}
\end{align}
	\item[ii)] For any $x\in \Z^n$
	\begin{equation}
	G_N(x,x)\ge cd(x)^{4-n}\,.\label{e:estGlowN}
	\end{equation}
\end{itemize}
\end{theorem}
$G_N$ is symmetric in $x$ and $y$, so we also have the analogous estimates for $|\nabla_yG_N(x,y)|$ and $|\nabla_y^2G_N(x,y)|$.
For the optimality of these estimates, see the discussion after Theorem~\ref{t:mainthmh}.

The estimates \eqref{e:estGuppN} and \eqref{e:estGlowN} immediately provide estimates for the variance and covariance of $\psi$ under $P_N$. From the estimates \eqref{e:estGuppN} and \eqref{e:estGlowN} one can also deduce that $P_N(\Omega_{V_N, +}) \le e^{- c N^{n-1}}$ for $n \in \{2,3\}$
and some $c > 0$
\cite{Kurt,Schweiger2016}.
In addition Theorem~\ref{t:mainthmN} implies the  following continuity estimates.
\begin{corollary}
 Under $P_N$, the random field $\psi$ satisfies
\begin{equation}  \label{e:continuity_field}
	E_N(|\psi_x-\psi_y|^2) \le
	\begin{cases} C|x-y|^{2}\log\left(2+\frac{N}{|x-y|}\right) & n=2, \\
	C |x-y|  & n=3
	\end{cases}.
\end{equation}
\end{corollary}

To show  \eqref{e:continuity_field} for $n=2$ one uses the identity
\begin{equation}  \label{e:continuity_E}
E_N(|\psi_x-\psi_y|^2)=G_N(x,x)-G_N(x,y)-G_N(y,x)+G_N(y,y)\,,
\end{equation}
as well as a discrete counterpart of the identity
\begin{align*}
	&H(x,x)-H(x,y)-H(y,x)+H(y,y) \\
	&\quad= \int_0^1 \int_0^1 \partial_s \partial_t  H(x + s(y-x), x + t(y-x)) \ud s \ud t\,,
\end{align*}
valid for every smooth function $H$, and \eqref{e:estnablaxnablayGuppN}.
For $n=3$ one uses  \eqref{e:continuity_E} and the estimates for $G(x,x) - G(x,y)$ and $G(y,y) - G(y,x)$ provided by \eqref{e:estnablaxGuppN} and its analogue for the $y$-derivative.
Since $\psi$ is a Gaussian field the estimate \eqref{e:continuity_field} and the Kolmogorov continuity criterion imply that
 the rescaled fields $ \psi'_{x'} = N^{-2 + n/2} \psi_{N x'}$ are uniformly H\"older continuous with exponents $\alpha < \alpha_n$ where 
 $\alpha_2 = 1$ and $\alpha_3 = \frac12$. More precisely
 \[ P \Big( \big\{ \psi' : \sup_{x' \ne y'}  \frac{|\psi'_{x'} - \psi'_{y'}|}{|x' - y'|^\alpha} \le K    \big\} \Big) \ge 1 - \varepsilon_\alpha(K)\]
 with $\lim_{K \to \infty} \varepsilon_\alpha(K) =0$.

\bigskip

In order to prove Theorem~\ref{t:mainthmN},
 we need  regularity improving estimates for discrete biharmonic functions and optimal decay estimates
 for various norms in annuli around the singularity. 
  The corresponding estimates for continuous biharmonic functions can be proved using well-established techniques. 
 One insight of this paper is that these estimates can 
be transferred to the discrete realm using two ingredients: a new compactness argument and the 
discrete version of the Cacciopoli (or reverse Poincaré) inequality.
It should also be possible to transfer continuous estimate to discrete estimates by using error estimates
in numerical analysis, see the discussion below Corollary~\ref{c:convergence_G}.

In order to derive the estimates in detail and to highlight the similarities between the continuous and discrete setting, it is convenient to change notation. In particular, we rescale our lattice to have width $h$, while the domain is fixed. We also shift the boundary by $h$ inwards.

Consider the lattice $(h\Z)^n$, where we assume $\frac1h\in\N$. Let $\Lambda_h^n=[0,1]^n\cap(h\Z)^n$, $\inte\Lambda_h^n=\left[\frac1h,1-\frac1h\right]^n\cap(h\Z)^n$ and let $\Delta_h$ be the discrete Laplacian on $(h\Z)^n$. Let $G_h(x,y)$ be the Green's function for $\Delta_h^2=(\Delta_h)^2$ on $\inte\Lambda_h^n$ with zero boundary values on $(h\Z)^n\setminus\inte\Lambda_h^n$. In this setting, Theorem~\ref{t:mainthmN} becomes
\begin{theorem}\label{t:mainthmh}
Let $n=2$ or $n=3$, and let $d(z)$ denote the distance of $z\in \inte\Lambda_h^n$ to $(h\Z)^n\setminus\inte\Lambda_h^n$. Then there exist $c,C>0$ independent of $h$ such that
\begin{itemize}
	\item[i)] for any $x,y\in (h\Z)^n$
\begin{align}
	|G_h(x,y)|&\le C\min\left(d(x)^{2-\frac n2}d(y)^{2-\frac n2},\frac{d(x)^2d(y)^2}{(|x-y|+h)^n}\right)\,,\label{e:estGupp}\\
	|\nabla_{h,x}G_h(x,y)|&\le C\min\left(d(y)^{3-n},\frac{(d(x)+h)d(y)^2}{(|x-y|+h)^n}\right)\,,\label{e:estnablaxGupp}\\	
	|\nabla_{h,x}^2G_h(x,y)|&\le \begin{cases}C\log\left(1+\frac{d(y)^2}{(|x-y|+h)^2}\right)&\quad n=2\\ C\min\left(\frac{1}{|x-y|+h},\frac{d(y)^2}{(|x-y|+h)^3}\right)&\quad n=3\end{cases}\,,\label{e:estnablax2Gupp}\\
	|\nabla_{h,x}\nabla_{h,y}G_h(x,y)|&\le\begin{cases}C\log\left(1+\frac{(d(x)+h)(d(y)+h)}{(|x-y|+h)^2}\right)&\quad n=2\\ C\min\left(\frac{1}{|x-y|+h},\frac{(d(x)+h)(d(y)+h)}{(|x-y|+h)^3}\right)&\quad n=3\end{cases}\,.\label{e:estnablaxnablayGupp}
\end{align}
	\item[ii)] for any $x\in (h\Z)^n$
\begin{equation}
	G_h(x,x)\ge cd(x)^{4-n}\,.\label{e:estGlow}
\end{equation}
\end{itemize}
\end{theorem}
Theorem~\ref{t:mainthmN} can be easily derived from Theorem~\ref{t:mainthmh} if one chooses $h=\frac1{2N+2}$, rescales by a factor of $2N+2$ and observes that the estimates are scale-invariant. 
One can also obtain estimates for higher discrete derivatives, see Remark~\ref{r:otherest} below.

Comparison with the Green's function of the continuous bilaplacian in the ball (see \cite[eqn. (48)]{Boggio1905} or \cite[eqn. (2.65) and Thm.\ 4.7]{Gazzola2010}),
a general bounded smooth set \cite[Thm.\  3 and Thm.\ 12]{DallAcqua2004}
or a half-space \cite[eqn. (2.66)]{Gazzola2010}
shows that the estimates in Theorem~\ref{t:mainthmh} are optimal in the interior and near the regular boundary points  (edges for $n=2$ and 
faces for $n=3$). 

Near the singular boundary points (corners for $n=2$ and edges and corners for $n=3$) the continuous regularity theory gives
a more rapid decay of biharmonic functions (and their derivatives) and hence a more rapid decay for the Green's function
 with a decay exponent $\gamma$. Our  compactness argument can be used to establish a similar decay estimate
 for all exponents  $\gamma' < \gamma$. Since the general continuum  theory provides an open interval of admissible
 exponents $\gamma$ (due to possible logarithmic terms) there is no loss in passing to the discrete estimates.
 
The general statement is rather tedious, so let us look instead at an illustrative example, the corner point $0$ of the square $(0,1)^2$. 
In this case the distance of a point $x$ from the corner point is given by $|x|$. If $|x| < \frac14 |y|$ then $|x-y| \ge \frac12 |y| \ge \frac12 d(y)$ and
 the continuous theory implies that
\begin{equation} \label{e:decay_corner}
 |G(x,y)| \le C \left(\frac{|x|}{|y|}\right)^{2 + \theta/2} d^2(y).
\end{equation}
where $0<\theta <\theta_0$, and $\theta_0\approx3.47918$. To  see this use Lemma~\ref{l:decayvertexcontinuous} and note that 
\[\|\nabla^2  G( \cdot, y)\|_{L^2(Q_{|y|/2} \cap (0,1)^2)} \le C |y|^{-1} d^2(y)\] (this follows
from the continuous counterparts of
\eqref{e:upperbound1} and Lemma~\ref{l:decaycubeoutsideav}). Moreover we have\[\sup_{Q_s \cap (0,1)^2} G(\cdot, y)  \le s \| \nabla_h^2 G(\cdot, y\|_{L^2(Q_s \cap (0,1)^2)}\] by the  Sobolev-Poincar\'e inequality and
scaling. 

The estimate  \eqref{e:decay_corner} is better than the  estimate $$G(x,y) \le \frac{d^2(x) d^2(y)}{|x-y|^2} \sim C \frac{d(x)^2}{|y|^2}   d^2(y)$$ if
$$ \frac{d(x)}{|y|} \gg    \left(\frac{|x|}{|y|}\right)^{1 + \theta/4}.$$
Note that this condition holds in particular if $|x|$ and $d(x)$ are comparable and $|x| \ll |y|$. 
 The compactness argument  shows that the discrete Green's function $G_h$  satisfies a counterpart of  \eqref{e:decay_corner} if we replace $\theta$ by any smaller exponent $\theta'$ and $C$ by $C_{\theta'}$. 
 
\bigskip

It is also easy to show that the discrete Green's function converges to the the continuous Green's function.
\begin{corollary}  \label{c:convergence_G}
Let $G(\cdot, y) \in W^{2,2}_0( (0,1)^n)$ denote the continuous Green's function, i.e., the unique weak solution of
$\Delta^2 G(\cdot, y) = \delta_y$. Extend $G_h(x,y)$ to $y \in (0,1)^n$ by piecewise constant interpolation in the second variable. 
Then  for each $y \in (0,1)^n$ the following assertions hold.
\begin{itemize}
\item[i)] \label{it:uniform_conv}
We have \[I^{pc}_h G_h(\cdot, y) \to G(\cdot, y)\quad \text{uniformly}\,,\]
where $I^{pc}_h$ denotes the piecewise constant interpolation in the first variable.
\item[ii)] \label{it:conv_n=2}
If $n=2$ then $I^{pc}_h \nabla_h G_h(\cdot, y)$ converges uniformly to $\nabla G(\cdot, y)$
and \linebreak$I^{pc}_h \nabla^2 G_h(\cdot, y)$ converges to $\nabla^2 G(\cdot, y)$ in $L^p((0,1)^2)$ for all $p < \infty$.
\item[iii)] \label{it:conv_n=3}
If $n=3$ then $I^{pc}_h \nabla_h G_h(\cdot,y)$ is uniformly bounded and converges to $\nabla G(\cdot, y)$ in $L^p((0,1)^3)$ for all $p<\infty$ and locally uniformly in $[0,1]^3 \setminus \{y\}$. Moreover $I^{pc}_h \nabla^2_h G_h(\cdot, y)$ converges to $\nabla^2 G(\cdot, y)$ in $L^p$ for all
$p < 3$. 
\end{itemize}
\end{corollary}
A slight variant of the argument given below shows that the convergence in i) is also uniform in $y$, i.e., that we
have uniform convergence of the piecewise constant interpolation of $G_h$ to $G$ in $(0,1)^n \times (0,1)^n$. 
The proof of  asssertion i)  in Corollary~\ref{c:convergence_G} uses essentially  only the elementary discrete $W^{2,2}$ estimate
in Lemma~\ref{l:upperbound} and the compact embedding from $W^{2,2}$ to $C^0$.
The other two assertion follow from Theorem~\ref{t:mainthmh} and the  local compactness argument in Section~\ref{s:inner_special}. See Section~\ref{s:proofs} for the details.

For $n=2$ quantitative estimates for the discrete $W^{2,2}$ norm of  difference between the solutions of the discretised and the continuous biharmonic equation
under weak assumptions on  the regularity of the continuous solution have been obtained by 
Lazarov \cite{Lazarov1981},  Gavrilyuk, Makarov and Pirnazarov  \cite{Gavrilyuk1983}, Gavrilyuk et al.\ 
\cite{Gavrilyuk1983b} and 
Ivanovi\'c, Jovanovi\'c and  Süli  \cite{Ivanovich1986},  see also Chapter 2.7 in \cite{Jovanovic2014} which includes estimates
for more general fourth order equations in divergence form with variable coefficients. 
More precisely, let $u \in (W^{2,2}_0 \cap W^{s,2})((0,1)^2)$ and let $\hat u_h$ be the solution of

\[\Delta_h^2 \hat u_h = K_h \ast \Delta^2 u \quad  \hbox{in $\inte \Lambda_h^2$}\]

subject to the discrete boundary conditions 
\begin{equation} \label{e:symmetric_discrete_bc}
\hat u_h(x) = 0 \quad \hbox{and} \quad \hat  u_h(x + h e_i) - \hat  u_h(x - h e_i) = 0 \quad \forall x \in \Lambda_h^2 \setminus \inte \Lambda_h^2 \quad  \forall i \in \{1, 2\}.
\end{equation}
Here  $K_h(x) = h^{-2} K(\frac{x}{h})$ and $K(z) = (1 - |z_1|)_+ (1 - |z_2|)_+$.
The boundary condition  \eqref{e:symmetric_discrete_bc} has the advantage that it leads to a higher order of consistency
compared to our boundary condition $u_h = 0$ on $(h\Z)^2 \setminus \inte \Lambda_h$ (this latter condition is arguably more natural from the point
of view of probability and statistical mechanics).
For the discrete $W^{2,2}$ norm the optimal error estimates
\begin{equation}  \label{e:optimal_error}
\| u - \hat u_h\|_{W^{2,2}(\Lambda_h)} \le C |h|^{s-2} \| u\|_{W^{s,2}((0,1)^2)}
\end{equation}
were established in  \cite{Gavrilyuk1983} for $s=3$ and in  
 \cite[Thm. 2.69]{Jovanovic2014}
for $\frac52< s < \frac72$.
In  \cite{Gavrilyuk1983}  the estimate  \eqref{e:optimal_error} is also proved for $s=4$, but under the additional condition that
that the symmetric extension $\tilde u$ of $u$ outside $(0,1)^2$ still belongs to $W^{4,2}$. This holds only if the third normal derivatives
of $u$ (which exist in the sense of trace) vanish. 

Because $K_h \ast \delta = \delta_h$ these estimates can be used to compare the continuous Green's function
$G_y \in W^{2,2}_0$ and the discrete Green's function $\hat G_{h,y}$ (defined using the boundary conditions \eqref{e:symmetric_discrete_bc}
rather than $G_{h,y} = 0$ on $(h\Z)^2 \setminus \inte \Lambda_h$) and one obtains $\| G_y - \hat G_{h,y} \|_{W^{2,2}(\Lambda_h)} \le  C_s h^{s-2} d^{3-s}(y)$ 
for $s \in (\frac52, 3)$. More precise estimates can be obtained if one applies the error estimates to $u = G_y -  \eta \tilde G_y$ where
$\tilde G_y$ is a suitable Green's function in $\R^2$ and $\eta$ is a suitable cut-off function (see below). 

One can also use Theorem~\ref{t:mainthmh} to obtain quantitative  error estimates for $G_h - G$ and its discrete derivatives 
and we  plan to pursue this elsewhere.
\bigskip

Let us briefly discuss some other approaches to prove Theorem~\ref{t:mainthmh}.
For $n=2$ the estimates \eqref{e:estGupp} and \eqref{e:estGlow} 
as well as a discrete BMO estimate
 for the mixed derivative were proved in the second author's MSc thesis  \cite{Schweiger2016}. There a different
 approach was used to obtain the estimates near the corners. One starts from a discrete biharmonic function,
 defines a careful interpolation to get a continuous functions which is biharmonic up to a small error and uses
 the continuous theory to get good estimates for that interpolation which can then be transferred back to the original discrete function. 
 This approach can in principle be extended to $n=3$, but we found the compactness argument more  flexible and more convenient to use.

Hackbusch \cite[Thm. 2.1] {Hackbusch1983} has developed a very general approach 
to derive discrete stability estimates on a scale of Banach spaces from the corresponding continuous
estimates. One advantage of the compactness method is that it avoids the construction of suitable
discrete norms and restriction and prolongation operators which is a bit delicate near the singular boundary points.

Alternatively, for $n=2$ and the symmetric boundary condition \eqref{e:symmetric_discrete_bc} one can use the optimal error estimates \eqref{e:optimal_error} in  connection with the asymptotic expansion of the discrete Green's function $\tilde G_{h,y}$ on $(h \Z)^2$ in \cite{Mangad1967} (see also Section~\ref{s:full_space_green}).
One applies the estimate   \eqref{e:optimal_error} with $s=3$ to $u = G_y -  \eta \tilde G_y$ where $\tilde G_y$ is a suitable Green's function in $\R^2$. It is not difficult to estimate the additional error term 
$w_h = G_h-\eta \tilde{G}_h-\hat{u}_h$  in the discrete $W^{2,2}$ norm by computing
$\Delta_h^2 w_h$ and testing with $w_h$. This yields the estimate
$\| \hat G_{h,y} - G_y\|_{W^{2,2}(\Lambda^2_h)} \le C h$ and the discrete inverse estimate implies that
$\| \hat G_{h,y} - G_y\|_{W^{2,\infty}(\Lambda^2_h)} \le C$. Together with the known  estimates for $\nabla^2 G_y$
one concludes in particular that 
\begin{equation} \label{e:est_hat_Gh}
|\nabla^2_h \hat G_{h,y}| \le C d^2(y)/ (|x-y| + h)^2  \quad \hbox{for $|x-y| \le C d(y)$.}
\end{equation}

To get the optimal estimate for $|x-y| \gg d(y)$ one may proceed as follows.
 From the estimate for $|x-y| \le C d(y)$ one can obtain the crucial discrete $L^\infty-L^2$ estimate \eqref{e:decaycube}  for the second discrete
 derivatives for cubes of length $2r$  that touch the boundary by using the identity $u(x) = \sum_{y \in \inte \Lambda_h}
\hat  G_h(x,y) \Delta_h^2 (\eta u)(y)h^2$  for an arbitrary  lattice function $u$ and  a suitable cut-off function $\eta$ with $|\nabla_h^k \eta| \le C_k r^{-k}$.
 For cubes which do not touch the boundary one can apply the identity $v(x) = \sum_{y \in \inte \Lambda_h} \hat G_h(x,y) \Delta_h^2 (\eta v) (y)h^2$ to $v(x) = u(x) - a - b \cdot x$ where $a$ is the average of $u$ over the cube and $b$ is the average of $\nabla_h u$. Together with the duality argument in   Lemma~\ref{l:decaycubeoutsideav}  and Theorem~\ref{t:decaycubeoutside} and similar estimates for the discrete $y$-derivatives of $G_y - \hat G_{h,y}$ this yields the estimates in  Theorem~\ref{t:mainthmh} for $n=2$ for the Green's function $\hat G_{h,y}$ which satisfies the modified boundary conditions \eqref{e:symmetric_discrete_bc}.
 The same argument applies to $G_h$.

These estimates initially hold for  $\hat G_{h,y}$ and not  for the function $G_{h,y}$ in  Theorem~\ref{t:mainthmh}. 
Note, however, that $\Delta_h^2 (G_{h,y} - \hat G_{h,y}) = 0$ in $\inte \Lambda_h$. 
Using  this fact as well as    careful comparison of the 
different boundary conditions for $\hat G_h$ and $G_h$  one can show that $\| \hat G_{h,y} - G_{h,y} \|_{W^{2,2}(\Lambda_h)}  \le Ch$.
This shows that the estimate \eqref{e:est_hat_Gh} also holds for $G_h$. For the estimates for $|x-y| \gg d(y)$ one can then argue as for 
$\hat G_h$.

\bigskip

The remainder of this paper is organised as follows. 
In Section~\ref{s:preliminaries} we introduce some notation in the discrete setting and recall 
discrete counterparts of the product rule as well as  Sobolev and Poincar\'e estimates.
In Section~\ref{s:discrete_bilaplacian} we give the weak and strong formulation of the discrete bilaplace equation and
prove the Cacciopoli inequality (or reverse Poincar\'e inequality). The proof is very similar to the argument  in the continuous case
based on testing the equation with a cut-off function  times the solution, 
but due to the discrete product rule some additional terms appear. 
In Section~\ref{s:interpolation} we associate to each discrete function a continuous function by discrete convolution with a B-spline
and prove basic estimates of the interpolation. 

Sections~\ref{s:inner_special} and~\ref{s:general_decay} contain the key estimates. The first key ingredient is an $L^\infty-L^2$ estimate for 
the discrete second derivative of
discrete biharmonic functions in cubes which may intersect the boundary (see Theorem~\ref{t:decaycube}).
This estimate is deduced from decay estimates for the second derivative of continuous biharmonic functions
using a discrete version of the Kolmogorov-Riesz-Fr\'echet compactness criterion and the Cacciopoli inequality.
The transition from continuous to discrete decay estimates is carried out in Section~\ref{s:inner_special} separately for interior cubes, cubes near
regular boundary points and cubes near singular boundary points.

The second key estimate is an  $L^\infty$ decay estimate for discretely biharmonic functions in the complement of a cube
(see Lemma~\ref{l:decaycubeoutsideav}  and Theorem~\ref{t:decaycubeoutside}). 
This follows by duality from the $L^\infty-L^2$ estimate in Theorem~\ref{t:decaycube}.
The estimates in the interior and near regular boundary points can alternatively be derived by using discrete scaled
$L^2$ estimates, i.e., by translating the
continuous Campanato regularity theory to the discrete setting (see Dolzmann \cite{Dolzmann1993,Dolzmann1999}).
For the behaviour near the singular boundary points there seems to be no argument, however, which is only based on 
scaled $L^2$-norms and testing. For ease of exposition we  use the compactness approach in all three regimes: interior points, regular boundary
points and singular boundary points. 

In Section~\ref{s:full_space_green} we recall Mangad's  \cite{Mangad1967} asymptotic expansion 
of  a Green's function $\tilde G_h$ of the discrete biharmonic operator in $(h\Z)^n$.
Finally in Section~\ref{s:proofs} we prove Theorem~\ref{t:mainthmh} and Corollary~\ref{c:convergence_G}.  An $L^2$ estimate for the second discrete derivatives
of $G_h$ is easily obtained by testing with $G_h$ and Poincar\'e's inequality.  We then choose a suitable cut-off function $\eta_h$ and use the fact that
$G_h(\cdot, y) - \eta_h(x) \tilde{G}_h(x-y)$ is biharmonic near $x=y$  
to prove estimates for the mixed third discrete derivative $\nabla_{h,x}^2 \nabla_{h,y} G_h$.
The estimates for the lower derivatives now follow essentially   by discrete integration over suitable paths
(the relevant path are the discrete counterparts of the paths used  in  \cite{DallAcqua2004}). 
For the  estimate for the first discrete derivatives for $n=3$ we directly use the discrete Sobolev embedding since integration of the second derivative would generate an unnecessary additional logarithmic term.

\section{Preliminaries}\label{s:preliminaries}
\subsection{Notation}
In the following $C$ denotes a constant that may change from line to line but is independent of $h$, unless stated otherwise.

Let $n\in\N^+$ (most of the time $n=2$ or $n=3$) be the dimension. We use standard notation for continuous quantities: We consider $\R^n$ with the standard basis $e_1,\ldots ,e_n$ and the usual Euclidean norm $|\cdot|$ and the $l^\infty$-norm $|\cdot|_\infty$. The differential of a map $f\colon\R^n\rightarrow\R^m$ is $Df=(D_if_j)_{ij}$. For $\alpha\in \N^n$ we let $D^\alpha=D_1^{\alpha_1}\ldots D_n^{\alpha_n}$. We also use the gradient $\nabla$, the Hessian $\nabla^2$, the Laplacian $\Delta$ and the divergence $\diverg$.

By $B_r(x)$ we denote the open ball of radius $r$ around $x\in\R^n$, and by $Q_r(x)=x+(-r,r)^n$ the open cube with half-sidelength $r$ around $x$. If $x=0$, we omit $x$.

Given $a\in\R^n$, we define $\tau_af=f(\cdot+a)$ for any $f$. This corresponds to shifting $f$ by $-a$.

For a function $f$ we denote by $[f]_\Omega=\frac1{|\Omega|}\int_\Omega f\ud x$ its average over the bounded open set $\Omega$.

We use the standard $L^p$-norms $\|\cdot\|_{L^p}$, Sobolev-norms $\|\cdot\|_{W^{k,p}}$ and Hölder norms $\|\cdot\|_{C^{k,\alpha}}$ and  Hölder seminorms $[\cdot]_{C^{0,\alpha}}$.

For discrete quantities we choose notation in such a way that it resembles the continuous notation. Let $h>0$ be the (typically small) lattice width. We consider the lattice $(h\Z)^n\subset\R^n$.

For $r\in\R$ we define $\lfloor r\rfloor_h:=h\left\lfloor \frac{r}{h}\right\rfloor$, the largest element of $h\Z$ less than or equal to $r$.

For $x\in (h\Z)^n$, $r\in\R$, $r\ge0$ we define $Q^h_r(x)=\{y\in (h\Z)^n\colon|y-x|_\infty\le r\}=\overline{Q_r(x)}\cap(h\Z)^n$. Then $Q^h_r(x)$ is a cube of sidelength $2\lfloor r\rfloor_h$ and center $x$. If $x=0$, we omit $x$.

Given $A_h\subset (h\Z)^n$, we define a corresponding subset $(A_h)_{pc}\subset \R^n$ as \[(A_h)_{pc}=\inte\left(A+\left[-\frac h2,\frac h2\right]^n\right)\,.\] For example, for $x\in(h\Z)^n$, $r\in h\N$, $\left(Q^h_r(x)\right)_{pc}=Q_{r+\frac h2}(x)$. For a function $u_h\colon A_h\rightarrow\R$, we define its piecewise constant interpolation $I_h^{pc}u_h\colon A_{pc}\rightarrow\R$ by $I_h^{pc}u_h(y)=u_h(x)$ on each square $x+\left[-\frac h2,\frac h2\right)^n$, where $x\in A$.

Given $u_h\colon (h\Z)^n\rightarrow \R$ and $x\in(h\Z)^n$, define the forward derivative $D^h_iu_h(x)=\frac1h(u_h(x+he_i)-u_h(x))$ and the backward derivative $D^h_{-i}u_h(x)=\frac1h(u_h(x)-u_h(x-he_i))$. Furthermore, $\nabla_{\pm h}u_h(x)=(D^h_{\pm 1}u_h(x),\ldots, D^h_{\pm n}u_h(x))$ are the forward and backward gradient, $\diverg_{\pm h}u_h(x)=\sum_{i=1}^nD^h_{\pm i}u_{h,i}(x)$ the forward and backward divergence, $\Delta_hu_h(x)=\sum_{i=1}^nD^h_{-i}D^h_iu_h(x)$ the Laplacian, and $\nabla^2_hu_h(x)=(D^h_{-i}D^h_ju_h(x))_{i,j}$ the Hessian matrix. Note that the Hessian matrix is in general not symmetric.

For a multi-index $\alpha\in\N^n$ we also define \[D_{\pm h}^\alpha u_h(x)=(D^h_{\pm 1})^{\alpha_1}\ldots(D^h_{\pm n})^{\alpha_n}u_h(x)\,,\] and for $a\in\N$, $a>2$ we set \[\nabla^a_hu_h(x)=(D^h_{-i_1}D^h_{i_2}\ldots\linebreak D^h_{i_n}u_h(x))_{i_1,i_2,\ldots ,i_n}\,.\]

If $a\in(h\Z)^n$ then $\tau_a$ also defines the shift $\tau_ax=a+x$ of $(h\Z)^n$. We denote $\tau_{\pm he_i}$ by $\tau^h_{\pm i}$. Thus $\tau^h_{\pm i}f_h(x)=f_h(x\pm he_i)$.

The discrete product rule then takes the form
\[D^h_i(f_hg_h)=(D^h_if_h)g_h+\tau^h_if_hD^h_ig_h\,.\]
When dealing with functions of several variables we use a sub- or superscript to indicate the variable with respect to which a derivative is taken. So for example in $\nabla_{h,x}\nabla_{h,y}G_h(x,y)$ we take one gradient in each variable.

As mentioned in the introduction, we set $\Lambda_h^n=[0,1]^n\cap(h\Z)^n$ and $\inte\Lambda_h^n=\left[\frac1h,1-\frac1h\right]^n\cap(h\Z)^n$. We also set $\partial\Lambda_h^n=\Lambda_h^n\setminus\inte\Lambda_h^n$.

\subsection{Function spaces and inequalities}\label{s:functionspaces}
Let $u_h,v_h\colon (h\Z)^n\rightarrow \R$. For $\Omega\subset\R^n$ measurable, $p\in[1,\infty]$, $k\in\N$, $\alpha\in[0,1]$ we define (slightly abusing notation)
\begin{align*}
	\|u_h\|_{L^p(\Omega)}&:=\|I_h^{pc}u_h\|_{L^p(\Omega)}\,,\\
	(u_h,v_h)_{L^2(\Omega)}&:=\left(I_h^{pc}u_h,I_h^{pc}v_h\right)_{L^2(\Omega)}\,,\\
	\|u_h\|_{W^{k,p}(\Omega)}&:=\left(\sum_{|\alpha|\le k}\|I_h^{pc}D_h^\alpha u_h\|_{L^p(\Omega)}^p\right)^\frac1p\,,\\
	[u_h]_{C^{0,\alpha}_h(\Omega)}&=\sup_{\substack{x,y\in\Omega\\|x-y|\ge h}}\frac{|I_h^{pc}u_h(x)-I_h^{pc}u_h(y)|}{|x-y|^\alpha}\,.
\end{align*}
For $[\cdot]_{C^{0,\alpha}_h}$ we add the index $h$ to emphasize the fact that we only take the supremum over $x,y$ with $|x-y|\ge h$.

For $A_h\subset(h\Z)^n$ these definitions take a familiar form. For example, if $p<\infty$
\begin{align*}
	\|u_h\|_{L^p((A_h)_{pc})}&=\left( \sum_{x\in A_h}h^n|u_h(x)|^p\right)^\frac1p\,,\\
	[u_h]_{C^{0,\alpha}_h((A_h)_{pc})}&=\sup_{\substack{x,y\in A_h\\x\neq y}}\frac{|u_h(x)-u_h(y)|}{|x-y|^\alpha}
\end{align*}

We extend these definitions to vector-valued functions by taking the Euclidean norm of the norms of the components.

We also set $[u_h]_\Omega=[I_h^{pc}u_h]_\Omega=\frac{1}{|\Omega|}\int_\Omega I_h^{pc}u_h$.

We then have the discrete analogues of Poincaré and Sobolev inequalities. All of them can be proved easily by applying their continuous counterpart to the piecewise multilinear interpolation of the function. We state the results that we will need.
\begin{lemma}[Poincaré inequality on cubes with 0 boundary values]
Let $p\in[1,\infty]$, let $u_h\colon(h\Z)^n\rightarrow\R$, $x\in (h\Z)^n$, $r\in h\N+\frac h2$, and suppose that $u_h=0$ on at least one of the faces of $Q^h_r(x)$. Then
\[\|u_h\|_{L^p(Q_r(x))}\le Cr\|\nabla_hu_h\|_{L^p(Q_r(x))}\]
where $C$ is independent of $h$ and $r$.
\end{lemma}
\begin{lemma}[Poincaré inequality on annuli with 0 boundary values]
Let $p\in[1,\infty]$, $u_h\colon(h\Z)^n\rightarrow\R$, let $x\in (h\Z)^n$, $r,s\in h\N+\frac h2$, $s<r$ and suppose that $u_h=0$ on at least one of the faces of $Q^h_r(x)$. Then
\[\|u_h\|_{L^p(Q_r(x)\setminus Q_s(x))}\le Cr\|\nabla_hu_h\|_{L^p(Q_r(x)\setminus Q_s(x))}\]
where $C$ only depends on $\frac sr$, $p$ and $n$.
\end{lemma}
\begin{lemma}[Sobolev-Poincaré inequality on cubes with 0 boundary values]
Let $p\in[1,\infty]$, $u_h\colon(h\Z)^n\rightarrow\R$, let $x\in (h\Z)^n$, $r\in h\N+\frac h2$, and suppose that $u_h=0$ on at least one of the faces of $Q^h_r(x)$.

If $q\in[1,\infty]$ is such that $\frac nq+1\ge\frac np$ and $(p,q)\neq(n,\infty)$, then
\[\|u_h\|_{L^q(Q_r(x))}\le Cr^{1+\frac nq-\frac np}\|\nabla_hu_h\|_{L^p(Q_r(x))}\]
and if $\alpha\in(0,1]$ is such that $\alpha+\frac np\le1$, then
\[[u_h]_{C_h^{0,\alpha}(Q_r(x))}\le Cr^{1-\frac np-\alpha}\|\nabla_hu_h\|_{L^p(Q_r(x))}\,.\]
\end{lemma}
\section{The discrete bilaplacian equation}\label{s:discrete_bilaplacian}
\subsection{Definitions and basic properties}
We consider the space of functions 
\[\Phi_h=\{u_h\colon(h\Z)^n\rightarrow\R\colon u_h(x)=0\ \forall x\in(h\Z)^n\setminus\inte\Lambda_h^n\}\,.\]

The discrete bilaplacian equation on $\Lambda_h^n$ with 0 boundary data is the equation
\begin{equation}
	\Delta_h^2u_h=f_h \ \text{ in }\inte\Lambda_h^n\label{e:bilaplace}
\end{equation}
where $f_h\colon(h\Z)^n\rightarrow\R$ is given and we are looking for a solution $u_h\in\Phi_h$. 

This equation is the discrete analogue of the bilaplace equation with clamped boundary conditions,
\[\begin{array}{rl}
		\Delta^2u=f & \text{in }[0,1]^n\\
		u=0 & \text{on }\partial[0,1]^n\\
		D_\nu u=0& \text{on }\partial[0,1]^n
	\end{array}\]

If we multiply \eqref{e:bilaplace} with a test function $\varphi_h\in\Phi_h$ and use summation by parts, we obtain the weak form of the bilaplace equation
\begin{equation}
	(\nabla_h^2u_h,\nabla_h^2\varphi_h)_{L^2(\R^n)}=(f_h,\varphi_h)_{L^2(\R^n)}\quad\forall\varphi_h\in\Phi_h\,.\label{e:bilaplaceweak}
\end{equation}
It is easy to check that \eqref{e:bilaplace} and \eqref{e:bilaplaceweak} are equivalent.

Written as a sum over lattice points, \eqref{e:bilaplaceweak} becomes
\[h^n\sum_{x\in\Lambda_h^n}\nabla_h^2u_h(x):\nabla_h^2\varphi_h(x)=h^n\sum_{x\in\inte\Lambda_h^n}f_h(x)\varphi_h(x)\,.\]
Observe that the sum on the left-hand side has nonzero terms for $x\in\Lambda_h^n$, whereas the right-hand side has nonzero terms only for $x\in\inte\Lambda_h^n$.

If we choose $\varphi_h=u_h$ in \eqref{e:bilaplaceweak}, we obtain
\[(\Delta_h^2u_h,u_h)_{L^2(\R^n)}=(\nabla_h^2u_h,\nabla_h^2u_h)_{L^2(\R^n)}=\|\nabla_h^2u_h\|^2_{L^2(\R^n)}\,.\]
Hence $\Delta_h^2$, seen as a linear operator on $\Phi_h$, is positive definite and hence invertible, and so \eqref{e:bilaplace} has a unique solution for any right-hand side $f_h$.

The discrete Green's function $G_h$ is now defined as the inverse of $\Delta_h^2$ (considered as a matrix operating on $\R^{\inte\Lambda_h^n}$ with the scalar product $\langle u_h,v_h\rangle=(u_h,v_h)_{L^2(\R^n)}$).

Let us also give an alternative description of $G_h$: The discrete delta function is given as 
\[\delta_{h,x}(y)=\begin{cases}\frac1{h^n}&\text{if }x=y\\0 &\text{otherwise}\end{cases}\,.\]
The discrete Green's function $G_h$ of $\Lambda^n_h$ is then the function $(h\Z)^n\times (h\Z)^n\rightarrow\R$ such $G_h(x,y)=0$ when $y\notin\inte\Lambda_h^n$ and such that $G_h(\cdot,y)$ is the unique solution in $\Phi_h$ of
\[\Delta_h^2u_h=\delta_{h,y} \text{ in }\inte\Lambda_h^n\]
when $y\in\inte\Lambda_h^n$.

As in the continuous case one can easily show that $G_h$ is symmetric in $x$ and $y$. We will frequently denote $G_h(x,\cdot)$ and $G_h(\cdot,y)$ by $G_{h,x}$ and $G_{h,y}$, respectively.

Let us return our attention to \eqref{e:bilaplaceweak} for a moment. If $f_h$ is given in divergence form as $\diverg_h\diverg_{-h}g_h$, this equation takes the form
\[(\nabla_h^2u_h,\nabla_h^2\varphi_h)_{L^2(\R^n)}=(g_h,\nabla_h^2\varphi_h)_{L^2(\R^n)}\]
and if we choose $\varphi_h=u_h$, we obtain the energy estimate 
\[\|\nabla_h^2u_h\|_{L^2(\R^n)}\le\|g_h\|_{L^2(\R^n)}\,.\]

\subsection{Caccioppoli inequalities}
We will need a discrete counterpart of the Cacciopoli (or reverse Poincaré) estimate for biharmonic functions (see e.g. \cite[Cap. II, Lemma 1.II]{Campanato1980}). It can be derived by testing $\Delta_h^2u_h=0$ with $\eta_hu_h$ for a suitable cut-off function $\eta_h$ and some manipulations of the error terms.
\begin{lemma}\label{l:Cacc}
Let $u_h\in\Phi_h$, $x\in(h\Z)^n$, $r>0$ and assume that $\Delta_h^2u_h(y)=0$ for all $y\in Q_{r-h}^h(x)\cap\inte\Lambda_h^n$. Then for any $0<s\le r-4h$ we have
\[\|\nabla_h^2u_h\|_{L^2(Q_{s}(x))}^2\le\frac{C}{(r-s)^4}\|u_h\|_{L^2(Q_{r}(x))}^2+\frac{C}{(r-s)^2}\|\nabla_hu_h\|_{L^2(Q_{r}(x))}^2\,.\]
\end{lemma}
The proof is similar to the continuous case. However, the fact that the discrete chain rule only holds up to translations generates additional error terms. Therefore we will give the somewhat lenghty proof in full detail. The proof is adapted from that of Lemma 2.9 in \cite{Dolzmann1993}.
\begin{proof}
By replacing $r$ by $\lfloor r-\frac h2\rfloor_h+\frac h2$ and $s$ by $\lfloor s-\frac h2\rfloor_h+\frac{3h}2$, we can assume that $r,s\in h\Z+\frac h2$ and $s\le r-3h$.

Choose a discrete cut-off function $\eta_h$ with support in $Q_{r-2h}(x)$ that is 1 on $Q_{s+h}(x)$ und such that $|\nabla_h^\kappa\eta|\le\frac{C}{(r-s)^\kappa}$ for $\kappa\leq2$. Note that $\eta_h^4u_h\in\Phi_h$, and $\eta_h^4u_h=0$ whenever $\Delta_h^2u_h\neq0$. Thus the weak form of \eqref{e:bilaplaceweak} with $\varphi_h=\eta_h^4u_h$ is
\[0=\left(\Delta_h^2u_h,\eta_h^4u_h\right)_{L^2(\R^n)}=\left(\nabla_h^2u_h,\nabla_h^2(\eta_h^4u_h)\right)_{L^2(\R^n)}\,.\]
We can expand the right-hand side and obtain
\begin{align*}
	0&=\left(\nabla_h^2u_h,\nabla_h^2(\eta_h^4u_h)\right)_{L^2(\R^n)}\\
	&=\sum_{i,j}^n\Big(D^h_{-i}D^h_ju_h,\eta_h^4D^h_{-i}D^h_ju_h\Big)_{L^2(\R^n)}\\
	&\quad+\sum_{i,j}^n\Big(D^h_{-i}D^h_ju_h,D^h_j(\eta_h^4)\tau^h_jD^h_{-i}u_h+D^h_{-i}(\eta_h^4)\tau^h_{-i}D^h_ju_h\Big)_{L^2(\R^n)}\\
	&\quad+\sum_{i,j}^n\Big(D^h_{-i}D^h_ju_h,D^h_{-i}D^h_j(\eta_h^4)\tau^h_{-i}\tau^h_ju_h\Big)_{L^2(\R^n)}\,.
\end{align*}
We can rewrite this as
\begin{align}
	&\|\eta_h^2\nabla_h^2u_h\|^2_{L^2(\R^n)}=\sum_{i,j}^n\left\|\eta_h^2D^h_{-i}D^h_ju_h\right\|^2_{L^2(\R^n)}\nonumber\\
	&\le\left|\sum_{i,j}^n\Big(D^h_{-i}D^h_ju_h,D^h_j(\eta_h^4)\tau^h_jD^h_{-i}u_h\Big)_{L^2(\R^n)}\right|\nonumber\\
	&\quad+\left|\sum_{i,j}^n\Big(D^h_{-i}D^h_ju_h,D^h_{-i}(\eta_h^4)\tau^h_{-i}D^h_ju_h\Big)_{L^2(\R^n)}\right|\nonumber\\
	&\quad+\left|\sum_{i,j}^n\Big(D^h_{-i}D^h_ju_h,D^h_{-i}D^h_j(\eta_h^4)\tau^h_{-i}\tau^h_ju_h\Big)_{L^2(\R^n)}\right|\,.\label{e:Cacc}
\end{align}
We will estimate the terms on the right-hand side separately.

Using $\frac{a^4-b^4}{a-b}=a^3+a^2b+ab^2+b^3$ for $a=\eta_h^4\circ\tau^h_j$ and $b=\eta_h^4$ we can rewrite the summands of the first term as
\begin{align*}
	&\Big(D^h_{-i}D^h_ju_h,D^h_j(\eta_h^4)\tau^h_jD^h_{-i}u_h\Big)_{L^2(\R^n)}\\
	&=\Big(D^h_{-i}D^h_ju_h,\left(\eta_h^3+\eta_h^2\tau^h_j\eta_h+\eta_h\tau^h_j\eta_h^2+\tau^h_j\eta_h^3\right)D^h_j\eta_h\tau^h_jD^h_{-i}u_h\Big)_{L^2(\R^n)}\\
	&=\Big(D^h_{-i}D^h_ju_h,4\eta_h^3D^h_j\eta_h\tau^h_jD^h_{-i}u_h\Big)_{L^2(\R^n)}\\
	&\ +\Big(D^h_{-i}D^h_ju_h,\left(\eta_h^2(\tau^h_j\eta_h-\eta_h)+\eta_h(\tau^h_j\eta_h^2-\eta_h^2)+(\tau^h_j\eta_h^3-\eta_h^3)\right)D^h_j\eta_h\tau^h_jD^h_{-i}u_h\Big)_{L^2(\R^n)}\,.
\end{align*}
The second term here is problematic\footnote{Note that in a continuous setting this term would not occur at all.}, because it does not contain a factor $\eta_h^2D^h_{-i}D^h_ju_h$. We will control it by moving a factor $\frac1h$ from the left-hand side to the right-hand side, so that we are no longer taking second derivatives of $u_h$. We obtain
\begin{align*}
	&\Big(D^h_{-i}D^h_ju_h,D^h_j(\eta_h^4)\tau^h_jD^h_{-i}u_h\Big)_{L^2(\R^n)}\\
	&=\Big(\eta_h^2D^h_{-i}D^h_ju_h,4\eta_h^3D^h_j\eta_h\tau^h_jD^h_{-i}u_h\Big)_{L^2(\R^n)}\\
	&\ +\Big(D^h_{-i}\tau^h_ju_h-D^h_{-i}u_h,\left(\eta_h^2D^h_j\eta_h+\eta_hD^h_j(\eta_h^2)+D^h_j(\eta_h^3)\right)D^h_j\eta_h\tau^h_jD^h_{-i}u_h\Big)_{L^2(\R^n)}\,.
\end{align*}
Therefore, using the Cauchy-Schwarz inequality, $ab\le\delta a^2+\frac{1}{4\delta}b^2$ and the pointwise bounds on $\eta_h$ and its derivatives we get
\begin{align*}
	&\left|\sum_{i,j}^n\Big(D^h_{-i}D^h_ju_h,D^h_j(\eta_h^4)\tau^h_jD^h_{-i}u_h\Big)_{L^2(\R^n)}\right|\\
	&=\left|\sum_{i,j}^n\Big(\eta_h^2D^h_{-i}D^h_ju_h,4\eta_h^3D^h_j\eta_h\tau^h_jD^h_{-i}u_h\Big)_{L^2(\R^n)}\right|\\
	& +\left|\sum_{i,j}^n\Big(D^h_{-i}\tau^h_ju_h-D^h_{-i}u_h,\left(\eta_h^2D^h_j\eta_h+\eta_hD^h_j(\eta_h^2)+D^h_j(\eta_h^3)\right)D^h_j\eta_h\tau^h_jD^h_{-i}u_h\Big)_{L^2(\R^n)}\right|\\
	&\le\frac14\|\eta_h^2\nabla_h^2u_h\|^2_{L^2(\R^n)}+\sum_{i,j}^n\left\|4\eta_h^3D^h_j\eta_h\tau^h_jD^h_{-i}u_h\right\|^2_{L^2(Q_{r-h}(x))}\\
	& +\frac{1}{2(r-s)^2}\sum_{i,j}^n\left\|D^h_{-i}\tau^h_ju_h-D^h_{-i}u_h\right\|^2_{L^2(Q_{r-h}(x))}\\
	& +\frac{(r-s)^2}2\sum_{i,j}^n\left\|\left(\eta_h^2D^h_j\eta_h+\eta_hD^h_j(\eta_h^2)+D^h_j(\eta_h^3)\right)D^h_j\eta_h\tau^h_jD^h_{-i}u_h\right\|^2_{L^2(Q_{r-h}(x))}\\
	&\le\frac14\|\eta_h^2\nabla_h^2u_h\|^2_{L^2(\R^n)}+\frac{C}{(r-s)^4}\|u_h\|^2_{L^2(Q_r(x))}+\frac{C}{(r-s)^2}\|\nabla_hu_h\|^2_{L^2(Q_r(x))}\,.
\end{align*}
Analogously we can find the same upper bound for the other two terms on the right-hand side of \eqref{e:Cacc}. Then we obtain
\begin{align*}
	&\|\eta_h^2\nabla_h^2u_h\|^2_{L^2(\R^n)}\\
	&\quad\le\frac34\|\eta_h^2\nabla_h^2u_h\|^2_{L^2(\R^n)}+\frac{C}{(r-s)^4}\|u_h\|^2_{L^2(Q_r(x))}+\frac{C}{(r-s)^2}\|\nabla_hu_h\|^2_{L^2(Q_r(x))}
\end{align*}
and hence
\[\|\eta_h^2\nabla_h^2u_h\|^2_{L^2(\R^n)}\le\frac{C}{(r-s)^4}\|u_h\|^2_{L^2(Q_r(x))}+\frac{C}{(r-s)^2}\|\nabla_hu_h\|^2_{L^2(Q_r(x))}\,.\]

This implies the claim, once one notes that \[\|\nabla_h^2u_h\|_{L^2(Q_s(x))}\le\|\eta_h^2\nabla_h^2u_h\|_{L^2(\R^n)}\,.\] 
\end{proof}

\section{Interpolation}\label{s:interpolation}
We want to deduce discrete estimates from their continuous counterparts using compactness arguments. To do so, we need an interpolation operator that turns discrete functions into continuous functions having similar features. The most important property of this interpolation operator that we require is that the continuous derivatives of the output are comparable to the discrete derivatives of the input.

To construct such an operator we use B-splines (cf., e.g., \cite[§4.4]{Schumaker1981}): For $m\ge1$, $x\in\R$ the $m$-th normalized B-spline is given by 
\[N^m(x)=m\sum_{i=0}^m\frac{(-1)^i\binom{m}{i}\max(x-i,0)^{m-1}}{m!}\,.\]
The function $N^m$ is piecewise a polynomial of degree $m-1$, has support in $[0,m]$ and satisfies $\sum_{z\in\Z}N^m(x-z)=1$ for all $x\in\R$. Furthermore its discrete and continuous derivatives are closely related. Indeed we have
\begin{equation}
	\partial_xN^m(x)=N^{m-1}(x)-N^{m-1}(x-1)=D^1_{-1}N^{m-1}(x)\label{e:bspline1}
\end{equation}
for all $x\in\R$ (see \cite{Schumaker1981} for proofs).

We need a multidimensional version of these splines which is also adapted to the lattice $(h\Z)^n$. So for $h>0$, $\mu=(\mu_1,\ldots, \mu_n)\in\N^n$ with $\mu_i\ge1$
let 
\[N^\mu_h(x_1,\ldots, x_n)=N^{\mu_1}\left(\frac{x_1}{h}\right)\cdots N^{\mu_n}\left(\frac{x_n}{h}\right)\,.\]
It follows easily from \eqref{e:bspline1} that for any $\alpha\in\N^n$ with $\alpha_i<\mu_i$ for all $i$ we have 
\begin{equation}
	D^\alpha N^\mu_h=D^\alpha_{-h}N^{\mu-\alpha}_h\,.\label{e:bspline2}
\end{equation}

Using this, we can define our interpolation operator:
\begin{definition}
Let $h>0$, $\mu=(\mu_1,\ldots ,\mu_n)\in\N^n$ with $\mu_i\ge1$ for all $i$. Define $J^\mu_h\colon\R^{(h\Z)^n}\rightarrow L^1_{loc}(\R^n)$ by
\[(J^\mu_hu_h)(x)=\sum_{z\in (h\Z)^n}u_h(z)N^\mu_h(x-z)\]
and extend $J^\mu_h$ to vector-valued functions component-wise.
\end{definition}
Note that $N^\mu_h$ has compact support so that the above sum has only finitely many nonzero terms.

$J^\mu_h$ does not interpolate the values of $u_h$ (i.e. in general we will not have $J^\mu_hu_h(x)=u_h(x)$ for all $x\in (h\Z)^n$). The maps $J^\mu_hu_h$ and $u_h$, however, share so many properties that we still call $J^\mu_h$ an interpolation operator.

Let us collect some properties of $J_h^\mu$.
\begin{proposition}\label{p:interpolationproperties}
Let $J^\mu_h$ be the family of interpolation operators that we have just defined, and let $u_h\colon (h\Z)^n\rightarrow\R$.
\begin{itemize}
	\item[i)] $J^\mu_h$ is linear.
	\item[ii)] $J^\mu_hu_h$ is piecewise a polynomial and is in the Sobolev space $W^{(\min_i\mu_i)-1,2}_{loc}$
	\item[iii)] $J^\mu_h$ is local in the sense that $(J^\mu_hu_h)(x)$ only depends on the values of $u_h$ in $Q_{(\max_i\mu_i)h}(x)$.
	\item[iv)] $J^\mu_h$ preserves constant functions, i.e. $(J^\mu_hc)(x)=c$ for any $c\in\R$ and any $x\in\R^n$.
	\item[v)] For every $\alpha$ with $\alpha_i<\mu_i$ we have $(D^\alpha J^\mu_hu_h)(x)=(J^{\mu-\alpha}_h(D^\alpha_hu_h))(x)$.
	\item[vi)] For every $\alpha$ with $\alpha_i<\mu_i$ and any $p\in[1,\infty]$ there is a constant $C=C(\mu,\alpha,n,p)$ such that for any $x\in\R^n$ and any $r\ge s+(1+\max_i\mu_i)h$ we have 
	\begin{equation}
		\|D^\alpha J^\mu_hu_h\|_{L^p(Q_s(x))}\le C\|D_h^\alpha u_h\|_{L^p(Q_r(x))}\label{e:interpolation1}
	\end{equation}
	 and 
	\begin{equation}
		\|D_h^\alpha u_h\|_{L^p(Q_s(x))}\le C\|D^\alpha J^\mu_hu_h\|_{L^p(Q_r(x))}\,.\label{e:interpolation2}
	\end{equation}
\end{itemize}
\end{proposition}
\begin{proof}
Properties i), ii) and iii) are obvious. Property iv) easily follows from $\sum_{z\in\Z}N^m(x-z)=1$ for all $x\in\R$, so it remains to prove v) and vi).

For v), note that we can assume that $u_h$ is zero far away from $x$ by iii). This means that all sums in the following calculations have only finitely many nonzero terms. Now, using \eqref{e:bspline2}, we can calculate that
\begin{align*}
	(D^\alpha J^\mu_hu_h)(x)&=D^\alpha \left(\sum_{z\in (h\Z)^n}u_h(z)N^\mu_h(x-z)\right)\nonumber\\
	&=\sum_{z\in (h\Z)^n}u_h(z)D^\alpha N^\mu_h(x-z)\nonumber\\
	&=\sum_{z\in (h\Z)^n}u_h(z)D_{-h}^\alpha N^{\mu-\alpha}_h(x-z)\nonumber\\
	&=\sum_{z\in (h\Z)^n}D_h^\alpha u_h(z)N^{\mu-\alpha}_h(x-z)=(J^{\mu-\alpha}_h(D^\alpha_hu_h))(x)\,.
\end{align*}

Finally we prove vi). In view of v) it is sufficient to consider the case $\alpha=0$ here. We can also assume that $x\in (h\Z)^n$ and $r,s\in h\N+\frac h2$, $r\ge s+(\max_i\mu_i)h$ (otherwise move $x$ to the nearest lattice point, and replace $r$ and $s$ by $\lfloor r-\frac h2\rfloor_h+\frac h2$ and $\lfloor s-\frac h2\rfloor_h+\frac {3h}2$ respectively).

Let $y\in Q^h_s(x)$. The definition of $J^\mu_h$ immediately implies
\[\|J^\mu_hu_h\|_{L^\infty(Q_{h/2}(y))}\le C\sup_{\substack{z\in(h\Z)^n\\|z-y|\le(\max_i\mu_i)h}}|u_h(z)|\]
and thus
\[\|J^\mu_hu_h\|_{L^p(Q_{h/2}(y))}^p\le C\sum_{\substack{z\in(h\Z)^n\\|z-y|\le(\max_i\mu_i)h}}|u_h(z)|^p\le C\|u_h\|_{L^p(Q_{(\max_i\mu_i+1/2)h}(y))}^p\,.\]
If we sum this over all $y\in Q^h_s(x)$, we easily obtain \eqref{e:interpolation1}.

For \eqref{e:interpolation2}, by a similar argument it suffices to show
\begin{equation}
	|u_h(y)|\le C\|J^\mu_hu_h\|_{L^p(Q_{h/2}(y))}\label{e:interpolation3}
\end{equation}
for all $y\in Q^h_s(x)$.

One can see this as follows: $N_h^\mu$ has support $[0,\mu_1]\times\cdots\times[0,\mu_n]$. This means that the values of $J^\mu_hu_h$ in $Q_{h/2}(y)$ depend on the finitely many values $\{u_h(z)\}_{z\in I_y}$, where $I_y:= [y_1-\mu_1]\times\cdots\times[y_n-\mu_n]\cap(h\Z)^n$ and no others. Furthermore by linear independence of the B-splines (see \cite[Theorem 4.18]{Schumaker1981} for the one-dimensional case; the $n$-dimensional case is analogous) $J^\mu_hu_h$ is identically 0 in $Q_{h/2}(y)$ only if all $\{u_h(z)\}_{z\in I_y}$ are 0. This means that $\|J^\mu_hu_h\|_{L^p(Q_{h/2}(y))}$ is not only a seminorm on $\R^{I_y}$ but actually a norm. Now all norms on a finite-dimensional vector space are equivalent, so in particular
\[\|u_h\|_{l^2(I_y)}=\left(\sum_{z\in I_y}|u_h(z)|^2\right)^\frac12\leq C\|J^\mu_hu_h\|_{L^p(Q_{h/2}(y))}\]
for a constant $C$ that is independent of $y$. This immediately implies \eqref{e:interpolation3}.
\end{proof}
Using these interpolation operators $J^\mu_h$ we define the two operators that we will actually use most often: One is $J_h:=J^{(3,3,\ldots, 3)}_h$ and the other is the matrix interpolation operator $\tilde{J}_h$ given by $(\tilde{J}_h)_{ij}=J^{(3,3,\ldots, 3)-e_i-e_j}_h\circ\tau^h_i$ (for example $(\tilde{J}_h)_{11}=J^{(1,3,\ldots, 3)}_h\circ\tau^h_1$).

One easily checks using parts ii) and v) of Proposition~\ref{p:interpolationproperties} that for any $f_h\colon(h\Z)^n\rightarrow\R$ we have $J_hf_h\in W^{2,2}_{loc}(\R^n)$ and
\begin{equation}
	\nabla^2J_hf_h=\tilde{J}_h\nabla_h^2f_h\,.\label{e:propertyJ}
\end{equation}

\section{\texorpdfstring{Inner decay estimates for discrete biharmonic\\ functions: special cases}{Inner decay estimates for discrete biharmonic functions: special cases}}  \label{s:inner_special}
Our goal is to prove an $L^\infty$-$L^2$ estimate for discrete biharmonic functions (see Theorem~\ref{t:decaycube}): If $u_h\in\Phi_h$, $x\in\Lambda_h^n$, $r>0$ and $\Delta_h^2u_h(y)=0$ for all $y\in Q_{r-h}(x)\cap\inte\Lambda^n_h$, then, for all $z\in Q_{\frac r2}(x)\cap\Lambda_h^n$,
\[|\nabla_h^2u_h(z)|\le\frac{C}{r^\frac n2}\|\nabla_h^2u_h\|_{L^2(Q_r(x))}\,.\]
To prove this estimate it will be necessary to distinguish where $x$ lies in relation to $\partial\Lambda_h^n$: $x$ can be far inside $\Lambda_h^n$, near a face, near an edge or near a vertex. In the following subsections we will study these cases separately and prove some decay estimates that we will then assemble to prove the aforementioned estimate.

\subsection{Full space}

\begin{lemma}\label{l:decayfull}
Let $u_h\colon(h\Z)^n\rightarrow\R$, let $x\in(h\Z)^n$, $r>0$. Suppose $\Delta_h^2u_h(y)=0$ for all $y\in Q_{r-h}^h(x)$. Then
\[|\nabla_h^2u_h(x)|\le\frac{C}{r^\frac n2}\|\nabla_h^2u_h\|_{L^2(Q_r(x))}\,.\]
\end{lemma}

The main tool to prove this statement will be the following estimate:

\begin{lemma}\label{l:decayfull2}
There exist constants $M\in\N$, $0<\rho<\frac12$ with the following property: Let $u_h\colon(h\Z)^n\rightarrow\R$, $r>0$, such that $\Delta_h^2u_h(y)=0$ for all $y\in Q_{r-h}^h$. Assume that $\rho r\ge Mh$. Then we have that
\[\left\|\nabla_h^2u_h-\left[\nabla_h^2u_h\right]_{Q_{\rho r}}\right\|^2_{L^2(Q_{\rho r})}\le \rho^{n+1}\left\|\nabla_h^2u_h-\left[\nabla_h^2u_h\right]_{Q_r}\right\|^2_{L^2(Q_r)}\,.\]
\end{lemma}

We will prove this lemma by contradiction using a compactness argument and the following decay estimate for continuous biharmonic functions:
\begin{lemma}\label{l:decayfullcontinuous}
Let $0<s\le \frac r2$, $u\in W^{2,2}(Q_r)$ such that $\Delta^2u=0$ weakly in $Q_r$. Then we have
\begin{equation}
	\left\|\nabla^2u-\left[\nabla^2u\right]_{Q_s}\right\|^2_{L^2(Q_s)}\le C\left(\frac{s}{r}\right)^{n+\frac32}\left\|\nabla^2u-\left[\nabla^2u\right]_{Q_r}\right\|^2_{L^2(Q_r)}\,.\label{e:decayfullcontinuous1}
\end{equation}
\end{lemma}

\begin{proof}
The estimate \eqref{e:decayfullcontinuous1} expresses the fact that the second derivatives of biharmonic functions are in the Campanato space $\mathcal{L}^{2,n+\frac32}\simeq C^{0,3/4}$. The easiest way to show it is to use Schauder estimates for higher order elliptic equations as follows.

By scaling we can assume $r=1$. By replacing $u$ with $u-\frac12\left[\nabla^2u\right]_{Q_1}:x\otimes x$ we can assume that $\left[\nabla^2u\right]_{Q_1}=0$. Now by Schauder estimates (see e.g. \cite[Theorem 6.4.8]{Morrey1966} or \cite[Cap. II, Teorema 6.I]{Campanato1980}) we have that any $C^{0,\alpha}$-Hölder seminorm of $\nabla^2u$ in $Q_{1/2}$ is bounded by the $L^2$-norm of $\nabla^2u$ in $Q_1$. In particular, we have
\[\left[\nabla^2u\right]_{C^{0,\frac34}(Q_{1/2})}\le C\left\|\nabla^2u\right\|_{L^2(Q_1)}\,.\]
On the other hand, Jensen's inequality easily yields that
\begin{align*}
\left\|\nabla^2u-\left[\nabla^2u\right]_{Q_s}\right\|^2_{L^2\left(Q_s\right)}&\le\frac1{|Q_s|}\int_{Q_s}\int_{Q_s}|\nabla^2u(y)-\nabla^2u(y')|^2\ud y\ud y'\\
&\le Cs^{n+\frac32}\left[\nabla^2u\right]_{C^{0,\frac34}(Q_{1/2})}^2\,.
\end{align*}
Together with the previous estimate this yields the result. 
\end{proof}

We will also need a local version of the well-known Kolmogorov-Riesz-Fréchet compactness theorem.
\begin{lemma}\label{l:kolmogorov}
Let $p\in[1,\infty)$, let $U,V,W\subset\R^n$ be open with $U$ compactly contained in $V$, and $V$ compactly contained in $W$. Let $A$ be a subset of $L^p(W)$.
\begin{itemize}
	\item[i)] If $A$ is bounded in $L^p(W)$ and
\[\lim_{\delta\rightarrow0}\sup_{f\in A}\|\tau_\delta f-f\|_{L^p(V)}=0\]
then $A$ (or rather the restriction of the elements of $A$ to $U$) is precompact in $L^p(U)$.
	\item[ii)] If $A$ is precompact in $L^p(W)$ then
\[\lim_{\delta\rightarrow0}\sup_{f\in A}\|\tau_\delta f-f\|_{L^p(V)}=0\,.\]
\end{itemize}
\end{lemma}
\begin{proof}
Part i) follows by applying the usual Kolmogorov-Riesz-Fréchet compactness theorem (see e.g. \cite[Corollary 4.27 and Exercise 4.34]{Brezis2011}) to the family $\{\eta f\colon f\in A\}$, where $\eta$ is a smooth cut-off function that is 1 on $U$ and 0 outside of $V$.

For part ii) let $\tilde{V}$ be open such that $V$ is compactly contained in $\tilde{V}$ and $\tilde{V}$ is compactly contained in $W$, and let $\zeta$ be a cut-off function that is 1 on $\tilde{V}$ and 0 outside of $W$. Then the family $\{\zeta f\colon f\in A\}$ is precompact in $L^p(\R^n)$ and the statement is obtained by applying the converse of the Kolmogorov-Riesz-Fréchet compactness theorem to that family. 
\end{proof}

After these preparations we can return to the proofs of Lemma~\ref{l:decayfull} and Lemma~\ref{l:decayfull2}.

\begin{proof}[Proof of Lemma~\ref{l:decayfull2}]$ $\\
\emph{Step 1: Set-up of the compactness argument}\\
Let the constant $\rho\le\frac12$ be fixed later, and suppose that the statement for that fixed $\rho$ is wrong. Then for any $k\in\N$ there exist $M_k\ge k$, $h_k>0$, $u_{h_k}\colon(h_k\Z)^n\rightarrow\R$, $r_k>0$ such that
\begin{equation}\label{e:decayfull2_1}
	\left\|\nabla_{h_k}^2u_{h_k}-\left[\nabla_{h_k}^2u_{h_k}\right]_{Q_{\rho r_k}}\right\|^2_{L^2(Q_{\rho r_k})}>\rho^{n+1}\left\|\nabla_{h_k}^2u_{h_k}-\left[\nabla_{h_k}^2u_{h_k}\right]_{Q_{r_k}}\right\|^2_{L^2(Q_{r_k})}\,.
\end{equation}
By rescaling the lattice by a factor of $r_k$, we can assume that all the $r_k$ are equal to 1. Because $h_k\le\frac{\rho}{M_k}\le\frac{\rho}{k}$, we have that $h_k\rightarrow0$. Omitting finitely many $k$, we can assume that all $h_k$ are small (less than $\frac{1}{1000}$, say).

By replacing $u_{h_k}$ with $u_{h_k}-\frac12\left[\nabla^2_{h_k}u_{h_k}\right]_{Q_1}:x\otimes x$ we can assume that $\left[\nabla^2_{h_k}u_{h_k}\right]_{Q_1}=0$, and by scaling we can assume that $\left\|\nabla_{h_k}^2u_{h_k}\right\|_{L^2(Q_1)}=1$ (note that $\nabla_{h_k}^2u_{h_k}$ cannot be identically 0, as then $u_{h_k}$ would be affine, and so both sides of \eqref{e:decayfull2_1} would be 0). Then \eqref{e:decayfull2_1} implies that
\begin{equation}\label{e:decayfull2_4}
	\left\|\nabla_{h_k}^2u_{h_k}-\left[\nabla_{h_k}^2u_{h_k}\right]_{Q_{\rho}}\right\|^2_{L^2(Q_{\rho})}>\rho^{n+1}\,.
\end{equation}
Finally, we replace $u_{h_k}$ by $u_{h_k}-a_k-b_k\cdot x$, where $a_k\in\R$, $b_k\in\R^n$ are constants that will be chosen below (such that equation \eqref{e:averagev} is satisfied). This leaves $\nabla_{h_k}^2u_{h_k}$ unaffected, so all the above statements about $\nabla_{h_k}^2u_{h_k}$ remain true.

We let $v_k=J_{h_k}u_{h_k}$, where $J_{h_k}=J_{h_k}^{(3,\ldots,3)}$ is the interpolation operator introduced in Section~\ref{s:interpolation}. From $\left\|\nabla_{h_k}^2u_{h_k}\right\|_{L^2(Q_1)}=1$ and Proposition~\ref{p:interpolationproperties} vi) we immediately conclude that $\|\nabla^2v_k\|_{L^2(Q_{13/14})}\le C$.

Now we choose $a_k$ and $b_k$ in such a way that 
\begin{equation}
	[v_k]_{Q_{13/14}}=0\,,\qquad
	[\nabla v_k]_{Q_{13/14}}=0\,.\label{e:averagev}
\end{equation}

The Poincaré inequality on $Q_{13/14}$ implies that \[\|v_k\|_{W^{2,2}(Q_{13/14})}\le C\|\nabla^2v_k\|_{L^2(Q_{13/14})}\le C\,.\]Therefore the $v_k$ are bounded in $W^{2,2}(Q_{13/14})$ and hence have a subsequence (not relabeled) that converges weakly to some $v\in W^{2,2}(Q_{13/14})$.

\emph{Step 2: $\Delta^2v=0$}\\
We claim that $\Delta^2v=0$ weakly in $Q_{13/14}$. To prove this, let $\varphi\in C_c^\infty(Q_{13/14})$ be arbitrary and let $\varphi_{h_k}$ be its restriction to $(h_k\Z)^n$. We need to prove that $\int_{Q_{13/14}}\nabla^2v:\nabla^2\varphi\ud x=0$.

We have by \eqref{e:propertyJ} that
\begin{align*}
	&\int_{Q_{13/14}}\nabla^2v_k:\nabla^2\varphi\ud x=\int_{Q_{13/14}} \nabla^2 J_{h_k}u_{h_k}:\nabla^2 \varphi\ud x\\
	&=\int_{Q_{13/14}} \tilde{J}_{h_k}\nabla^2_{h_k} v_k:\nabla^2 \varphi\ud x\\
	&=\sum_{i,j=1}^n\int_{Q_{13/14}} J^{(3,3,\ldots ,3)-e_i-e_j}_{h_k}\circ\tau^{h_k}_iD^{h_k}_{-i}D^{h_k}_j v_kD_iD_j \varphi\ud x\\
	&=\sum_{i,j=1}^n\int_{Q_{13/14}} \sum_{z\in(h_k\Z)^n}N^{(3,3,\ldots ,3)-e_i-e_j}_{h_k}(x-z)D^{h_k}_iD^{h_k}_j u_{h_k}(z)D_iD_j \varphi(x)\ud x\\
	&=\sum_{i,j=1}^n\sum_{z\in(h_k\Z)^n}D^{h_k}_iD^{h_k}_j u_{h_k}(z)\int_{Q_{13/14}} N^{(3,3,\ldots ,3)-e_i-e_j}_{h_k}(x-z)D_iD_j\varphi(x)\ud x\,.
\end{align*}
Now Taylor expansion and the fact that $\int_{Q_{13/14}} N^{(3,3,\ldots ,3)-\delta_i-\delta_j}_{h_k}=1$ imply that 
\begin{align*}
	\int_{Q_{13/14}} N^{(3,3,\ldots ,3)-e_i-e_j}_{h_k}(x-z)D_iD_j \varphi(x)\ud x&=D_iD_j\varphi(z)+O(h_k)\\
	&=D^{h_k}_iD^{h_k}_j\varphi_{h_k}(z)+O(h_k)
\end{align*}
In addition, from $\Delta_{h_k}^2u_{h_k}=0$ in $Q_{13/14}$ we conclude that
\begin{align*}
	&\sum_{i,j=1}^n\sum_{z\in (h_k\Z)^n}D^{h_k}_iD^{h_k}_ju_{h_k}(z)D^{h_k}_iD^{h_k}_j\varphi_{h_k}(z)\\
	&\quad=\sum_{i,j=1}^n\sum_{z\in (h_k\Z)^n}D^{h_k}_{-i}D^{h_k}_ju_{h_k}(z)D^{h_k}_{-i}D^{h_k}_j\varphi_{h_k}(z)\\
	&\quad=(\nabla_{h_k}^2u_{h_k},\nabla_{h_k}^2\varphi_{h_k})_{L^2(\R^n)}=0
\end{align*}
and so we obtain
\[\left|\int_{Q_{13/14}}\nabla^2v_k:\nabla^2\varphi\ud x\right|\le C\left\|\nabla_{h_k}^2u_{h_k}\right\|_{L^2(Q_1)}h_k=Ch_k\,.\]
Using weak convergence of $\nabla^2v_k$ we can pass to the limit here and get
\[\int_{Q_{13/14}}\nabla^2v:\nabla^2\varphi\ud x=0\,.\]

\emph{Step 3: Strong convergence of $v_k$}\\
Let $w_k=I_{h_k}^{pc}\nabla_{h_k}^2u_{h_k}$. We claim that both $\nabla^2v_k$ and $w_k$ converge strongly in $L^2(Q_{1/2})$ to $\nabla^2v$.

\emph{Step 3.1: Precompactness of $w_k$}\\
We first prove that $(w_k)_{k\in\N}$ is precompact in $L^2(Q_{4/7})$.

Because $(\nabla_{h_k}^2u_{h_k})$ is bounded in $L^2(Q_1)$, $w_k$ is bounded in $L^2(Q_1)$. So, according to Lemma~\ref{l:kolmogorov} i), it suffices to verify that 
\begin{equation}\label{e:decayfull2_2}
	\lim_{\substack{a\in\R^n\\|a|\rightarrow0}}\sup_{k\in\N}\|\tau_a w_k-w_k\|_{L^2(Q_{5/7})}=0\,.
\end{equation}
Let $a\in(h\Z)^n$ such that $|a|\le\frac17$. Then $\Delta_{h_k}^2(\tau_au_{h_k}-u_{h_k})=0$ in $Q_{11/14}$, so by the Cacciopoli inequality we obtain
\begin{align*}&\|\nabla_{h_k}^2(\tau_au_{h_k}-u_{h_k})\|_{L^2(Q_{5/7}(x))}^2\le C\|\tau_au_{h_k}-u_{h_k}\|_{L^2(Q_{11/14}(x))}^2\\
&\quad+C\|\nabla_{h_k}(\tau_au_{h_k}-u_{h_k})\|_{L^2(Q_{11/14}(x))}^2\,.
\end{align*}
Here the left-hand side is equal to $\|\tau_aw_k-w_k\|^2_{L^2(Q_{5/7})}$, while we can use Proposition~\ref{p:interpolationproperties} vi) to bound the right-hand side. We obtain
\[\|\tau_aw_k-w_k\|^2_{L^2(Q_{5/7})}\le C\|\tau_av_k-v_k\|_{L^2(Q_{6/7}(x))}^2+C\|\tau_a\nabla v_k-\nabla v_k\|_{L^2(Q_{6/7}(x))}^2\,.\]
Recall that $(v_k)$ is bounded in $W^{2,2}(Q_{13/14})$. Hence by the compact Sobolev embedding, $(v_k)$ and $(\nabla v_k)$ are precompact in $L^2(Q_{13/14})$. Thus by Lemma~\ref{l:kolmogorov} ii),
\[\lim_{a\rightarrow0}\sup_{k\in\N}\left(\|\tau_av_k-v_k\|_{L^2(Q_{6/7}(x))}^2+\|\tau_a\nabla v_k-\nabla v_k\|_{L^2(Q_{6/7}(x))}^2\right)=0\]
(note that this expression is defined for all $a>0$, not just those in $(h\Z)^n$).

In particular,
\[\lim_{\delta\rightarrow0}\sup_{k\in\N}\sup_{\substack{a\in(h_k\Z)^n\\|a|\le\delta}}\left(\|\tau_av_k-v_k\|_{L^2(Q_{6/7}(x))}^2+\|\tau_a\nabla v_k-\nabla v_k\|_{L^2(Q_{6/7}(x))}^2\right)=0\]
and therefore
\[\lim_{\delta\rightarrow0}\sup_{k\in\N}\sup_{\substack{a\in(h_k\Z)^n\\|a|\le\delta}}\|\tau_aw_k-w_k\|_{L^2(Q_{5/7}(x))}=0\,.\]

It remains to consider shifts $\tau_a$ where $a\notin (h_k\Z)^n$. This is possible because $w_k$ is piecewise constant on cubes of sidelength $h_k$. This easily implies that for any $a\in\R^n$ we have
\[\|\tau_aw_k-w_k\|_{L^2(Q_{9/14}(x))}\le C\sup_{\substack{b\in(h_k\Z)^n\\|b-a|\le h_k}}\|\tau_bw_k-w_k\|_{L^2(Q_{5/7}(x))}\,.\]
Combining this with the previous estimate we find that
\[\lim_{\delta\rightarrow0}\sup_{k\in\N}\sup_{\substack{a\in\R^n\\|a|\le\delta+h_k}}\|\tau_aw_k-w_k\|_{L^2(Q_{9/14}(x))}=0\,.\]
Because $h_k\rightarrow0$, this implies
\begin{equation}\label{e:decayfull2_3}
	\lim_{\substack{a\in\R^n\\|a|\rightarrow0}}\limsup_{k\rightarrow\infty}\|\tau_aw_k-w_k\|_{L^2(Q_{9/14}(x))}=0\,.
\end{equation}
We finally show that \eqref{e:decayfull2_3} already implies \eqref{e:decayfull2_2}. It follows from \eqref{e:decayfull2_3} that for every fixed $\varepsilon>0$ there are $\delta>0$, $K\in\N$ such that $\sup_{k\ge K}\|\tau_aw_k-w_k\|_{L^2(Q_{9/14}(x))}\le\varepsilon$ for all $a$ with $|a|\le\delta$. For the finitely many $k<K$, we use that $\lim_{\substack{a\in\R^n\\|a|\rightarrow0}}\|\tau_aw_k-w_k\|_{L^2(Q_{9/14}(x))}=0$ to see that for a potentially smaller $\delta'$ we have $\sup_{k\in\N}\|\tau_aw_k-w_k\|_{L^2(Q_{9/14}(x))}\le\varepsilon$ for all $a$ with $|a|\le\delta'$.

Therefore the sequence $(w_k)$ is precompact in $L^2(Q_{4/7}(x))$. Choose a subsequence (not relabeled) converging strongly to some $w\in L^2(Q_{4/7}(x))$.

\emph{Step 3.2: Strong convergence of $(\nabla^2v_k)$ and $w=\nabla^2v$}\\
We split $w$ into a smooth part and a part with small $L^2$-norm. Let $\varepsilon>0$ be arbitrary, and choose a $w^{(\varepsilon)}$ in $C_c^\infty(Q_{4/7})$ such that $\|w-w^{(\varepsilon)}\|_{L^2(Q_{4/7})}\le\varepsilon$. We denote the restriction of $w^{(\varepsilon)}$ to $(h_k\Z)^n$ by $w^{(\varepsilon)}_{h_k}$. Using Taylor expansion, one immediately verifies that then $I_{h_k}^{pc}w^{(\varepsilon)}_{h_k}$ and $\tilde{J}_{h_k}w^{(\varepsilon)}_{h_k}$ converge to $w^{(\varepsilon)}$ in $L^2(Q_{4/7})$ and $L^2(Q_{1/2})$, respectively.

This means in particular that
\[\lim_{k\rightarrow\infty}\|w^{(\varepsilon)}_{h_k}-\nabla_{h_k}^2u_{h_k}\|_{L^2(Q_{4/7})}=\|w^{(\epsilon)}-w\|_{L^2(Q_{4/7})}\le\varepsilon\,.\]
Using Proposition~\ref{p:interpolationproperties} vi), we conclude that
\[\limsup_{k\rightarrow\infty}\left\|\tilde{J}_{h_k}\left(w^{(\varepsilon)}_{h_k}-\nabla_{h_k}^2u_{h_k}\right)\right\|_{L^2(Q_{1/2})}\le C\varepsilon\,.\]
The left-hand side here equals $\limsup_{k\rightarrow\infty}\|w-\nabla^2v_k\|_{L^2(Q_{1/2})}$, and so we obtain
\[\limsup_{k\rightarrow\infty}\|w-\nabla^2v_k\|_{L^2(Q_{1/2})}\le C\epsilon\,.\]
Since $\varepsilon$ was arbitrary, we conclude that $(\nabla^2v_k)$ converges strongly in $L^2(Q_{1/2})$ to $w$. But we already know that $(\nabla^2v_k)$ converges weakly in $L^2(Q_{13/14})$ to $\nabla^2v$, so we obtain that $\nabla^2v=w$ in $Q_{1/2}$.

\emph{Step 4: Conclusion of the argument}\\
We proved that $w_k=I^{pc}_{h_k}\nabla_{h_k}^2u_{h_k}$ converges strongly in $L^2(Q_{1/2})$ to $\nabla^2v$. Because $\rho\le\frac12$ then also $\nabla_{h_k}^2u_{h_k}-\left[\nabla_{h_k}^2u_{h_k}\right]_{Q_{\rho}}$ converges strongly in $L^2(Q_{1/2})$ to $\nabla^2v-\left[\nabla^2v\right]_{Q_{\rho}}$, and so from \eqref{e:decayfull2_4} we conclude that
\[\left\|\nabla^2v-\left[\nabla^2v\right]_{Q_{\rho}}\right\|^2_{L^2(Q_\rho)}\ge \rho^{n+1}\,.\]
In addition we know that $\left\|\nabla^2v_k\right\|_{L^2(Q_{13/14})}\le C$ and that $\nabla^2v_k$ converges weakly in $L^2(Q_{13/14})$ to $\nabla^2v$. This implies
\begin{align*}
	\left\|\nabla^2v-\left[\nabla^2v\right]_{Q_{13/14}}\right\|^2_{L^2(Q_{13/14})}&\le\left\|\nabla^2v\right\|^2_{L^2(Q_{13/14})}\\
	&\le\liminf_{k\rightarrow\infty}\left\|\nabla^2v_k\right\|^2_{L^2(Q_{13/14})}\le C\,.
\end{align*}
In summary, we have proved that there is a constant $C_1$ independent of $\rho$ such that
\begin{equation}\label{e:decayfull2_5}
\left\|\nabla^2v-\left[\nabla^2v\right]_{Q_{\rho}}\right\|^2_{L^2(Q_\rho)}\ge \frac{\rho^{n+1}}{C_1}\left\|\nabla^2v-\left[\nabla^2v\right]_{Q_{13/14}}\right\|^2_{L^2(Q_{13/14})}\,.
\end{equation}
On the other hand, $\Delta^2v=0$ in $Q_{13/14}$, and thus Lemma~\ref{l:decayfullcontinuous} implies that
\[\left\|\nabla^2v-\left[\nabla^2v\right]_{Q_\rho}\right\|^2_{L^2(Q_\rho)}\le C_2\left(\frac{\rho}{\frac{13}{14}}\right)^{n+\frac32}\left\|\nabla^2v-\left[\nabla^2v\right]_{Q_{13/14}}\right\|^2_{L^2(Q_{13/14})}\]
for a constant $C_2$ independent of $\rho$.

This is a contradiction to \eqref{e:decayfull2_5} provided that we choose $\rho$ small enough, namely $\rho<\frac{1}{C_1^2C_2^2}\left(\frac{13}{14}\right)^{2n+3}$. So we finally fix a $\rho$ satisfying this condition, and proved that falsity of the claim leads to a contradiction. 
\end{proof}

Now we can return to Lemma~\ref{l:decayfull}.

\begin{proof}[Proof of Lemma~\ref{l:decayfull}]
We can assume w.l.o.g. that $x=0$.

We claim that for any $0<s'\le s\le r$ we have
\begin{equation}\label{e:decayfull_1}
	\left\|\nabla_h^2u_h-\left[\nabla_h^2u_h\right]_{Q_{s'}}\right\|^2_{L^2(Q_{s'})}\le C\left(\frac{s'}{s}\right)^{n+1}\left\|\nabla_h^2u_h-\left[\nabla_h^2u_h\right]_{Q_s}\right\|^2_{L^2(Q_s)}\,.
\end{equation}
To prove this estimate, observe first that we can assume $s'\ge\frac h2$, as otherwise the left-hand side is 0. We can also assume $\frac{s}{s'}\ge\frac{2M}{h}$ (where $M$ is the constant from Lemma~\ref{l:decayfull2}), as otherwise we can trivially estimate
\begin{align*}
\left\|\nabla_h^2u_h-\left[\nabla_h^2u_h\right]_{Q_{s'}}\right\|^2_{L^2(Q_{s'})}&\le \left\|\nabla_h^2u_h-\left[\nabla_h^2u_h\right]_{Q_{s}}\right\|^2_{L^2(Q_{s})}\\
&\le C\left(\frac{s'}{s}\right)^{n+1}\left\|\nabla_h^2u_h-\left[\nabla_h^2u_h\right]_{Q_s}\right\|^2_{L^2(Q_s)}\,,
\end{align*}
which holds for $C\ge\left(\frac{2M}{h}\right)^{n+1}$.

So we assume $s'\ge\frac h2$ and $\frac{s}{s'}\ge\frac{2M}{h}$. Then in particular $s\ge Mh$. Consider the $\rho$ from Lemma~\ref{l:decayfull2} and let $\kappa$ be the largest integer such that $\rho^\kappa s\ge\max(s',Mh)$. We can then apply Lemma~\ref{l:decayfull2} repeatedly with radii $s,\rho s,\ldots ,\rho^\kappa s$ to find
\[\left\|\nabla_h^2u_h-\left[\nabla_h^2u_h\right]_{Q_{\rho^\kappa s}}\right\|^2_{L^2(Q_{\rho^\kappa s})}\le\rho^{\kappa(n+1)}\left\|\nabla_h^2u_h-\left[\nabla_h^2u_h\right]_{Q_s}\right\|^2_{L^2(Q_s)}\,.\]
Because $s'\le \rho^\kappa s$, we also have
\[\left\|\nabla_h^2u_h-\left[\nabla_h^2u_h\right]_{Q_{s'}}\right\|^2_{L^2(Q_{s'})}\le \left\|\nabla_h^2u_h-\left[\nabla_h^2u_h\right]_{Q_{\rho^\kappa s}}\right\|^2_{L^2(Q_{\rho^\kappa s})}\,.\]
Here we have used the fact that $\|f-[f]_\Omega\|_{L^2(\Omega)}$ is monotone in $\Omega$. If we combine the last two estimates and observe that $\rho^{\kappa+1}s<\max(s',Mh)\le2Ms'$, i.e. $\rho^\kappa\le\frac{2M}{\rho}\frac{s'}{s}$, we indeed obtain \eqref{e:decayfull_1} with $C=\left( \frac{2M}{\rho}\right)^{n+1}$.

Now using \eqref{e:decayfull_1} to prove the lemma is a standard iteration argument as e.g. in \cite[Theorem 3.1]{Giaquinta1993}. For the sake of completeness we sketch the proof.

If we apply \eqref{e:decayfull_1} with $s=r$ and $s'=\frac{r}{2^\lambda}$ or $s'=\frac{r}{2^{\lambda+1}}$, we can estimate
\begin{align*}
	&\left\|\left[\nabla_h^2u_h\right]_{Q_{r/2^{\lambda+1}}}-\left[\nabla_h^2u_h\right]_{Q_{r/2^{\lambda}}}\right\|^2_{L^2(Q_{r/2^{\lambda+1}})}\\
	&\quad\le2\left\|\nabla_h^2u_h-\left[\nabla_h^2u_h\right]_{Q_{r/2^{\lambda}}}\right\|^2_{L^2(Q_{r/2^{\lambda}})}\\
	&\quad\quad+2\left\|\nabla_h^2u_h-\left[\nabla_h^2u_h\right]_{Q_{r/2^{\lambda+1}}}\right\|^2_{L^2(Q_{r/2^{\lambda+1}})}\\
	&\quad\le\frac{C}{2^{\lambda(n+1)}}\left\|\nabla_h^2u_h-\left[\nabla_h^2u_h\right]_{Q_r}\right\|^2_{L^2(Q_r)}
\end{align*}
and hence
\[\left|\left[\nabla_h^2u_h\right]_{Q_{r/2^{\lambda+1}}}-\left[\nabla_h^2u_h\right]_{Q_{r/2^{\lambda}}}\right|\le \frac{C}{r^\frac n2 2^\frac\lambda2}\left\|\nabla_h^2u_h-\left[\nabla_h^2u_h\right]_{Q_r}\right\|_{L^2(Q_r)}\,.\]
If we sum this for $\lambda=0,1,\ldots$ and observe that for $\lambda$ small enough \linebreak$\left[\nabla_h^2u_h\right]_{Q_{r/2^{\lambda}}}=\nabla_h^2u_h(0)$ we obtain
\[\left|\nabla_h^2u_h(0)-\left[\nabla_h^2u_h\right]_{Q_r}\right|\le \frac{C}{r^\frac n2}\left\|\nabla_h^2u_h-\left[\nabla_h^2u_h\right]_{Q_r}\right\|_{L^2(Q_r)}\,.\]
Now we can estimate
\begin{align*}
	\left|\nabla_h^2u_h(0)\right|^2&\le2\left|\nabla_h^2u_h(0)-\left[\nabla_h^2u_h\right]_{Q_r}\right|^2+2\left|\left[\nabla_h^2u_h\right]_{Q_r}\right|^2\\
	&\le\frac{C}{r^n}\left(\left\|\nabla_h^2u_h-\left[\nabla_h^2u_h\right]_{Q_r}\right\|^2_{L^2(Q_r)}+\left\|\left[\nabla_h^2u_h\right]_{Q_r}\right\|^2_{L^2(Q_r)}\right)\\
	&=\frac{C}{r^n}\|\nabla_h^2u_h\|^2_{L^2(Q_r)}\,,
\end{align*}
which proves the claim. 
\end{proof}

\subsection{Half-space}
In the half-space we want to prove the following statement, which is a slightly weaker analogue of Lemma~\ref{l:decayfull}:
\begin{lemma}\label{l:decayhalf}
Let $u_h\colon(h\Z)^n\rightarrow\R$, let $x\in(h\Z)^n$, $r>0$, $\nu\in\{e_1,-e_1,\ldots ,\allowbreak e_n,-e_n\}$. Suppose that $u_h(y)=0$ for all $y\in Q^h_r(x)$ such that $(y-x)\cdot\nu\le0$, and $\Delta_h^2u_h(y)=0$ for all $y\in Q_{r-h}^h(x)$ such that $(y-x)\cdot\nu>0$. Then, for any $s\le r$,
\[\|\nabla_h^2u_h\|_{L^2(Q_s(x))}\le C\left(\frac{s}{r}\right)^\frac n2\|\nabla_h^2u_h\|_{L^2(Q_r(x))}\,.\]
\end{lemma}
The proof is mostly similar to that, Lemma~\ref{l:decayfull}, so we only give details where a new idea is required.

For $r>0$ let $Q_{r,+}=Q_r\cap\{x_1>0\}$. The main step in the proof of Lemma~\ref{l:decayhalf} will be to prove the following estimate.
\begin{lemma}\label{l:decayhalf2}
There exist constants $M\in\N$, $0<\rho<\frac12$ with the following property: Let $u_h\colon(h\Z)^n\rightarrow\R$, $r>0$ be such that $u_h(y)=0$ whenever $y\in Q^h_r$ and $y_1\le 0$, and $\Delta_h^2u_h(y)=0$ for all $y\in Q_{r-h}^h$ such that $y_1>0$. Assume that $\rho r\geq Mh$. Then we have
\begin{align*}
&\left\|\nabla_h^2u_h-\left[D^h_{-1}D^h_1u_h\right]_{Q_{\rho r,+}}e_1\otimes e_1\right\|^2_{L^2(Q_{\rho r,+})}\\
&\quad\le \rho^{n+1}\left\|\nabla_h^2u_h-\left[D^h_{-1}D^h_1u_h\right]_{Q_{r,+}}e_1\otimes e_1\right\|^2_{L^2(Q_{r,+})}\,.
\end{align*}
\end{lemma}
Using a compactness argument, we will deduce this estimate from the following continuous estimate.
\begin{lemma}\label{l:decayhalfcontinuous}
Let $0<s\le \frac r2$, $u\in W^{2,2}(Q_{r,+})$. Assume that $\Delta^2u=0$ weakly in $Q_{r,+}$ and that $u=0$, $D_1u=0$ on $\partial Q_{r,+}\cap\{x_1=0\}$ in the sense of traces. Then we have
\begin{align*}
	&\left\|\nabla^2u-\left[D_1^2u\right]_{Q_{s,+}}e_1\otimes e_1\right\|^2_{L^2(Q_{s,+})}\\
	&\quad\le C\left(\frac{s}{r}\right)^{n+\frac32}\left\|\nabla^2u-\left[D_1^2u\right]_{Q_{r,+}}e_1\otimes e_1\right\|^2_{L^2(Q_{r,+})}\,.
\end{align*}
\end{lemma}
\begin{proof}
This follows like Lemma~\ref{l:decayfullcontinuous} from Schauder estimates up to the boundary (cf. \cite[Theorem 6.4.8]{Morrey1966}). 
\end{proof}
\begin{proof}[Proof of Lemma~\ref{l:decayhalf2}]$ $\\
\emph{Step 1: Preparations}\\
We follow the same strategy as in the proof of Lemma~\ref{l:decayfull2}. That is, we assume that the claim is wrong for some fixed $\rho$, and consider a sequence of counterexamples $u_{h_k}$ and their interpolations $v_k=I_{h_k}u_{h_k}$. We can assume that $r_k=1$.

Next observe that for $\omega_h(x):=\begin{cases}\frac{x_1(x_1+h)}{2}&x_1\ge0\\0&x_1<0\end{cases}$ we have $\omega_h(x)=0$ if $x_1\leq0$ and $D^h_{-1}D^h_1\omega_h(x)=\begin{cases}1&x_1\ge0\\0&x_1<0\end{cases}$. So by replacing $u_h$ with $u_h-\left[D^h_{-1}D^h_1u_h\right]_{Q_{1,+}}\omega_h$ we can also assume $\left[D^h_{-1}D^h_1u_h\right]_{Q_{1,+}}=0$. 
Having normalized $u_h$ on $Q_{1,+}$ in this way, we now consider $Q_1$ again. We can assume 
\begin{equation}
	\left\|\nabla_{h_k}^2u_{h_k}\right\|_{L^2(Q_1)}=1\,.\label{e:decayhalf2_3}
\end{equation}
Note that
\[\left\|\nabla_{h_k}^2u_{h_k}\right\|^2_{L^2(Q_1)}=\left\|\nabla_{h_k}^2u_{h_k}\right\|^2_{L^2(Q_{1,+})}+\left\|\nabla_{h_k}^2u_{h_k}\right\|^2_{L^2((-h/2,0)\times(-1,1)^{n-1})}\]
and
\[\left\|\nabla_{h_k}^2u_{h_k}\right\|^2_{L^2\left((-h/2,0)\times(-1,1)^{n-1}\right)}=\left\|\nabla_{h_k}^2u_{h_k}\right\|^2_{L^2((0,h/2)\times(-1,1)^{n-1})}\,.\]
Now \eqref{e:decayhalf2_3} implies that $\left\|\nabla_{h_k}^2u_{h_k}\right\|_{L^2(Q_{1,+})}\ge\frac12$, so that
\begin{equation}\label{e:decayhalf2_2}
	\left\|\nabla_{h_k}^2u_{h_k}-\left[D^h_{-1}D^h_1u_{h_k}\right]_{Q_{\rho}}e_1\otimes e_1\right\|^2_{L^2(Q_{\rho})}>\frac{\rho^{n+1}}{2}\,.
\end{equation}

By \eqref{e:decayhalf2_3}, Proposition~\ref{p:interpolationproperties} and the Poincaré inequality with 0 boundary values $(v_k)$ is bounded in $W^{2,2}(Q_{3/4})$, and so a non-relabeled subsequence converges weakly to some $v$ in $W^{2,2}(Q_{3/4})$.

As in step 2 of the proof of Lemma~\ref{l:decayfull2} we can show that $\Delta^2v=0$ weakly in $Q_{3/4,+}$.
We have $u_{h_k}=0$ in $Q_1\cap\{x_1<0\}$ and hence $v_k=0$ in $Q_{3/4}\cap\{x_1<-3h_k\}$. Since $v_k$ converges to $v$ strongly in $L^2(Q_{3/4})$, $v=0$ in $\{x_1<0\}$, and because $v\in W^{2,2}(Q_{3/4})$, we obtain that $v=0$ and $D_1 v=0$ on  $Q_{3/4}\cap\{x_1=0\}$ in the sense of traces.

We define $w_k=I^{pc}_{h_k}\nabla_{h_k}^2u_{h_k}$ and want to show next that $\nabla^2v_k$ and $w_k$ converge to $\nabla^2v$ strongly in $L^2(Q_{1/2})$. We cannot directly reuse the argument in Step 3 of the proof of Lemma~\ref{l:decayfull2}, as we now have to deal with boundary values. However, we can use that argument on any cube $Q_{\tilde{r}}(\tilde{x})\subset Q_{5/8}\cap\{x_1>0\}$ to conclude that $\nabla^2v_k$ and $w_k$ converge to $\nabla^2v$ strongly in $L^2(Q_{\tilde{r}/2})$. Since we can do this for any such cube, we conclude that $\nabla^2v_k$ and $w_k$ converge to $\nabla^2v$ strongly in $L^2_{loc}(Q_{5/8,+})$.

Because $u_{h_k}=0$ in $Q_{5/8}\cap\{x_1<-3h_k\}$, we also have that $\nabla^2v_k$ and $w_k$ converge to 0 strongly in $L^2_{loc}(Q_{5/8}\cap\{x_1<0\})$. In summary, we have proved that  $\nabla^2v_k$ and $w_k$ converge to $\nabla^2v$ strongly in $L^2_{loc}(Q_{5/8}\setminus\{x_1=0\})$.

We still have to deal with $\{x_1=0\}$, and for this we need a new idea.

\emph{Step 2: Nonconcentration at the boundary}\\
We claim that for any $y\in Q_{1/2}\cap\{x_1=0\}$ we have
\begin{equation}\label{e:decayhalf2_1}
	\lim_{\tilde{r}\rightarrow0}\limsup_{k\rightarrow\infty}\left\|\nabla_{h_k}^2u_{h_k}\right\|_{L^2(Q_{\tilde{r}}(y))}=0\,.
\end{equation}
To see this, let $\tilde{r}>0$. For $h_k$ small enough Lemma~\ref{l:Cacc} and Proposition~\ref{p:interpolationproperties} imply that
\begin{align*}
	\left\|\nabla_{h_k}^2u_{h_k}\right\|^2_{L^2(Q_{\tilde{r}}(y))}&\le\frac{C}{\tilde{r}^2}\left\|\nabla_{h_k}u_{h_k}\right\|^2_{L^2(Q_{2\tilde{r}}(y))}+\frac{C}{\tilde{r}^4}\left\|u_{h_k}\right\|^2_{L^2(Q_{2\tilde{r}}(y))}\\
	&\le\frac{C}{\tilde{r}^2}\left\|\nabla v_k\right\|^2_{L^2(Q_{4\tilde{r}}(y))}+\frac{C}{\tilde{r}^4}\left\|v_k\right\|^2_{L^2(Q_{4\tilde{r}}(y))}\,.
\end{align*}
Now $v_k$ converges to $v$ weakly in $W^{2,2}(Q_{3/4})$, so $v_k$ and $\nabla v_k$ converge strongly in $L^2(Q_{3/4})$. Hence we can pass to the limit in the above inequality and find
\[\limsup_{k\rightarrow\infty}\left\|\nabla_{h_k}^2u_{h_k}\right\|^2_{L^2(Q_{\tilde{r}}(y))}\le\frac{C}{\tilde{r}^2}\left\|\nabla v\right\|^2_{L^2(Q_{4\tilde{r}}(y))}+\frac{C}{\tilde{r}^4}\left\|v\right\|^2_{L^2(Q_{4\tilde{r}}(y))}\,.\]
Furthermore $v$ is 0 in $Q_{4\tilde{r}}(y)\cap\{x_1<0\}$, so we can apply the Poincaré inequality to conclude
\[\limsup_{k\rightarrow\infty}\left\|\nabla_{h_k}^2u_{h_k}\right\|^2_{L^2(Q_{\tilde{r}}(y))}\le C\left\|\nabla^2v\right\|^2_{L^2(Q_{4\tilde{r}}(y))}\,.\]
Now $\nabla^2v$ is a fixed $L^2$-function, so if we pass to the limit $\tilde{r}\rightarrow0$ here, we indeed obtain \eqref{e:decayhalf2_1}.

It is easy to see that \eqref{e:decayhalf2_1} together with the fact that $w_k=I_{h_k}^{pc}\nabla_h^2u_{h_k}$ converges to $\nabla^2v$ strongly in $L^2_{loc}(Q_{5/8}\setminus\{x_1=0\})$ imply that $w_k$ actually converges to $\nabla^2v$ strongly in $L^2(Q_{1/2})$.

We have 
\[\limsup_{k\rightarrow\infty}\left\|\nabla^2v_k\right\|_{L^2(Q_{\tilde{r}}(y))}\le C\limsup_{k\rightarrow\infty}\left\|\nabla_{h_k}^2u_{h_k}\right\|_{L^2(Q_{2\tilde{r}}(y))}\]
and so from \eqref{e:decayhalf2_1} we also conclude
\[\lim_{\tilde{r}\rightarrow0}\limsup_{k\rightarrow\infty}\left\|\nabla^2v_k\right\|_{L^2(Q_{\tilde{r}}(y))}=0\,.\]
This in turn implies that also $\nabla^2v_k$ converges to $\nabla^2v$ strongly in $L^2(Q_{1/2})$.

\emph{Step 3: Conclusion of the argument}\\
We can now continue as in Step 4 of the proof of Lemma~\ref{l:decayfull2}: The strong convergence of $w_k$ to $\nabla^2v$ allows us to conclude from \eqref{e:decayhalf2_2} that
\[\left\|\nabla^2v-\left[D_1^2v\right]_{Q_{\rho,+}}e_1\otimes e_1\right\|^2_{L^2(Q_{\rho,+})}\ge\frac{\rho^{n+1}}{2}\,.\]
On the other hand, we have
\[\left\|\nabla^2v-\left[D_1^2v\right]_{Q_{3/4,+}}e_1\otimes e_1\right\|^2_{L^2(Q_{3/4,+})}\le C\]
and it is easy to check that we arrive at a contradiction to Lemma~\ref{l:decayhalfcontinuous} once we choose $\rho$ small enough. 
\end{proof}

\begin{proof}[Proof of Lemma~\ref{l:decayhalf}]
The proof is similar to the first half of the proof of Lemma~\ref{l:decayfull}: One can assume that $x=0$, $\nu=e_1$. Then one first proves that, for any $0<s'\le s\le r$,
\begin{align*}
	&\left\|\nabla_h^2u_h-\left[D^h_{-1}D^h_1u_h\right]_{Q_{s',+}}e_1\otimes e_1\right\|^2_{L^2(Q_{s',+})}\nonumber\\
	&\quad\le C\left(\frac{s'}{s}\right)^{n+1}\left\|\nabla_h^2u_h-\left[D^h_{-1}D^h_1u_h\right]_{Q_{s,+}}e_1\otimes e_1\right\|^2_{L^2(Q_{s,+})}\,,
\end{align*}
which already looks similar to the claimed estimate. We can again use this with $s=r$ and $s'=\frac{r}{2^\lambda}$ or $s'=\frac{r}{2^{\lambda+1}}$ to conclude
\begin{align*}&\left|\left[D^h_{-1}D^h_1u_h\right]_{Q_{r/2^{\lambda+1},+}}-\left[D^h_{-1}D^h_1u_h\right]_{Q_{r/2^{\lambda},+}}\right|\\
&\quad\le \frac{C}{r^\frac n2 2^\frac\lambda2}\left\|\nabla_h^2u_h-\left[D^h_{-1}D^h_1u_h\right]_{Q_{r,+}}e_1\otimes e_1\right\|_{L^2(Q_{r,+})}\,.
\end{align*}
Let $\lambda_0$ be the largest integer such that $\frac{r}{2^{\lambda_0}}\ge s$. We can apply this estimate with radii $r,\frac r2,\ldots,\frac{r}{2^{\lambda_0-1}}$ and sum to conclude
\begin{align*}
	&\left|\left[D^h_{-1}D^h_1u_h\right]_{Q_{r/2^{\lambda_0},+}}-\left[D^h_{-1}D^h_1u_h\right]_{Q_{r,+}}\right|\\
	&\quad\le \frac{C}{r^\frac n2}\left\|\nabla_h^2u_h-\left[D^h_{-1}D^h_1u_h\right]_{Q_{r,+}}e_1\otimes e_1\right\|_{L^2(Q_{r,+})}\,.
\end{align*}
Using all this, we can estimate
\begin{align*}
	&\left\|\nabla_h^2u_h\right\|^2_{L^2(Q_{s,+})}\le\left\|\nabla_h^2u_h\right\|^2_{L^2(Q_{r/2^{\lambda_0},+})}\\
	&\quad\le	2\left\|\nabla_h^2u_h-\left[D^h_{-1}D^h_1u_h\right]_{Q_{r/2^{\lambda_0},+}}e_1\otimes e_1\right\|^2_{L^2(Q_{r/2^{\lambda_0},+})}\\
	&\quad\quad+2\left\|\left[D^h_{-1}D^h_1u_h\right]_{Q_{r/2^{\lambda_0},+}}\right\|^2_{L^2(Q_{r/2^{\lambda_0},+})}\\
	&\quad\le\left(\frac{C}{2^{\lambda_0(n+1)}}+\frac{C}{2^{\lambda_0n}}\right)\left\|\nabla_h^2u_h-\left[D^h_{-1}D^h_1u_h\right]_{Q_{r,+}}e_1\otimes e_1\right\|^2_{L^2(Q_{r,+})}\\
	&\quad\quad+\frac{C}{2^{\lambda_0n}}\left\|\left[D^h_{-1}D^h_1u_h\right]_{Q_{r,+}}\right\|^2_{L^2(Q_{r,+})}\\
	&\quad\leq \frac{C}{2^{\lambda_0n}}\left\|\nabla_h^2u_h\right\|^2_{L^2(Q_{r,+})}\,,
\end{align*}
which implies
\begin{align}
	\left\|\nabla_h^2u_h\right\|^2_{L^2(Q_{s,+})}&\le C\left(\frac{r}{s}\right)^n\left\|\nabla_h^2u_h\right\|^2_{L^2(Q_{r,+})}\nonumber\\
	&\le C\left(\frac{r}{s}\right)^n\left\|\nabla_h^2u_h\right\|^2_{L^2(Q_r)}\,.\label{e:decayhalf_1}
\end{align}
Now by the same argument as in Step 1 of the proof of Lemma~\ref{l:decayhalf2} we have
\[\left\|\nabla_h^2u_h\right\|^2_{L^2(Q_s)}\le 2\left\|\nabla_h^2u_h\right\|^2_{L^2(Q_{s+})}\,.\]
Combining this with \eqref{e:decayhalf_1} yields the result. 
\end{proof}

\subsection{Edges and vertices}\label{s:edgesvertices}
It remains to prove the analogue of Lemma~\ref{l:decayhalf} near edges (in 3D) and vertices (in 2D and 3D). The actual compactness argument requires no new idea, so we will only give a very brief sketch of the proofs. However, this time the continuous estimate require a bit more work, so we will go into detail there.
Let us first state the two results:
\begin{lemma}\label{l:decayedge}
Let $u_h\colon(h\Z)^3\rightarrow\R$, let $x\in(h\Z)^3$, $r>0$, $\nu_1,\nu_2\in\{e_1,-e_1,\ldots ,\allowbreak e_n,-e_n\}$ such that $\nu_1\neq\pm\nu_2$. Suppose that $u_h(y)=0$ for all $y\in Q^h_r(x)$ such that $(y-x)\cdot\nu_1\le0$ or $(y-x)\cdot\nu_2\le0$, and $\Delta_h^2u_h(y)=0$ for all $y\in Q_{r-h}^h(x)$ such that $(y-x)\cdot\nu_1>0$ and $(y-x)\cdot\nu_2>0$. Then, for any $s\le r$,
\[\|\nabla_h^2u_h\|_{L^2(Q_s(x))}\le  C\left(\frac{s}{r}\right)^\frac 32\|\nabla_h^2u_h\|_{L^2(Q_r(x))}\,.\]
\end{lemma}
\begin{lemma}\label{l:decayvertex}
Let $n=2$ or $n=3$, $u_h\colon(h\Z)^n\rightarrow\R$, let $x\in(h\Z)^n$, $r>0$, $\nu_i\in\{e_i,-e_i\}$ for $i\in\{1,\ldots ,n\}$. Suppose that $u_h(y)=0$ for all $y\in Q^h_r(x)$ such that $(y-x)\cdot\nu_i\le0$ for at least one $i$, and $\Delta_h^2u_h(y)=0$ for all $y\in Q_{r-h}^h(x)$ such that $(y-x)\cdot\nu_i>0$ for all $i$. Then, for any $s\le r$,
\[\|\nabla_h^2u_h\|_{L^2(Q_s(x))}\le  C\left(\frac{s}{r}\right)^\frac n2\|\nabla_h^2u_h\|_{L^2(Q_r(x))}\,.\]
\end{lemma}
\begin{proof}[Proof of Lemma~\ref{l:decayedge} and Lemma~\ref{l:decayvertex}]
This follows easily from the following two lemmata. 
\end{proof}
\begin{lemma}\label{l:decayedge2}
There are constants $M\in\N$, $0<\rho<\frac12$ with the following property: let\\ $u_h\colon(h\Z)^3\rightarrow\R$, $r>0$, such that $u_h(y)=0$ for all $y\in Q^h_r$ such that $y_1\le 0$ or $y_2\le0$, and $\Delta_h^2u_h(y)=0$ for all $y\in Q_{r-h}^h(x)$ such that $y_1>0$ and $y_2>0$. Then we have that
\[\left\|\nabla_h^2u_h\right\|^2_{L^2(Q_{\rho r})}\le \rho^n\left\|\nabla_h^2u_h\right\|^2_{L^2(Q_r)}\,.\]
\end{lemma}
\begin{lemma}\label{l:decayvertex2}
There are constants $M\in\N$, $0<\rho<\frac12$ with the following property: let $n=2$ or $n=3$, $u_h\colon(h\Z)^n\rightarrow\R$, $r>0$, such that $u_h(y)=0$ for all $y\in Q^h_r$ such that $y_i\le 0$ for at least one $i\in\{1,\ldots ,n\}$, and $\Delta_h^2u_h(y)=0$ for all $y\in Q_{r-h}^h$ such that $y_i>0$ for all $i$. Assume that $\rho r\geq Mh$. Then we have that
\[\left\|\nabla_h^2u_h\right\|^2_{L^2(Q_{\rho r})}\le \rho^n\left\|\nabla_h^2u_h\right\|^2_{L^2(Q_r)}\,.\]
\end{lemma}
We will deduce these two lemmata from the following continuous estimates. $D_\nu$ denotes the derivative in normal direction.
\begin{lemma}\label{l:decayedgecontinuous}
There is a constant $\theta>0$ with the following property: let $n=3$, $0<s\le \frac r2$, $u\in W^{2,2}(Q_{r,++})$, where $Q_{r,++}=Q_r\cap\{x_1>0,x_2>0\}$. Assume that $\Delta^2u=0$ weakly in $Q_{r,++}$ and that $u=0$, $D_\nu u=0$ on $\partial Q_{r,++}\cap\{x_1=0\lor x_2=0\}$ in the sense of traces. Assume that $\rho r\geq Mh$. Then we have
\[\left\|\nabla^2u\right\|^2_{L^2(Q_{s,++})}\le C\left(\frac{s}{r}\right)^{3+\theta}\left\|\nabla^2u\right\|^2_{L^2(Q_{r,++})}\,.\]
\end{lemma}
\begin{lemma}\label{l:decayvertexcontinuous}
There is a constant $\theta>0$ with the following property: let $0<s\le \frac r2$, $u\in W^{2,2}(Q_{r,n+})$, where $Q_{r,n+}=Q_{r,++}=Q_r\cap\{x_1>0,x_2>0\}$ if $n=2$, and $Q_{r,n+}=Q_{r,+++}=Q_r\cap\{x_1>0,x_2>0,x_3>0\}$ if $n=3$. Assume that $\Delta^2u=0$ weakly in $Q_{r,n+}$ and that $u=0$, $D_\nu u=0$ on $\partial Q_{r,n+}\cap\{x_i=0\text{ for some }i\}$ in the sense of traces. Then we have
\[\left\|\nabla^2u\right\|^2_{L^2(Q_{s,n+})}\le C\left(\frac{s}{r}\right)^{n+\theta}\left\|\nabla^2u\right\|^2_{L^2(Q_{r,n+})}\,.\]
\end{lemma}
The proof of Lemma~\ref{l:decayedgecontinuous} and Lemma~\ref{l:decayvertexcontinuous} relies heavily on the theory of elliptic equations in domains with singularities. We use results from \cite{Kozlov1997} and \cite{Mazya2010} and refer the reader to these monographs for more background information.
\begin{proof}[Proof of Lemma~\ref{l:decayedgecontinuous}]
Let $\R^3_{++}=\R^3\cap\{x_1>0,x_2>0\}$. For $x\in\R^3_{++}$ write $x=(x',x_3)$.

The statement is trivial if $s\ge\frac r4$, so assume $s<\frac r4$. Let $\eta\in C_c^\infty(Q_r)$ be a cut-off function that is 1 on $Q_{r/2,++}$ and such that $|\nabla^\kappa\eta|\le\frac{C}{r^\kappa}$ for $\kappa\le 4$. Then $\eta\Delta^2u=0$ in $\partial\R^3_{++}$, and we can calculate (as an identity in the sense of distributions) that
\[\Delta^2(\eta u)=(\Delta^2\eta)u+4\nabla\Delta\eta\cdot\nabla u+2\Delta\eta\Delta u+4\nabla^2\eta:\nabla^2u+4\nabla\eta\cdot\nabla\Delta u\,.\]
In order to avoid terms with too many derivatives of $u$ we rewrite the last term as
\[\nabla\eta\cdot\nabla\Delta u=\diverg(\nabla\eta\Delta u)-\Delta\eta\Delta u\]
to obtain
\[\Delta^2(\eta u)=\Delta^2\eta u+4\nabla\Delta\eta\cdot\nabla u-2\Delta\eta\Delta u+4\nabla^2\eta:\nabla^2u+4\diverg(\nabla\eta\Delta u)=:f\,.\]
Because $u\in W^{2,2}(Q_{r,n+})$ with zero boundary values on $\partial\R^3_{++}$, the right-hand side $f$ is an element of $W^{-2,2}(\R^3_{++})$, while $\eta u$ is in $W^{2,2}_0(\R^3_{++})$. Hence (cf. \cite{Mazya2010}, Theorem 2.5.1) we can represent $\eta u$ via the Green's function of $\R^3_{++}$ as
\[(\eta u)(x)=\int_{\R^3_{++}}G(x,\xi)f(\xi)\ud \xi\,.\]
For $x\in Q_{s,++}\subset Q_{r/4,++}$ this implies
\[\nabla^2u(x)=\int_{\R^3_{++}}\nabla_x^2G(x,\xi)f(\xi)\ud \xi\,.\]

Now $f$ is supported in $Q_{r,++}\setminus Q_{r/2,++}$, whereas $x\in Q_{s,++}\subset Q_{r/4,++}$. So a decay estimate for $G$ will directly lead to a pointwise estimate for $\nabla^2u$.

In fact, Theorem 2.5.4 in \cite{Mazya2010} states that if $|x-\xi|\ge \min(|x'|,|\xi'|)$ we have, for every $\varepsilon>0$,
\begin{equation}\label{e:green_edge2}
	|D^\alpha_{x'}D^j_{x_3}D^\beta_{\xi'}D^k_{\xi_3}G(x,\xi)|\le C_\varepsilon\frac{|x'|^{1+\delta_+-|\alpha|-\varepsilon}|\xi'|^{1+\delta_--|\beta|-\varepsilon}}{|x-\xi|^{1+\delta_++\delta_-+j+k-2\varepsilon}}\,.
\end{equation}
Here $\delta_+$ and $\delta_-$ are certain real parameters defined in terms of eigenvalue problems related to the bilaplacian (see \cite[Section 2.4]{Mazya2010} for the precise definition). According to \cite[Section 4.3]{Mazya2010} we have that $\delta_+=\delta_-\approx2.73959$. In particular, $\delta_\pm>1$, so we can choose $\theta>0$ such that $1+\frac\theta2<\delta_\pm$. Then let $\varepsilon=\delta_\pm-1-\frac\theta2>0$.

We are interested in the case where $x\in Q_{s,++}$, $\xi\in Q_{r,++}\setminus Q_{r/2,++}$. In that case the inequality $|x-\xi|\ge \min(|x'|,|\xi'|)$ certainly holds, and we can estimate $|x'|\le s$, $|\xi'|\le r$, $|x-\xi|\ge\frac r4$, so that \eqref{e:green_edge2} turns into
\begin{align*}
	|D^\alpha_{x'}D^j_{x_3}D^\beta_{\xi'}D^k_{\xi_3}G(x,\xi)|&\le C_\varepsilon s^{1+\delta_+-|\alpha|-\varepsilon}r^{\varepsilon-\delta_+-|\beta|-j-k}\\
	&=C\frac{s^{2+\frac\theta2-|\alpha|}}{r^{1+\frac\theta2+|\beta|+j+k}}\,.
\end{align*}
This estimate is sharp enough to allow us to estimate the terms of $f$. For example we can calculate using the Poincaré and Hölder inequality that
\begin{align*}
	\left|\int_{\R^3_{++}}\nabla_x^2G(x,\xi)\Delta^2\eta(\xi)u(\xi)\ud \xi\right|&\le C\int_{Q_{r,++}}\frac{s^{2+\frac\theta2-2}}{r^{1+\frac\theta2+0+0+0}}\frac{1}{r^4}\left|u(\xi)\right|\ud \xi\\
	&= C\frac{s^\frac\theta2}{r^{5+\frac\theta2}}\int_{Q_{r,++}}\left|u\right|\ud \xi\\
	&\le C\frac{s^\frac\theta2}{r^{5+\frac\theta2}}r^2r^\frac32\left(\int_{Q_{r,++}}\left|\nabla^2u\right|^2\ud \xi\right)^\frac12\\
	&\le C\frac{s^\frac\theta2}{r^{\frac32+\frac\theta2}}\left(\int_{\R^3_{++}}\left|\nabla^2u\right|^2\ud \xi\right)^\frac12
\end{align*}
and that
\begin{align*}
	\left|\int_{\R^3_{++}}\nabla_x^2G(x,\xi)\diverg(\nabla\eta\Delta u)(\xi)\ud \xi\right|&=\left|\int_{\R^3_{++}}\nabla_x^2\nabla_{\xi}G(x,\xi)\cdot\nabla\eta(\xi)\Delta u(\xi)\ud \xi\right|\\
	&\le C\int_{Q_{r,++}}\frac{s^\frac\theta2}{r^{2+\frac\theta2}}\frac{1}{r}\left|\Delta u(\xi)\right|\ud \xi\\
	&\le C\frac{s^\frac\theta2}{r^{\frac32+\frac\theta2}}\left(\int_{\R^3_{++}}\left|\nabla^2u\right|^2\ud \xi\right)^\frac12\,.
\end{align*}
We can estimate the other terms on $f$ analogously. If we integrate the sum of the squares of all these inequalities with respect to $x$ we immediately obtain the conclusion. 
\end{proof}
\begin{proof}[Proof of Lemma~\ref{l:decayvertexcontinuous}]
The proof in the case of a vertex is very similar. One can again deduce the representation
\begin{equation}
	\nabla^2u(x)=\int_{\R^n_{n+}}\nabla_x^2G(x,\xi)f(\xi)\ud \xi\label{e:decayvertexcontinuous1}
\end{equation}
for $x\in Q_{r/4,n+}$, so that one only needs sharp estimates for the Green's function to complete the argument.

If $n=2$, we can use for this purpose Theorem 8.4.8 in combination with Theorem 6.1.2 in \cite{Kozlov1997}. Theorem 8.4.8 gives a Green's function for right-hand sides in $L^2$. However, according to Theorem 6.1.2, the solution operator has a continuous extension to right-hand sides in $W^{-2,2}$, so that \eqref{e:decayvertexcontinuous1} holds for this Green's function. Now Theorem 8.4.8 also gives asymptotics for $G$ in terms of the eigenvalues of a certain eigenvalue problem. If we stay in the eigenvalue-free strip, this estimate reads
\[|D_x^\alpha D_\xi^\beta G(x,\xi)|\le C_\varepsilon |x|^{1+\delta_+-|\alpha|-\varepsilon}|\xi|^{1-\delta_+-|\beta|+\varepsilon}\]
where $2|x|\le|\xi|$ and $\varepsilon>0$ is arbitrary. Using this estimate we can continue as in the proof of Lemma~\ref{l:decayedgecontinuous}.

The case $n=3$ is slightly more complicated. We can use \cite[Theorem 3.4.5]{Mazya2010}, which states that if $2|x|\le|\xi|$, then for any $\varepsilon>0$
\begin{align*}
&|D_x^\alpha D_\xi^\beta G(x,\xi)|\\
&\le C_\varepsilon |x|^{\Lambda_+-|\alpha|-\varepsilon}|\xi|^{1-\Lambda_+-|\beta|+\varepsilon}\prod_{j=1}^3\left(\frac{r_j(x)}{|x|}\right)^{1+\delta_+-|\alpha|-\varepsilon}\prod_{k=1}^3\left(\frac{r_k(\xi)}{|\xi|}\right)^{1+\delta_--|\beta|-\varepsilon}
\end{align*}
where $\delta_\pm$ are as before, $\Lambda_+$ is another constant defined in terms of a certain eigenvalue problem (see \cite[Section 3.4]{Mazya2010} for the precise definition) and $r_j(x)$ denotes the distance of $x$ to the line $\{x_j=0\}$. If we choose $\varepsilon\le\delta_+-1=\delta_--1$, then the exponents of the terms $\frac{r_k(x)}{|x|}$ and $\frac{r_k(\xi)}{|\xi|}$ are non-negative whenever $|\alpha|\le2$ and $|\beta|\le2$. So we obtain under these assumptions
\[|D_x^\alpha D_\xi^\beta G(x,\xi)|\le C_\varepsilon |x|^{\Lambda_+-|\alpha|-\varepsilon}|\xi|^{1-\Lambda_+-|\beta|+\varepsilon}\,.\]

In \cite[Section 4.3]{Mazya2010} it is proved that $\Lambda_+\ge3$. This allows us to take $\theta>0$ such that $2+\frac\theta2\le\Lambda_+$ and $1+\frac\theta2<\delta_\pm$. By choosing $\epsilon=\min\left(\Lambda_+-2-\frac\theta2,\delta_\pm-1\right)$ we conclude
\[|D_x^\alpha D_\xi^\beta G(x,\xi)|\le C\frac{s^{2+\frac\theta2-|\alpha|}}{r^{1+\frac\theta2+|\beta|}}\]
for $|\alpha|\le2$ and $|\beta|\le2$. Now we can continue as in the proof of Lemma  \ref{l:decayedgecontinuous} (observe that in that proof we only needed estimates for $D_x^\alpha D_\xi^\beta G(x,\xi)$ with $|\alpha|\le2$ and $|\beta|\le1$). 
\end{proof}
\begin{proof}[Proof of Lemma~\ref{l:decayedge2} and Lemma~\ref{l:decayvertex2}]
We follow the proofs of Lemma~\ref{l:decayfull2} and Lemma~\ref{l:decayhalf2}. The proof is slightly easier than the proof of Lemma~\ref{l:decayhalf2} because we no longer need to worry about the subtraction of the averages of $u_h$. We assume that the claim is wrong for some fixed $\rho$, and consider a sequence of counterexamples $u_{h_k}$ and their interpolations $v_k=I_{h_k}u_{h_k}$.
We can assume that $r_k=1$, and $\left\|\nabla_{h_k}^2u_{h_k}\right\|_{L^2(Q_1)}=1$ and conclude that $(v_k)$ is bounded in $W^{2,2}(Q_{3/4})$, and so a non-relabeled subsequence converges to some $v$ in $W^{2,2}(Q_{3/4})$.

As before we see that $\Delta^2v=0$ in $Q_{3/4,+++}$ and $Q_{3/4,n+}$ respectively and that $v$ has 0 boundary values. Also we obtain strong convergence of $\nabla^2v_k$ and $w_k:=I^{h_k}_{pc}\nabla_{h_k}^2u_{h_k}$ in $L^2_{loc}(Q_{5/8}\setminus\partial Q_{3/4,+++})$ and $L^2_{loc}(Q_{5/8}\setminus\partial Q_{3/4,n+})$, respectively. Now, as in Step 2 of the proof of Lemma~\ref{l:decayhalf2}, we find that $\nabla_{h_k}^2u_{h_k}$ does not concentrate at the boundary, so that $\nabla^2v_k$ and $w_k$ actually converge strongly in $L^2(Q_{1/2})$.

This convergence allows us to pass to the limit in 
\[\left\|\nabla_{h_k}^2u_{h_k}\right\|^2_{L^2(Q_{\rho})}>\rho^n\]
so that we easily arrive at a contradiction to Lemma~\ref{l:decayedgecontinuous} or Lemma~\ref{l:decayvertexcontinuous} once we choose $\rho$ small enough. 
\end{proof}

\section{Inner and outer decay estimates for discrete biharmonic functions}   \label{s:general_decay}
\subsection{Inner estimates}\label{s:inner_decay}
We can now combine the results from the previous section in one general decay estimate for biharmonic functions:
\begin{theorem}\label{t:decaycube}
Let $u_h\in\Phi_h$. Let $x\in\Lambda_h^n$, $r>0$ and suppose that $\Delta_h^2u_h(y)=0$ for all $y\in Q_{r-h}(x)\cap \inte\Lambda^n_h$. Then, for all $z\in Q^h_{r/2}(x)\cap\Lambda^n_h$,
\begin{equation}\label{e:decaycube}
	|\nabla_h^2u_h(z)|\le\frac{C}{r^\frac n2}\|\nabla_h^2u_h\|_{L^2(Q_r(x))}\,.
\end{equation}
\end{theorem}
Observe that $\nabla_h^2u_h=0$ is zero in $(h\Z)^n\setminus\Lambda_h^n$. Therefore we could equivalently only integrate over $Q_r(x)\cap(\Lambda_h^n)_{pc}$ on the right-hand side.
\begin{proof}
The proofs for the cases $n=2$ and $n=3$ are similar, but the latter is somewhat more tedious. Therefore we give the proof for $n=2$ in detail and then describe how to adapt it to the case $n=3$. So let $n=2$.

We first prove the statement in the special case $z=x$. By rotating and reflecting $\Lambda^2_h$ we may assume $x_2\le x_1\le\frac12$. We may also assume $r\ge\frac h2$, as otherwise we can replace $r$ by $\frac h2$ without changing \eqref{e:decaycube}.

Let $x^*=(x_1,0)$ be a point on $\partial\Lambda_h^2$ closest to $x$. We consider the three cases $r\le x_2$, $x_2<r\le x_1$ and $r>x_1$.

\emph{Case 1: $r\le x_2$}\\
In this case the interior estimate Lemma~\ref{l:decayfull} applied to $Q_r(x)$ directly implies
\[|\nabla_h^2u_h(x)|\le\frac{C}{r}\|\nabla_h^2u_h\|_{L^2(Q_r(x))}\,.\]

\emph{Case 2: $x_2<r\le x_1$}\\
Apply first Lemma~\ref{l:decayfull} to $Q_{x_2+h/2}(x)$ to find
\[|\nabla_h^2u_h(x)|\le\frac{C}{x_2+\frac h2}\|\nabla_h^2u_h\|_{L^2(Q_{x_2+h/2}(x))}\,.\]
If $r<3x_2$ then this already implies \eqref{e:decaycube} once we increase $C$ by a factor of 3. If $r\ge3x_2$ we have $Q_{x_2+h/2}(x)\subset Q_{2x_2+h/2}(x^*)\subset Q_r(x^*)\subset Q_r(x)$ and so, by Lemma~\ref{l:decayhalf},
\[|\nabla_h^2u_h\|_{L^2(Q_{x_2+h/2}(x))}\le\|\nabla_h^2u_h\|_{L^2(Q_{2x_2+h/2}(x^*))}\le C\frac{2x_2+\frac h2}{r}\|\nabla_h^2u_h\|_{L^2(Q_r(x^*))}\,.\]
This together with the previous equation implies \eqref{e:decaycube}.

\emph{Case 3: $x_1<r$}\\
As in the previous case we obtain
\begin{equation}
	|\nabla_h^2u_h(x)|\le\frac{C}{x_1+\frac h2}\|\nabla_h^2u_h\|_{L^2(Q_{x_1+h/2}(x^*))}\,.\label{e:decaycube2}
\end{equation}

Now either $r<3x_1$ and we are done, or we can continue with Lemma~\ref{l:decayvertex} to find
\[\|\nabla_h^2u_h\|_{L^2(Q_{x_1+h/2}(x^*))}\le \|\nabla_h^2u_h\|_{L^2(Q_{2x_1+h/2}(0))}\le C\frac{2x_1+\frac h2}{r}\|\nabla_h^2u_h\|_{L^2(Q_r(0))}\,,\]
which in combination with \eqref{e:decaycube2} implies \eqref{e:decaycube}.

This proves \eqref{e:decaycube} in the case $z=x$. For general $z$, it suffices to observe that $Q_{r/2}(z)\subset Q_r(x)$ and apply the statement we have just proved to $Q_{r/2}(z)$.

The proof for $n=3$ is analogous. However there is one more case and hence we need one more intermediate step, where we deal with the case of an edge. So one applies Lemmata \ref{l:decayfull}, \ref{l:decayhalf}, \ref{l:decayedge}, \ref{l:decayvertex} in order until one reaches a radius of order $r$. We omit the details. 
\end{proof}
\subsection{Outer estimates via duality}\label{s:inner_outer_decay}
Theorem~\ref{t:decaycube} states that if a discrete function is biharmonic in a subcube $Q_r(x)$ of $\Lambda_h^n$, then we have pointwise control over its second derivatives in a smaller subcube $Q_{r/2}(x)$. Remarkably, a dual statement is also true: If a discrete function is biharmonic outside a subcube $Q_r(x)$ of $\Lambda_h^n$, then we have control over its second derivatives outside of a larger subcube $Q_{2r}(x)$. The following lemma does not claim pointwise control, but only control in $L^2$. However we will combine it with Theorem~\ref{t:decaycube} into Theorem~\ref{t:decaycubeoutside} where we actually obtain pointwise control. 

\begin{lemma}\label{l:decaycubeoutsideav}
Let $u_h\in \Phi_h$. Let $x\in\Lambda_h^n$, $r\ge d(x)$ and suppose that $\Delta_h^2u_h(x)=0$ for all $x\in\inte\Lambda_h^n\setminus Q_r(x)$. Then, for all $s\ge r$,
\begin{equation}
	\|\nabla_h^2u_h\|_{L^2(\R^n\setminus Q_s(x))}\le C\left(\frac{r}{s}\right)^\frac n2\|\nabla_h^2u_h\|_{L^2(\R^n\setminus Q_r(x))}\,.\label{e:decaycubeoutsideav2}
\end{equation}
\end{lemma}

\begin{proof}
Consider first the case $r<h$. Then $d(x)=0$, i.e. $x\in\partial\Lambda_h^n$, and the assumptions imply $\Delta_h^2u_h=0$ in $\inte\Lambda_h^n$, i.e. $u_h=0$ in $\inte\Lambda_h^n$ by the uniqueness of the bilaplacian equation. So both sides of \eqref{e:decaycubeoutsideav2} are zero and the inequality holds.

So we can assume $r\ge h$. The statement is trivial in the case that $s<23r$, so we can also assume $s\ge23r$. We can then replace $r$ and $s$ by $\tilde{r}=\lfloor r-\frac h2\rfloor_h+\frac {3h}2$ and $\tilde{s}=\lfloor s-\frac h2\rfloor_h+\frac h2$, respectively. It is easy to see that then $\tilde{r}\ge r$, $\tilde{s}\le s$ and $\tilde{s}\ge11\tilde{r}$, and it suffices to prove the theorem for $\tilde{r},\tilde{s}$.
So we will directly assume $r,s\in h\N+\frac h2$, $s\ge11r$ and $r\ge\frac{3h}2$.

Let $f_h=\nabla_h^2u_h\chi_{\Lambda_h^n\setminus Q_s(x)}$, where $\chi_A$ is the indicator function of a set $A$. Let $v_h\in \Phi_h$ be the unique solution of $\Delta^2v_h=\diverg_{-h}\diverg_hf_h$. Then, for any $\varphi_h\in\Phi_h$,
\begin{equation}
	(\nabla_h^2v_h,\nabla_h^2\varphi_h)_{L^2(\R^n)}=(f_h,\nabla_h^2\varphi_h)_{L^2(\R^n)}\,.\label{e:decaycubeoutsideav3}
\end{equation}

Also let $\zeta_h$ and $\eta_h$ be discrete cut-off functions such that $\zeta_h$ is 1 on $\Lambda_h^n\setminus Q_{5r}(x)$, 0 on $Q_{3r}(x)\cap\Lambda_h^n$, $\eta_h$ is 1 on $Q_{7r}(x)\cap\Lambda_h^n$, 0 on $\Lambda_h^n\setminus Q_{9r}(x)$ and such that $|\nabla_h^\kappa\zeta_h|\le\frac{C}{r^\kappa}$ and $|\nabla_h^\kappa\eta_h|\le\frac{C}{r^k}$ for $\kappa\le 2$.

These choices ensure that 
\begin{equation}
	\nabla_h^2(\zeta_hu_h)=\nabla_h^2u_h\text{ on the support of }f_h\label{e:decaycubeoutsideav4}
\end{equation}
and that
\begin{equation}
	\eta_h=1\text{ on the support of }\Delta_h^2(\zeta_hu_h)\,.\label{e:decaycubeoutsideav5}
\end{equation}
Indeed, for example the support of $\Delta_h^2(\zeta_hu_h)$ is contained in $Q_{5r+2h}(x)\setminus Q_{3r-2h}(x)\subset Q_{7r}(x)$.

This implies

\begin{align}
	\|\nabla_h^2u_h\|^2_{L^2(\R^n\setminus Q_s(x))}&=(f_h,\nabla_h^2u_h)_{L^2(\R^n)}\nonumber\\
	&\stackrel{\mathclap{\eqref{e:decaycubeoutsideav4}}}{=}(f_h,\nabla_h^2(\zeta_hu_h))_{L^2(\R^n)}\nonumber\\
	&\stackrel{\mathclap{\eqref{e:decaycubeoutsideav3}}}{=}(\nabla_h^2v_h,\nabla_h^2(\zeta_hu_h))_{L^2(\R^n)}\nonumber\\
	&=(v_h,\Delta_h^2(\zeta_hu_h))_{L^2(\R^n)}\nonumber\\
	&\stackrel{\mathclap{\eqref{e:decaycubeoutsideav5}}}{=}(\eta_hv_h,\Delta_h^2(\zeta_hu_h))_{L^2(\R^n)}\nonumber\\
	&=(\nabla_h^2(\eta_hv_h),\nabla_h^2(\zeta_hu_h))_{L^2(\R^n)}\nonumber\\
	&\le\|\nabla_h^2(\eta_hv_h)\|_{L^2(\R^n)}\|\nabla_h^2(\zeta_hu_h)\|_{L^2(\R^n)}\,.\label{e:decaycubeoutsideav}
\end{align}
Now by the product rule
\begin{align*}
	&\nabla_h^2(\eta_hv_h)\\
	&=\sum_{i,j=1}^nD^h_{-i}D^h_j\eta_hv_h+\tau^h_jD^h_{-i}\eta_hD^h_jv_h+\tau^h_{-i}D^h_j\eta_hD^h_{-i}v_h+\tau^h_{-i}\tau^h_j\eta_hD^h_{-i}D^h_jv_h
\end{align*}
and so, using the Poincaré inequality\footnote{Here we have used the assumption $r\ge d(x)$ (or rather $7r\ge d(x)$): It ensures that we have zero boundary data somewhere on $Q^h_{7r}(x)\setminus Q^h_r(x)$ so that we can indeed use the Poincaré inequality.} on $Q_{9r}(x)$,
\begin{align}
	&\|\nabla_h^2(\eta_hv_h)\|_{L^2(\R^n)}\nonumber\\
	&\quad\le\frac{C}{r^2}\|v_h\|_{L^2(Q_{9r}(x))}+\frac{C}{r}\|\nabla_hv_h\|_{L^2(Q_{9r}(x))}+C\|\nabla_h^2v_h\|_{L^2(Q_{9r}(x))}\nonumber\\
	&\quad\le C\|\nabla_h^2v_h\|_{L^2(Q_{9r}(x))}\,.\label{e:poincare}
\end{align}
Similarly, by the Poincaré inequality on the annulus $Q_{7r}(x)\setminus Q_r(x)$,
\begin{align*}
	&\|\nabla_h^2(\zeta_hu_h)\|_{L^2(\R^n)}\\
	&\le\frac{C}{r^2}\|u_h\|_{L^2(Q_{7r}(x)\setminus Q_r(x))}+\frac{C}{r}\|\nabla_hu_h\|_{L^2(Q_{7r}(x)\setminus Q_r(x)}+C\|\nabla_h^2u_h\|_{L^2(\R^n\setminus Q_r(x))}\\
	&\le C\|\nabla_h^2u_h\|_{L^2(\R^n\setminus Q_r(x))}\,.
\end{align*}

If we plug the last two estimates into \eqref{e:decaycubeoutsideav} and then use Theorem~\ref{t:decaycube} for $v_h$ we obtain
\begin{align*}
	\|\nabla_h^2u_h\|^2_{L^2(\Lambda_h^n\setminus Q_s(x))}&\le C\|\nabla_h^2v_h\|_{L^2(Q_{9r}(x))}\|\nabla_h^2u_h\|_{L^2(\R^n\setminus Q_r(x))}\nonumber\\
	&\le C\left(\frac{9r}{s}\right)^\frac n2\|\nabla_h^2(v_h)\|_{L^2(\R^n)}\|\nabla_h^2u_h\|_{L^2(\R^n\setminus Q_r(x))}\,.
\end{align*}
This is equivalent to \eqref{e:decaycubeoutsideav2} once we use the energy estimate 
\[\|\nabla_h^2v_h\|_{L^2(\R^n)}\le\|f_h\|_{L^2(\R^n)}=\|\nabla_h^2u_h\|_{L^2(\R^n\setminus Q_s(x))}\,.\] 
\end{proof}
Now we can combine this lemma with Theorem~\ref{t:decaycube} to obtain a pointwise outer estimate.
\begin{theorem}\label{t:decaycubeoutside}
Let $u_h\in \Phi_h$. Let $x\in\Lambda_h^n$, $r>0$ and suppose that $\Delta_h^2u_h(x)=0$ for all $x\in\inte\Lambda_h^n\setminus Q_r(x)$.

Then, for all $y\in \Lambda_h^n\setminus Q_{2r}(x)$,
\begin{equation}
	|\nabla_h^2u_h(y)|\le C\frac{(\max(d(x),r))^\frac n2}{|x-y|^n}\|\nabla_h^2u_h\|_{L^2(\R^n\setminus Q_r(x))}\,.\label{e:decaycubeoutside}
\end{equation}
\end{theorem}
\begin{proof}
As in the proof of Lemma~\ref{l:decaycubeoutsideav} we see that $d(x)=0$ implies $u=0$ everywhere and \eqref{e:decaycubeoutside} holds. So assume $d(x)\ge h$.

Let $y\in \Lambda_h^n\setminus Q_{2r}(x)$. If $y\in Q_{2d(x)}(x)$ we use Theorem~\ref{t:decaycube} on $Q_{d(x)}(y)\subset\R^n\setminus Q_{2r}(x)$ to obtain
\[|\nabla_h^2u_h(y)|\le\frac{C}{d(x)^\frac n2}\|\nabla_h^2u_h\|_{L^2(Q_{d(x)}(y))}\le \frac{C}{d(x)^\frac n2}\|\nabla_h^2u_h\|_{L^2(\R^n\setminus Q_{2r}(x))}\,,\]
which implies \eqref{e:decaycubeoutside} because $|x-y|\le\sqrt{n}|x-y|_\infty\le2\sqrt{n}d(x)$ and hence $\frac{1}{d(x)}\le 4n\frac{d(x)}{|x-y|^2}$.

If, on the other hand, $y\in \Lambda_h^n\setminus Q_{2d(x)}(x)$ then we use Theorem~\ref{t:decaycube} on $Q_{|x-y|_\infty/2}(y)$ and then Lemma~\ref{l:decaycubeoutsideav} as follows:
\begin{align*}
	|\nabla_h^2u_h(y)|&\le\frac{C}{\left(\frac{|x-y|_\infty}{2}\right)^\frac n2}\|\nabla_h^2u_h\|_{L^2(Q_{|x-y|_\infty/2}(y))}\\
	&\le\frac{C}{|x-y|_\infty^{\frac n2}}\|\nabla_h^2u_h\|_{L^2(\R^n\setminus Q_{|x-y|_\infty/2}(x))}\\
	&\le\frac{C}{|x-y|_\infty^\frac n2}\left(\frac{\max(d(x),r)}{\frac{|x-y|_\infty}{2}}\right)^\frac n2\|\nabla_h^2u_h\|_{L^2(\R^n\setminus Q_{\max(d(x),r)}(x))}\\
	&\le C\frac{(\max(d(x),r))^\frac n2}{|x-y|_\infty^n}\|\nabla_h^2u_h\|_{L^2(\R^n\setminus Q_{\max(d(x),r)}(x))}\,,
\end{align*}
which implies \eqref{e:decaycubeoutside}. 
\end{proof}

\section{The discrete full-space Green's function}\label{s:full_space_green}
In order to obtain estimates for $G_h$, we will compare $G_h$ with a Green's function of $(h\Z)^n$. In the absence of boundary conditions such a Green's function is not uniquely defined. We will choose a normalization that is best suited for our application. The necessary asymptotics for the Green's function of $(h\Z)^n$, have been derived by Mangad \cite{Mangad1967} using Fourier-theoretic methods.

By $\mathcal{F}$ we denote the Fourier transform of tempered distributions (where we use the convention $(\mathcal{F}f)(x)=\int_{\R^n}f(\xi)\mathrm{e}^{-2\pi i x\cdot\xi}\ud\xi$).
\[\]
\begin{theorem}[\cite{Mangad1967}, Section 4]\label{t:asymptgreenfull}
Let $n\in\N^+$. Define $F\colon\Z^n\times\Z^n\rightarrow\R$ by
\[F(x,y)=\mathcal{F}\left(\frac{V(\xi)}{\left(4\sum_{j=1}^n\sin^2(\pi \xi_j)\right)^2} \right)(x-y)\]
where $V\in C_c^\infty([-1,1]^n)$ is chosen such that $V=1$ near 0 and $\sum_{z\in\Z}V(x+z)=1$ for all $x$ and $\frac{V(\xi)}{\left(4\sum_{j=1}^n\sin^2(\pi \xi_j)^2\right)^2}$ denotes the tempered distribution given by its finite part in the sense of Hadamard (see \cite[Chapitre II, §2 and §3]{Schwartz1966}).

Then $F$ is a Green's function for $\Delta_1^2$ in the sense that $\Delta_1^2F(\cdot,y)=\delta_y$. It satisfies the following asymptotic expansion: If $n=2$ and $z=x-y$,
\begin{align*}
	F(x,y)&=\frac{|z|^2\log|z|}{8\pi}+\frac{(\gamma-1+\log\pi)|z|^2}{8\pi}-\frac{\log|z|}{16\pi}+\frac{4(z_1^4+z_2^4)}{|z|^4}\\
	&\quad-12\log\pi-12\gamma-3+O\left(\frac{1}{|z|^2}\right)
\end{align*}
where $\gamma$ is the Euler-Mascheroni constant, and if $n=3$ and $z=x-y$,
\[F(x,y)=-\frac{|z|}{8\pi}+\frac{z_1^4+z_2^4+z_3^4}{64\pi|z|^5}+\frac{1}{64\pi|z|}+O\left(\frac{1}{|z|^3}\right)\,.\]

\end{theorem}Let us briefly sketch how to prove this theorem: Observe that $\sigma(\xi):=\left(4\sum_{j=1}^n\sin^2(\pi \xi_j)\right)^2$ is the symbol of $\Delta_1^2$, so that $\Delta_1^2F(x,y)=\mathcal{F}(V)(x-y)$. On the other hand one easily checks that $\sum_{z\in\Z}V(x+z)=1$ implies that $\mathcal{F}(V)(m)=\delta_0(m)$ for any $m\in\Z^n$. This proves that $F$ is a Green's function. To derive the asymptotic expansion, one develops a Laurent series
\[\frac{1}{\sigma(\xi)}=\frac{1}{16\pi^2|\xi|^4}+\frac{f_{-2}(\xi)}{|\xi|^2}+f_0(\xi)+\cdots+o(|\xi|^N)\,.\]
Then one can check using the explicit formulas for the Fourier transforms of $|\xi|^m$ (see \cite{Schwartz1966}) and the Riemann-Lebesgue lemma that
\[\mathcal{F}\left(\frac{V(\xi}{\sigma(\xi)}-\frac{1}{16\pi^2|\xi|^4}+\frac{f_{-2}(\xi)}{|\xi|^2}+f_0(\xi)+\cdots \right)=o(|x|^{-n-N}) \]
so it suffices to compute the Fourier transform of $\frac{1}{16\pi^2|\xi|^4}+\frac{f_{-2}(\xi)}{|\xi|^2}+f_0(\xi)+\cdots$. This one can again do explicitly and thereby obtain an asymptotic expansion for $F$ up to $O(|x|^N)$. For details we refer to \cite{Mangad1967}.

By scaling the lattice we can deduce from this estimates for Green's functions on $(h\Z)^n$. We state the estimates that we will need.
\begin{lemma}\label{l:estgreenfull}
Let $n=2$ or $n=3$, $h>0$, $r\ge 4h$. There exists a function $\tilde{G}_h\colon(h\Z)^n\times(h\Z)^n\rightarrow\R$ such that $\Delta_h^2\tilde{G}_h(\cdot,y)=\delta_{h,y}$ and such that the following estimates are satisfied:
\begin{align}
|\nabla_{h,y}\tilde{G}_h(x,y)|&\le Cr^{3-n}\quad&&\text{if }|x-y|_\infty\le\frac r2\,,\label{e:estgreenfull3}\\
|\nabla_{h,x}^2\nabla_{h,y}\tilde{G}_h(x,y)|&\le \frac{C}{(|x-y|+h)^{n-1}}\quad&&\text{if }|x-y|_\infty\le\frac r2\,,\label{e:estgreenfull1}\\
|\nabla_{h,x}^2\nabla_{h,y}^2\tilde{G}_h(x,y)|&\le \frac{C}{(|x-y|+h)^n}\quad&&\text{if }|x-y|_\infty\le\frac r2\label{e:estgreenfull4}
\intertext{and}
|D_{h,x}^\alpha D_{h,y}^\beta\tilde{G}_h(x,y)|&\le Cr^{4-n-|\alpha|-|\beta|}\quad&&\text{if }\frac r2\le|x-y|_\infty\le r,|\alpha|+|\beta|\le4\,.\label{e:estgreenfull2}
\end{align}
\end{lemma}
For $n=2$ the function $\tilde{G}_h$ depends on $r$, but we will suppress this dependence for ease of notation.
\begin{proof}
We begin with the slightly easier case $n=3$. The asymptotic expansion in Theorem~\ref{t:asymptgreenfull} easily implies that
\[|D_{1,x}^\alpha D_{1,y}^\beta F(x,y)|\le C|x-y|^{1-|\alpha|-|\beta|}\]
for $|\alpha|+|\beta|\le4$ and any $x,y$ with $|x-y|\ge10$, say (observe that $g=O(|x|^{-3})$ implies $D^1_{\pm i}g(x)=O(|x|^{-3})$, so we do not need to care about the error term). On the other hand $F$ is finite everywhere, so that 
\[|D_{1,x}^\alpha D_{1,y}^\beta F(x,y)|\le C\]
for $|\alpha|+|\beta|\le4$ and any $x,y$ with $|x-y|<10$. If we combine these two estimates we conclude that we have
\[|D_{1,x}^\alpha D_{1,y}^\beta F(x,y)|\le C(|x-y|+1)^{1-|\alpha|-|\beta|}\,.\]
Now if we set $\tilde{G}_h(x,y)=hF\left(\frac xh,\frac yh\right)$ then $\tilde{G}_h$ satisfies
\[|D_{h,x}^\alpha D_{h,y}^\beta \tilde{G}_h(x,y)|\le C(|x-y|+h)^{1-|\alpha|-|\beta|}\,,\]
which immediately implies the claimed estimates.

If $n=2$ we need to take care of the logarithmic terms. So we set 
\[\tilde{F}(x,y)=F(x,y)+\frac{|x-y|^2\log\left(\frac hr\right)}{8\pi}\,.\]
Then $\tilde{F}$ has the asymptotic expansion
\begin{align*}
	\tilde{F}(z)&=\frac{z|^2\log|z|}{8\pi}+\frac{(\log\left(\frac hr\right)+\gamma-1+\log\pi)|z|^2}{8\pi}-\frac{\log|z|}{16\pi}\\
	&\quad+\frac{4(z_1^4+z_2^4)}{|z|^4}-12\log\pi-12\gamma-3+O\left(\frac{1}{|z|^2}\right)
\end{align*}
and this implies
\[|D_{1,x}^\alpha D_{1,y}^\beta \tilde{F}(x,y)|\le C|x-y|^{2-|\alpha|-|\beta|}\left(\left|\log|x-y|+\log\left(\frac hr\right)\right|+1\right)\]
for $|\alpha|+|\beta|\le2$ and any $x,y$ with $|x-y|\ge10$. Because $D_{1,x}^\alpha D_{1,y}^\beta \tilde{F}(x,y)$ is bounded by $C\left(1+\left|\log\left(\frac hr\right)\right|\right)$ for $|x-y|<10$, we conclude
\[|D_{1,x}^\alpha D_{1,y}^\beta \tilde{F}(x,y)|\le C(|x-y|+1)^{2-|\alpha|-|\beta|}\left|\log\left(\frac {h(|x-y|+1)}{r}\right)\right|\,.\]
We now set $\tilde{G}_h(x,y)=h^2\tilde{F}\left(\frac xh,\frac yh\right)$ and obtain
\[|D_{h,x}^\alpha D_{h,y}^\beta \tilde{G}_h(x,y)|\le C(|x-y|+h)^{2-|\alpha|-|\beta|}\left|\log\left(\frac {|x-y|+h}{r}\right)\right|\,.\]
It is easy to check that this implies \eqref{e:estgreenfull3} and \eqref{e:estgreenfull2} for $|\alpha|+|\beta|\le2$. If $|\alpha|+|\beta|\ge3$ we need to be slightly more careful: Observe that third discrete derivatives of $|x-y|^2$ vanish, so that we actually have
\[|\nabla_{1,x}^\alpha\nabla_{1,y}^\beta\tilde{F}(x-y)|\le \frac{C}{|x-y|^{|\alpha|+|\beta|-2}}\]
if $|x-y|\ge 10$
from which we conclude
\[|\nabla_{1,x}^\alpha\nabla_{1,y}^\beta\tilde{F}(x-y)|\le \frac{C}{(|x-y|+1)^{|\alpha|+|\beta|-2}}\]
for any $x,y$. Recalling that $\tilde{G}_h(x,y)=h^2\tilde{F}\left(\frac xh,\frac yh\right)$ we immediately obtain \eqref{e:estgreenfull1}, \eqref{e:estgreenfull4} and \eqref{e:estgreenfull2} for $|\alpha|+|\beta|\ge3$. 
\end{proof}
\section{Proof of the main theorem}\label{s:proofs}
We are now able to prove Theorem~\ref{t:mainthmh}. We first give the straightforward proof of part ii) and then continue with part i).
\subsection{\texorpdfstring{Lower bounds for $G_h(x,x)$}{Lower bounds for G\_h(x,x)}}
The proof is rather short and based on the choice of an appropriate test function.

\begin{proof}[Proof of Theorem~\ref{t:mainthmh} ii)]$ $\\
We can assume $d(x)\ge h$, as otherwise $d(x)=0$ and hence $G_h(x,x)=0$. If we test the equation $\Delta_h^2G_{h,x}=\delta_{h,x}$ with $G_{h,x}$, we find
\begin{equation}
	\|\nabla_h^2G_{h,x}\|^2_{L^2(\R^n)}=(\Delta_h^2G_{h,x},G_{h,x})_{L^2(\R^n)}=(\delta_{h,x},G_{h,x})_{L^2(\R^n)}=G_h(x,x)\,.\label{e:greenweakenergy}
\end{equation}
Now let $\varphi_h\in \Phi_h$. Then testing the equation $\Delta_h^2G_{h,x}=\delta_{h,x}$ with $\varphi_h$ and using the Cauchy-Schwarz inequality we find
\begin{align*}
\varphi_h(x)&=(\nabla_h^2G_{h,x},\nabla_h^2\varphi_h)_{L^2(\R^n)}\\
&\le\|\nabla_h^2G_{h,x}\|_{L^2(\R^n)}\|\nabla_h^2\varphi_h\|_{L^2(\R^n)}\\
&=\sqrt{G_h(x,x)}\|\nabla_h^2\varphi_h\|_{L^2(\R^n)}\,.
\end{align*}
If $\varphi_h$ is not identically zero this implies
\[G_h(x,x)\ge\frac{(\varphi_h(x))^2}{\|\nabla_h^2\varphi_h\|^2_{L^2(\R^n)}}\]
and so it remains to find a $\varphi_h(x)$ such that $\frac{\varphi_h(x)}{\|\nabla_h^2\varphi_h\|_{L^2(\R^n)}}\ge Cd(x)^{2-\frac n2}$. But this is easy:\\
Take $\varphi_{h,x}\in \Phi_h$ supported in $Q_{d(x)}(x)$ such that $\varphi_{h,x}(x)=1$ and such that $|\nabla_h^2\varphi_{h,x}|\le \frac{C}{d(x)^2}$ and extend it by 0 to all of $\Lambda^n_h$. 
\end{proof}

\subsection{\texorpdfstring{Upper bounds for $G_h(x,y)$}{Upper bounds for G\_h(x,y)}}
In this section we prove part i) of  Theorem~\ref{t:mainthmh}.

We begin with a rather weak estimate for $G_h(x,y)$.
\begin{lemma}\label{l:upperbound}
Let $n=2$ or $n=3$ and $G_h$ be the Green's function of $\Lambda^n_h$. Then we have
\begin{equation}
	0\le G_h(x,x)=\|\nabla_h^2G_{h,x}\|^2_{L^2(\R^n)} \le Cd(x)^{4-n}\label{e:upperbound1}
\end{equation}
for any $x\in\Lambda_h^n$ and
\begin{equation}
	|G_h(x,y)| \le Cd(x)^{2-\frac n2}d(y)^{2-\frac n2}\label{e:upperbound2}
\end{equation}
for any $x,y\in\Lambda_h^n$.
\end{lemma}
\begin{proof}
We first prove \eqref{e:upperbound1}. By \eqref{e:greenweakenergy} we have
\begin{equation}
	\|\nabla_h^2G_{h,x}\|^2_{L^2(\R^n)}=G_h(x,x)\,.\label{e:upperbound3}
\end{equation}
If $x\in\partial\Lambda_h^n$ then $G_h(x,x)=0$ and \eqref{e:upperbound1} holds. So assume $x\in\inte\Lambda_h^n$, i.e. $d(x)\ge h$. The Sobolev-Poincaré inequality implies that
\begin{align*}
	G_h(x,x)&\le \|G_{h,x}\|_{L^\infty(Q_{d(x)+h/2}(x))}\\
	&\le C\left(d(x)+\frac h2\right)^{2-\frac n2}\|\nabla_h^2G_{h,x}\|_{L^2(Q_{d(x)+h/2}(x))}\\
	&\le C d(x)^{2-\frac n2}\|\nabla_h^2G_{h,x}\|_{L^2(Q_{2d(x)+h/2}(x))}\,.
\end{align*}

If we combine this estimate with \eqref{e:upperbound3} we find that
\begin{align*}
	\|\nabla_h^2G_{h,x}\|^2_{L^2(\R^n)}=G_h(x,x)&\le Cd(x)^{2-\frac n2}\|\nabla_h^2G_{h,x}\|_{L^2(Q_{2d(x)+h/2}(x))}\\
	&\le Cd(x)^{2-\frac n2}\|\nabla_h^2G_{h,y}\|_{L^2(\R^n)}
\end{align*}
and hence
\[0\le G_h(x,x)=\|\nabla_h^2G_{h,x}\|^2_{L^2(\R^n)}\le Cd(x)^{4-n}\,.\]
This proves \eqref{e:upperbound1}. For \eqref{e:upperbound2}, we test $\Delta_h^2G_{h,x}=\delta_{h,x}$ with $G_{h,y}$ and use the Cauchy-Schwarz inequality to obtain
\begin{align*}
	|G_h(x,y)|&=\left|(\delta_{h,x},G_{h,y})_{L^2(\R^n)}\right|\\
	&=\left|(\nabla_h^2G_{h,x},\nabla_h^2G_{h,y})_{L^2(\R^n)}\right|\\
	&\le\|\nabla_h^2G_{h,x}\|_{L^2(\R^n)}\|\nabla_h^2G_{h,y}\|_{L^2(\R^n)}\\
	&\stackrel{\mathclap{\eqref{e:upperbound1}}}{\le}Cd(x)^{2-\frac n2}d(y)^{2-\frac n2}\,.
\end{align*} 
\end{proof}

The next lemma gives estimates for $G_h$ and its derivatives that are sharp when $x$ and $y$ are far apart. We first prove a pointwise estimate for $\nabla_{h,x}^2\nabla_{h,y}G_h$ by applying Theorem~\ref{t:decaycubeoutside} to a cut-off version of $\nabla_{h,y}G_{h,y}$. Afterwards we integrate it along suitable paths to deduce the estimates in the lemma.

\begin{lemma}\label{l:estderfar}
Let $n=2$ or $n=3$ and $G_h$ be the Green's function of $\Lambda^n_h$. If $x,y\in\Lambda_h^n$ and $|x-y|_\infty>\frac{d(y)}{8}$ then
\begin{align}
	|G_h(x,y)|&\le C\frac{(d(x)+h)^2(d(y)+h)^2}{|x-y|^n}\,,\label{e:estGuppfar}\\
	|\nabla_{h,x}G_h(x,y)|&\le C\frac{(d(x)+h)(d(y)+h)^2}{|x-y|^n}\,,\label{e:estnablaxGuppfar}\\
	|\nabla_{h,x}^2G_h(x,y)|&\le C\frac{(d(y)+h)^2}{|x-y|^n}\,,\label{e:estnablax2Guppfar}\\
	|\nabla_{h,x}\nabla_{h,y}G_h(x,y)|&\le C\frac{(d(x)+h)(d(y)+h)}{|x-y|^n}\,.\label{e:estnablaxnablayGuppfar}
\end{align}
\end{lemma}
\begin{proof}$ $\\
\emph{Step 1: Pointwise estimate for $\nabla_{h,x}^2\nabla_{h,y}G_h(x,y)$}\\
We claim that if $x,y\in\Lambda_h^n$ and $|x-y|_\infty>\frac{d(y)}{8}$ then
\begin{equation}
	|\nabla_{h,x}^2\nabla_{h,y}G_h(x,y)|\le C\frac{d(y)+h}{|x-y|^n}\,.\label{e:estnablax2nablayGuppfar}
\end{equation}
In the following all derivatives will be with respect to $x$ unless we mark them with a sub- or superscript $y$.

If $d(y)<160h$ we can use a trivial estimate: From Lemma~\ref{l:upperbound} we know
\[\|\nabla_h^2G_{h,y'}\|_{L^2(\R^n)}\le Cd(y')^{2-\frac n2}\le Ch^{2-\frac n2}\]
if $|y'-y|_\infty\le h$.
If we now use \[|D^h_if_h(y)|^2=\left(\frac1h(f_h(y+e_i)-f_h(y)\right)^2\le\frac{2}{h^2}(f_h(y+e_i)^2+f_h(y)^2)\]
with $f=\nabla_h^2G_h$ we get that
\[\|\nabla_h^2D^{h,y}_iG_h\|^2_{L^2(\R^n)}\le\frac{2}{h^2}\left(\|\nabla_h^2\tau^{h,y}_iG_h\|^2_{L^2(\R^n)}+\|\nabla_h^2G_h\|^2_{L^2(\R^n)}\right)\le Ch^{2-n}\,,\]
i.e.
\[\|\nabla_h^2D^{h,y}_iG_h\|_{L^2(\R^n)}\le Ch^{1-\frac n2}\,.\]
Then Theorem~\ref{t:decaycubeoutside} with $r=h$ implies
\begin{align*}
	|\nabla_h^2D^{h,y}_iG_h(x,y)|&\le C\frac{\max(d(y),h)^\frac n2}{|x-y|^n}\|\nabla_h^2D^{h,y}_iG_h\|_{L^2(\R^n)}\\
	&\le C\frac{h^\frac n2}{|x-y|^n}h^{1-\frac n2}=C\frac{h}{|x-y|^n}\,,
\end{align*}
which implies \eqref{e:estnablax2nablayGuppfar} if we choose $C$ there large enough.

So assume $d(y)\ge160h$. Let $\eta_h$ be a discrete cut-off function that is 1 on $Q_{d(y)/32+2h}$, 0 on $(h\Z)^n\setminus Q_{d(y)/16-2h}(x)$, and such that $|\nabla^\kappa\eta_h|\le\frac{C}{d(y)^\kappa}$ for $\kappa\le2$. Let $H_h(x,y)=G_h(x,y)-\eta_h(x)\tilde{G}_h(x,y)$, where $\tilde{G}_h$ is the function from Lemma~\ref{l:estgreenfull} with $r=\frac{d(y)}{16}$. We write $H_{h,y}$ for $H_h(\cdot,y)$.

Then, for $i\in\{1,\ldots, n\}$, $D^{h,y}_iH_{h,y}\in\Phi_h$. Also, the singularities near $y$ cancel out, so that $\Delta_h^2D^{h,y}_iH_{h,y}=0$ in $Q_{d(y)/32}(y)$ and in $\inte\Lambda_h^n\setminus Q_{d(y)/16}(y)$.

Next, we want to bound $\|\nabla_h^2D^{h,y}_iH_{h,y}\|_{L^2(\R^n)}$. To do so, we introduce another cut-off function $\zeta_h$ that is 1 on $\inte\Lambda_h^n\setminus Q_{d(y)/32}(y)$, 0 on $Q_{d(y)/64}(y)$ and such that $|\nabla^\kappa\zeta_h|\le\frac{C}{d(y)^\kappa}$ for $\kappa\le2$. Then we have that
\begin{align*}
	\Delta_h^2D^{h,y}_iH_{h,y}=\zeta_h\Delta_h^2D^{h,y}_iH_{h,y}&=-\zeta_h\Delta_h^2D^{h,y}_i\left(\eta_h\tilde{G}_{h,y}\right)\\
	&=-\zeta_h\Delta_h^2\left(\eta_hD^{h,y}_i\tilde{G}_{h,y}\right)
\end{align*}
where we have used that $\eta_h$ does not depend on $y$. Thus
\begin{align}
	\|\nabla_h^2D^{h,y}_iH_{h,y}\|^2_{L^2(\R^n)}&=(\Delta_h^2D^{h,y}_iH_{h,y},D^{h,y}_iH_{h,y})_{L^2(\R^n)}\nonumber\\
	&=-(\zeta_h\Delta_h^2(\eta_hD^{h,y}_i\tilde{G}_{h,y}),D^{h,y}_iH_{h,y})_{L^2(\R^n)}\nonumber\\
	&=-(\Delta_h^2(\eta_hD^{h,y}_i\tilde{G}_{h,y}),\zeta_hD^{h,y}_iH_{h,y})_{L^2(\R^n)}\nonumber\\
	&=-(\nabla_h^2(\eta_hD^{h,y}_i\tilde{G}_{h,y}),\nabla_h^2(\zeta_hD^{h,y}_iH_{h,y}))_{L^2(\R^n)}\nonumber\\
	&\le\|\nabla_h^2(\eta_hD^{h,y}_i\tilde{G}_{h,y})\|_{L^2(\R^n)}\|\nabla_h^2(\zeta_hD^{h,y}_iH_{h,y})\|_{L^2(\R^n)}\,.\label{e:estderfar1}
\end{align}

If we use the pointwise estimates for $\tilde{G}_{h,y}$ from Lemma~\ref{l:estgreenfull}, we conclude
\[|\nabla_h^2(\eta_hD^{h,y}_i\tilde{G}_{h,y})|\le Cd(y)^{1-n}\] and hence
\[\|\nabla_h^2(\eta_hD^{h,y}_i\tilde{G}_{h,y})\|_{L^2(\R^n)}\le Cd(y)^{1-\frac n2}\,.\]

Furthermore, as in \eqref{e:poincare}, the Poincaré inequality on $Q_{d(y)+h/2}(y)$ and the pointwise estimates for $\zeta_h$ imply that
\begin{align*}
	&\|\nabla_h^2(\zeta_hD^{h,y}_iH_{h,y})\|_{L^2(\R^n)}\\
	&\quad\le\frac{C}{d(y)^2}\|D^{h,y}_iH_{h,y}\|_{L^2(Q_{d(y)+h/2}(y))}+\frac{C}{d(y)}\|\nabla_hD^{h,y}_iH_{h,y}\|_{L^2(Q_{d(y)+h/2}(y))}\\
	&\quad\quad+\|\nabla_h^2D^{h,y}_iH_{h,y}\|_{L^2(\R^n)}\\
	&\quad\le C\|\nabla_h^2D^{h,y}_iH_{h,y}\|_{L^2(\R^n)}\,.
\end{align*}
If we combine the last two estimates with \eqref{e:estderfar1} we conclude that
\[\|\nabla_h^2D^{h,y}_iH_{h,y}\|_{L^2(\R^n)}\le Cd(y)^{1-\frac n2}\,.\]

We recall that $\Delta_h^2H_h=0$ in $\inte\Lambda_h^n\setminus Q_{d(y)/16}$ and use Theorem~\ref{t:decaycubeoutside} to find that, for $x\in\Lambda_h^n\setminus Q_{d(y)/8}(y)$,
\begin{align*}
	|\nabla_h^2D^{h,y}_iH_h(x)|&\le C\frac{d(y)^\frac n2}{|x-y|^n}\|\nabla_h^2D^{h,y}_iH_h\|_{L^2(\R^n)}\\
	&\le C\frac{d(y)^\frac n2}{|x-y|^n}d(y)^{1-\frac n2}=C\frac{d(y)}{|x-y|^n}\,.
\end{align*}
This implies \eqref{e:estnablax2nablayGuppfar} because $D^{h,y}_iH_{h,y}$ is equal to $D^{h,y}_iG_{h,y}$ in $\Lambda_h^n\setminus Q_{d(y)/16}(y)$ and therefore $\nabla_h^2D^{h,y}_iH_{h,y}$ is equal to $\nabla_h^2D^{h,y}_iG_{h,y}$ in $\Lambda_h^n\setminus Q_{d(y)/8}(y)$.

\emph{Step 2: Proof of \eqref{e:estnablaxnablayGuppfar}}\\
We can obtain \eqref{e:estnablaxnablayGuppfar} by integrating \eqref{e:estnablax2nablayGuppfar} along a well-chosen path in $x$. Let $(x^{(k)})_{k=0}^L$ be a path of length $Lh$ from $x^{(0)}=x$ to $x^{(L)}\in(h\Z)^n\setminus\Lambda_h^n$ such that $|x^{(k+1)}-x^{(k)}|_\infty=h$, $|x^{(k)}-y|\ge |x-y|_\infty$ for all $k$, and $L\le2(d(x)+h)$. To construct such a path begin with the straight path from $x$ to a closest point $x^*\in(h\Z)^n\setminus\Lambda_h^n$ (which will have length $d(x)+h$). If this path does not intersect $Q^h_{|x-y|_\infty-h}(y)$, we are done. Else we modify the path by taking a (shortest-possible) detour around $Q^h_{|x-y|_\infty-h}(y)$. This detour lengthens the path by at most $|x-y|_\infty$, and it is easy to check that if it is necessary then $y\in Q^h_{d(x)}(x)$, so that $|x-y|_\infty\le d(x)$, and our path has length at most $d(x)+h+|x-y|_\infty\le2(d(x)+h)$.

Now, by \eqref{e:estnablax2nablayGuppfar},
\begin{align*}
	|\nabla_{h,x}^2\nabla_{h,y}G_h(x^{(k)},y)|&\le C\frac{d(y)+h}{(|x^{(k)}-y|+h)^n}\\
	&\le C\frac{d(y)+h}{(|x^{(k)}-y|_\infty+h)^n}\le C\frac{d(y)+h}{(|x-y|_\infty+h)^n}\,.
\end{align*}
Now we can perform  discrete integration along $(x^{(k)})_{k=0}^L$: Observe that $\nabla_{h,x}\nabla_{h,y}G_h(x^{(L)},y)=0$ and so
\begin{align*}
	|\nabla_{h,x}\nabla_{h,y}G_h(x,y)|&\le\sum_{k=0}^{L-1}|\nabla_{h,x}\nabla_{h,y}G_h(x^{(k+1)},y)-\nabla_{h,x}\nabla_{h,y}G_h(x^{(k)},y)|\\
	&\le\sum_{k=0}^{L-1}h|\nabla_{h,x}^2\nabla_{h,y}G_h(x^{(k)},y)|\\
	&\le L\frac{d(y)+h}{(|x-y|_\infty+h)^n}\,,
\end{align*}
which implies \eqref{e:estnablaxnablayGuppfar}.

\emph{Step 3: Proof of \eqref{e:estnablax2Guppfar}}\\
We proceed as in the previous step with the only difference that this time we integrate in $y$ along a path that avoids $x$. Let $(y^{(k)})_{k=0}^L$ be a path of length $Lh$ from $y^{(0)}=y$ to $y^{(L)}\in(h\Z)^n\setminus\Lambda_h^n$ such that $|y^{(k+1)}-y^{(k)}|_\infty=h$, $|y^{(k)}-x|_\infty\ge |y-x|_\infty$ for all $k$, and $L\le2(d(y)+h)$. If we construct this path as in the previous step, we can in addition ensure that $d(y^{(k)})\le d(y)$ for all $k$ (then in particular $|y^{(k)}-x|_\infty\ge\frac{d(y^{(k)})}{8}$, so that \eqref{e:estnablax2nablayGuppfar} is applicable for all $y^{(k)}$).

Now by \eqref{e:estnablax2nablayGuppfar}
\begin{align*}
	|\nabla_{h,x}^2\nabla_{h,y}G_h(x,y^{(k)})|&\le C\frac{d(y^{(k)})+h}{(|x-y^{(k)}|+h)^n}\\
	&\le C\frac{d(y^{(k)})+h}{(|x-y^{(k)}|_\infty+h)^n}\le C\frac{d(y)+h}{(|x-y|_\infty+h)^n}
\end{align*}
and if we integrate this along $(y^{(k)})_{k=0}^L$, we obtain \eqref{e:estnablax2Guppfar}.

\emph{Step 4: Proof of \eqref{e:estnablaxGuppfar} and \eqref{e:estGuppfar}}\\We proceed as in the previous two steps. If we integrate \eqref{e:estnablax2Guppfar} along a path $(x^{(k)})_{k=0}^L$ that avoids $y$ once, we obtain \eqref{e:estnablaxGuppfar}, and if we integrate once more, we obtain \eqref{e:estGuppfar}. 
\end{proof}

Now we complement this lemma with an estimate when $x$ and $y$ are close:
\begin{lemma}\label{l:estderclose}
Let $n=2$ or $n=3$ and $G_h$ be the Green's function of $\Lambda^n_h$. If $x,y\in\Lambda_h^n$ and $|x-y|_\infty\le\frac{d(y)}{8}$ then
\begin{align}
	|G_h(x,y)|&\le C(d(x)+h)^{2-\frac n2}(d(y)+h)^{2-\frac n2}\,,\label{e:estGuppclose}\\
	|\nabla_{h,x}G_h(x,y)|&\le C(d(y)+h)^{3-n}\,,\label{e:estnablaxGuppclose}\\
	|\nabla_{h,x}^2G_h(x,y)|&\le \begin{cases} C\log\left(\frac{d(y)+h}{|x-y|+h}\right)&\quad n=2\\\frac{C}{|x-y|+h}&\quad n=3\end{cases}\,,\label{e:estnablax2Guppclose}\\
	|\nabla_{h,x}\nabla_{h,y}G_h(x,y)|&\le \begin{cases} C\log\left(\frac{(d(x)+h)(d(y)+h)}{(|x-y|+h)^2}\right)&\quad n=2\\\frac{C}{|x-y|+h}&\quad n=3\end{cases}\,.\label{e:estnablaxnablayGuppclose}
\end{align}
\end{lemma}
\begin{proof}$ $\\
\emph{Step 1: Pointwise estimate for $\nabla_{h,x}^2\nabla_{h,y}G_h(x,y)$}\\
We claim that if $x,y\in\Lambda_h^n$ and $|x-y|_\infty\le\frac{d(y)}{4}$ then
\begin{equation}
		|\nabla_{h,x}^2\nabla_{h,y}G_h(x,y)|\le \frac{C}{(|x-y|+h)^{n-1}}\,.\label{e:estnablax2nablayGuppclose}
\end{equation}
The fact that we prove this for $|x-y|_\infty\le\frac{d(y)}{4}$ will give us a bit of space to wiggle around in the following steps where we integrate \eqref{e:estnablax2nablayGuppclose}. The proof of \eqref{e:estnablax2nablayGuppclose} is similar to the proof of \eqref{e:estnablax2nablayGuppfar}. The main difference is that this time we choose the cut-off function further away from the singularity.

If $d(y)<10h$ we can again use a trivial estimate: By Lemma~\ref{l:upperbound}, $G_h(x',y')$ is bounded by $Cd(x')^{2-\frac n2}d(y')^{2-\frac n2}\le Ch^{4-n}$ if $|x'-x|_\infty\le h$ and $|y'-y|_\infty\le h$, so that 
\[|\nabla_{h,x}^2\nabla_{h,y}G_h(x,y)|\le C\frac{1}{h^3}h^{4-n}=Ch^{1-n}\,.\]
Therefore \eqref{e:estnablax2nablayGuppclose} holds if we choose $C$ sufficiently large.

So assume that $d(y)\ge10h$. Let $\eta_h$ be a discrete cut-off function that is 1 on $Q_{d(y)/2+2h}(y)$ and 0 on $(h\Z)^n\setminus Q_{d(y)-2h}(y)$ and such that $|\nabla^\kappa\eta_h|\le\frac{C}{d(x)^k}$ for $\kappa\le2$ and let $H_h(x,y)=G_h(x,y)-\eta_h(x)\tilde{G}_h(x,y)$, where $\tilde{G}_h$ is the function from Lemma~\ref{l:estgreenfull} with $r=d(y)$.

Then, for $i\in\{1,\ldots ,n\}$, $D^{h,y}_iH_{h,y}\in\Phi_h$ and $\Delta_h^2D^{h,y}_iH_{h,y}=0$ in \linebreak $Q_{d(y)/2}(y)$ and in $\inte\Lambda_h^n\setminus Q_{d(y)}(y)$. We can estimate $\|\nabla_h^2D^{h,y}_iH_{h,y}\|_{L^2(\R^n)}$ just as in Step 1 of the proof of Lemma~\ref{l:estderfar} and obtain that
\begin{equation}
	\|\nabla_h^2D^{h,y}_iH_{h,y}\|_{L^2(\R^n)}\le Cd(y)^{1-\frac n2}\,.\label{e:estderclose5}
\end{equation}
Now recall that $H_h$ is biharmonic in $Q_{d(y)/2}(y)$. So Theorem~\ref{t:decaycube} implies for $x\in Q^h_{d(y)/4}(y)$
\[|\nabla_h^2D^{h,y}_iH_{h,y}(x)|\le\frac{C}{d(y)^\frac n2}\|\nabla_h^2D^{h,y}_iH_{h,y}\|_{L^2(\R^n)}\le Cd(y)^{1-n}\,.\]
Because $\nabla_h^2D^{h,y}_iH_{h,y}=\nabla_h^2D^{h,y}_iG_{h,y}-\nabla_h^2D^{h,y}_i\tilde{G}_{h,y}$ in $Q_{d(y)/2}(y)$ we can use \eqref{e:estgreenfull1} and obtain
\begin{align*}
	|\nabla_h^2D^{h,y}_iG_{h,y}(x)|&\le|\nabla_h^2D^{h,y}_iH_{h,y}(x)|+|\nabla_h^2D^{h,y}_i\tilde{G}_{h,y}(x)|\\
	&\le C\left(\frac{1}{d(y)^{n-1}}+\frac{1}{(|x-y|+h)^{n-1}}\right)\,.
\end{align*}
This implies \eqref{e:estnablax2nablayGuppclose} if we use that $|x-y|_\infty\le\frac{d(y)}{4}$ and $d(y)\ge 10h$ so that $|x-y|+h\le Cd(y)$.

\emph{Step 2: Proof of \eqref{e:estnablaxnablayGuppclose}}\\
If $d(y)<4h$ we can repeat the trivial estimate from the previous step, so assume $d(y)\ge4h$.

We want to integrate \eqref{e:estnablax2nablayGuppclose} along a suitable path. So let $(x^{(k)})_{k=0}^L$ be a straight path from $x^{(0)}=x$ to a closest point $x^{(L)}\in Q_{d(y)/4}(y)\setminus Q_{d(y)/4-h}(y)$. This path will have length $Lh=\left\lfloor\frac{d(y)}{4}\right\rfloor_h-|x-y|_\infty$. By Lemma~\ref{l:estderfar} we have
\begin{align}
|\nabla_{h,x}\nabla_{h,y}G_h(x^{(L)},y)|&\le C\frac{(d(x^{(L)})+h)(d(y)+h)}{|x^{(L)}-y|^n}\nonumber\\
&\le C\frac{(d(y)+h)^2}{|d(y)+h|^n}\le C\frac{1}{|d(y)+h|^{n-2}}\,.\label{e:estderclose1}
\end{align}
Furthermore \eqref{e:estnablax2nablayGuppclose} implies that
\begin{equation}
|\nabla_{h,x}^2\nabla_{h,y}G_h(x^{(k)},y)|\le \frac{C}{(|x^{(k)}-y|+h)^{n-1}}\le\frac{C}{(|x-y|+(k+1)h)^{n-1}}\,.\label{e:estderclose2}
\end{equation}
Now we can integrate \eqref{e:estderclose2} along $(x^{(k)})_{k=0}^L$ and use \eqref{e:estderclose1}, and after a short calculation we arrive at \eqref{e:estnablaxnablayGuppclose}.

\emph{Step 3: Proof of \eqref{e:estnablax2Guppclose}}
If $d(y)<77h$ we can again use the trivial estimate from Step 1, so assume $d(y)\ge77h$.

This is similar to the previous step: We choose a shortest-possible path $(y^{(k)})_{k=0}^L$ from $y^{(0)}=y$ to a point $y^{(L)}\in Q_{d(x)/6}(x)\setminus Q_{d(x)/6-h}(y)$. Then $|y^{(k)}-x|_\infty\le\frac{d(x)}6$, so that $\frac56d(x)\le d(y^{(k)})\le\frac76d(x)$ and hence \[|y^{(k)}-x|_\infty\le\frac{d(x)}6\le\frac{d(y^{(k)})}5\,.\]
Therefore we can apply \eqref{e:estnablax2nablayGuppclose} at the point $(x,y^{(k)})$ for each $k$ and conclude
\begin{equation}
|\nabla_{h,x}^2\nabla_{h,y}G_h(x,y^{(k)})|\le \frac{C}{(|x-y^{(k)}|+h)^{n-1}}\le\frac{C}{(|x-y|+(k+1)h)^{n-1}}\,.\label{e:estderclose3}
\end{equation}
On the other hand,
\[d(y^{(L)})\ge\frac56d(x)\ge\frac56\frac78d(y)\ge56h\]
so that
\[|y^{(L)}-x|_\infty\ge\frac{d(x)}6-h\ge\frac{d(y^{(L)})}{7}-h>\frac{d(y^{(L)})}{8}\,.\]
This means that we can apply \eqref{e:estnablaxnablayGuppfar} at the point $(x,y^{(L)})$ and conclude
\begin{equation}
|\nabla_{h,x}^2G_h(x,y^{(L)})|\le C\frac{(d(y^{(L)})+h)^2}{|x-y^{(L)}|^n}\le C\frac{(d(y)+h)^2}{|d(y)+h|^n}\le C\frac{1}{|d(y)+h|^{n-2}}\,.\label{e:estderclose4}
\end{equation}
Now we can integrate \eqref{e:estderclose3} along the path $(y^{(k)})_{k=0}^L$ and use the estimate \eqref{e:estderclose4} for the one endpoint to obtain \eqref{e:estnablax2Guppclose}.

\emph{Step 4: Proof of \eqref{e:estnablaxGuppclose}}\\
We could try to prove this by integrating \eqref{e:estnablax2Guppclose} along a path. However, this turns out to be not sharp enough at least if $n=3$ (we would get a logarithmic term instead of a constant term). Instead we will use the Sobolev inequality on the function $H_{h,y}$ from Step 1. Thereby we get a bound for $\nabla_{h,y}G_h(x,y)$ if $x,y$ are close. By the symmetry of $G_h$ we can turn this into a bound for $\nabla_{h,x}G_h(x,y)$.

If $d(y)<10h$ we can again use the trivial estimate from Step 1, so assume $d(y)\ge10h$. Recall the function $H_{h,y}$ from Step 1. If we use the Sobolev and Poincaré inequality on $Q_{d(y)+h/2}(y)$ and the estimate \eqref{e:estderclose5} we obtain
\begin{align*}
	\|D^{h,y}_iH_{h,y}\|_{L^\infty(Q_{d(y)+h/2}(y))}&\le C(d(y)+h/2)^{2-\frac n2}\|\nabla_h^2D^{h,y}_iH_{h,y}\|_{L^2(Q_{d(y)+h/2}(y))}\\
	&\le Cd(y)^{2-\frac n2}\|\nabla_h^2D^{h,y}_iH_{h,y}\|_{L^2(\R^n)}\\
	&\le Cd(y)^{3-n}
\end{align*}
and therefore
\[|\nabla_{h,y}H_{h,y}(x)|\le Cd(y)^{3-n}\]
for any $x\in Q_{d(y)}(y)$. Now we can use \eqref{e:estgreenfull1} and the fact that $D^{h,y}_iH_{h,y}=D^{h,y}_iG_{h,y}-D^{h,y}_i\tilde{G}_{h,y}$ in $Q_{d(y)/2}(y)$ and obtain
\[|D^{h,y}_iG_{h,y}(x)|\le|D^{h,y}_iH_{h,y}(x)|+|D^{h,y}_i\tilde{G}_{h,y}(x)|\le Cd(y)^{3-n}\]
for any $x\in Q_{d(y)/2}(y)$ and any $i\in\{1,\ldots, n\}$. By the symmetry of $G_h$ in $x$ and $y$ we conclude that also
\begin{equation}
	|D^{h,x}_iG_{h,x}(y)|\le Cd(x)^{3-n}\label{e:estderclose6}
\end{equation}
for any $y\in Q_{d(x)/2}(x)$.

Now in the setting of \eqref{e:estnablaxGuppclose} we are given $x,y$ with $|y-x|_\infty\le\frac{d(y)}4$. These satisfy $\frac34d(y)\le d(x)\le\frac54d(y)$, so that $|y-x|_\infty\le\frac13d(x)$ and in particular $y\in Q_{d(x)/2}(x)$. Thus we can apply \eqref{e:estderclose6} and obtain
\[|D^{h,x}_iG_{h,x}(y)|\le Cd(x)^{3-n}\le Cd(y)^{3-n}\,,\]
which implies \eqref{e:estnablaxGuppclose}.

\emph{Step 5: Proof of \eqref{e:estGuppclose}}\\
This follows immediately from \eqref{e:upperbound2}. 
\end{proof}
\begin{proof}[Proof of Theorem~\ref{t:mainthmh} i)]
Now that we have proved Lemma~\ref{l:estderclose} and Lemma~\ref{l:estderfar} the proof is straightforward. First observe that it suffices to consider $x,y\in\Lambda_h^n$ as otherwise $G_h$ and its relevant derivatives are trivially 0.

We claim that we can combine \eqref{e:estnablaxnablayGuppfar} and \eqref{e:estnablaxnablayGuppclose} to obtain \eqref{e:estnablaxnablayGupp}. Indeed, if $|x-y|_\infty\le\frac{d(y)}{8}$ we have $d(y)\le\frac87d(x)$ and $|x-y|+h\le\sqrt{n}|x-y|_\infty+h<d(y)+h$ which implies
\[1\le\frac{(d(x)+h)(d(y)+h)}{(|x-y|+h)^2}\]
and we are done by \eqref{e:estnablaxnablayGuppclose}.

If however $|x-y|_\infty>\frac{d(y)}{8}$, then we have in particular $|x-y|\ge h$, so that $|x-y|+h\le2|x-y|$. We also have $d(y)\le8|x-y|$ and $d(x)\le9|x-y|$ and we easily see that
\[\frac{(d(x)+h)(d(y)+h)}{|x-y|^n}\le \frac{C}{(|x-y|+h)^{n-2}}\]
so we are done by \eqref{e:estnablaxnablayGuppfar}.

Similarly, we can combine \eqref{e:estnablax2Guppfar} and \eqref{e:estnablax2Guppclose} into the estimate
\[|\nabla_{h,x}^2G_h(x,y)|\le \begin{cases}C\log\left(1+\frac{(d(y)+h)^2}{(|x-y|+h)^2}\right)&\quad n=2\\ C\min\left(\frac{1}{|x-y|+h},\frac{(d(y)+h)^2}{(|x-y|+h)^3}\right)&\quad n=3\end{cases}\,.\]
This is not quite \eqref{e:estnablax2Gupp}, but it implies \eqref{e:estnablax2Gupp} unless $d(y)=0$. On the other hand, if $d(y)=0$ then $y\in\partial\Lambda_h^n$. Therefore $G_{h,y}$ is identically 0, so that $\nabla_{h,x^2}G_h(x,y)=0$ and \eqref{e:estnablax2Gupp} holds as well.

Similarly we can combine \eqref{e:estnablaxGuppfar} and \eqref{e:estnablaxGuppclose}, and \eqref{e:estGuppfar} and \eqref{e:estGuppclose} into 
\begin{align*}
	|\nabla_{h,x}G_h(x,y)|&\le C\min\left((d(y)+h)^{3-n},\frac{(d(x)+h)d(y)^2}{(|x-y|+h)^n}\right)\,,\\
	|G_h(x,y)|&\le C\min\left((d(x)+h)^{2-\frac n2}(d(y)+h)^{2-\frac n2},\frac{(d(x)+h)^2(d(y)+h)^2}{(|x-y|+h)^n}\right)
\end{align*}
respectively. These estimates imply \eqref{e:estnablaxGupp} and \eqref{e:estGuppclose}, except in the cases $d(x)=0$ or $d(y)=0$, which are again trivial. 
\end{proof}

\begin{remark}\label{r:otherest}
As a byproduct of the proofs of Lemma~\ref{l:estderclose} and Lemma~\ref{l:estderfar} we proved the estimates \eqref{e:estnablax2nablayGuppfar} and \eqref{e:estnablax2nablayGuppclose} which can easily be combined into the estimate
\begin{equation}
	|\nabla_{h,x}^2\nabla_{h,y}G_h(x,y)|\le C\min\left(\frac{1}{(|x-y|+h)^{n-1}},\frac{d(y)+h}{(|x-y|+h)^n}\right)
\end{equation}
for any $x,y\in(h\Z)^n$.

With the same method of proof it is possible to prove an estimate for $\nabla_{h,x}^2\nabla_{h,y}^2G_h$ as well. One again considers $H_{h,y}=G_{h,y}-\eta_h\tilde{G}_{h,y}$ in Lemma~\ref{l:estderclose} and Lemma~\ref{l:estderfar} and derives estimates for $\|\nabla_{h,x}^2\nabla_{h,y}^2H_{h,y}\|_{L^2(\R^n)}$. In combination with the pointwise estimates for $\tilde{G}_h$ (in particular \eqref{e:estgreenfull4}) these again yield estimates for $\nabla_{h,x}^2\nabla_{h,y}^2G_h$ in the two regimes where $x$ and $y$ are far away and close together, respectively. The final result is
\begin{equation}
	|\nabla_{h,x}^2\nabla_{h,y}^2G_h(x,y)|\le\frac{C}{(|x-y|+h)^n}
\end{equation}
for any $x,y\in(h\Z)^n$.

Actually it is even possible to derive estimates for higher derivatives $\nabla_{h,x}^a\nabla_{h,y}^bG_h$, at least when $a\le2$ or $b\le2$. However we cannot expect these estimates to be optimal any more, because high derivatives are increasingly divergent near the singular boundary points, and our approach does not really capture this behaviour.
\end{remark}

\subsection{Convergence of Green's functions}
\begin{proof}[Proof of Corollary \ref{c:convergence_G}]
We begin with the proof of assertion i). We can assume that $h\leq\frac13$.
There exists a unique $y_h \in  \Lambda_h^n$ such that $y \in y_h + [-\frac{h}{2}, \frac{h}{2})^2$.
Set $u_h(x) = G_h(x,y_h)$. We extend $u_h$ by zero to $(h\Z)^n \setminus \inte \Lambda_h^n$. 

To prove (i) we have to show that $u_h$  converges uniformly to $G(\cdot, y)$. 
Testing the equation for $\Delta^2_h u_h$ with $u_h$ we get (see Lemma~\ref{l:upperbound})
\[ \| \nabla_h^2 u_h \|_{L^2(\R^n)} \le C d^{2-\frac{n}{2}}(y_h) \le C\,. \] 
The discrete Sobolev-Poincaré inequality implies in particular that the $u_h$ are uniformly Hölder continuous
\begin{equation} \label{e:uh_hoelder}
[u_h]_{C_h^{0, \frac14}(\R^n)} \le C.
\end{equation}

Denote by $J_h$ the interpolation operator introduced in Section \ref{s:interpolation}.
From Proposition \ref{p:interpolationproperties} vi) and the Poincaré inequality we deduce that the sequence  $J_h u_h$ is bounded in 
$W^{2,2}(\R^n)$ and $J_h u_h  = 0$ in $\R^n \setminus (-3h, 1 + 3h)^n$. It follows that for a subsequence 
\[ 
J_{h_k} u_{h_k} \rightharpoonup u \quad \hbox{in $W^{2,2}(\R^n)$}\,, \qquad u = 0 \text{ in }\R^n \setminus (0,1)^n\,.
\]
From the uniform Hölder continuity \eqref{e:uh_hoelder}  and Proposition~\ref{p:interpolationproperties} iii), iv) and vi) we deduce that, for any $x\in(-3h,1+3h)^n$,
\begin{align*}
	|J_{h_k} u_{h_k}(x)-I^{pc}_{h_k} u_{h_k}(x)|&=|J_{h_k}(u_{h_k}(\cdot)-u_{h_k}(x))(x)|\\
	&\le C\|u_{h_k}-u_{h_k}(x)\|_{L^\infty(Q_{3h_k}(x))}\le C {h_k}^{\frac14}
\end{align*}
and therefore
\[
\sup_{x\in(-1,2)^n}|J_{h_k} u_{h_k}(x) - I^{pc}_{h_k} u_{h_k}(x)| \le C {h_k}^{\frac14}\,.
\]
In connection with the compact embedding from $W^{2,2}_0((-1,2)^n)$ to $C^0((-1,2)^n)$ we conclude that
\begin{equation} \label{e:uniform_conv}
I^{pc}_{h_k} u_{h_k} \to u  \quad \text{uniformly}\,.
\end{equation}
If we can show that $u(x) = G(x,y)$ then by uniqueness of the limit it follows that the convergences
above do not only hold along a particular subsequence $h_k \to 0$ but for every subsequence $h_k \to 0$
and we are done.

To show that $u(x) = G(x,y)$ we use that by definition of $G_h( \cdot, y_h)$ we have for each $\varphi \in C_c^5( (0,1)^n)$
\begin{align*}
\varphi(y_{h_k}) = &  \sum_{ x \in \inte\Lambda_h} \Delta_{h_k}^2 u_{h_k}(x) \varphi(x)  h_k^n 
=    \sum_{ x \in \inte\Lambda_h} u_{h_k}(x)    \Delta_{h_k}^2 \varphi(x)  h_k^n \\
= &  \int_{(0,1)^n} I^{pc}_{h_k} u_{h_k} I^{pc}_{h_k} \Delta_{h_k}^2\varphi(x)  \ud x\,.
\end{align*}
Now by Taylor expansion $|I^{pc}_{h_k} \Delta_{h_k}^2 \varphi - \Delta^2 \varphi| \le C {h_k}$. 
Together with \eqref{e:uniform_conv} we get
\[
\varphi(y) = \lim_{k \to \infty} \varphi(y_{h_k}) = \lim_{k \to \infty}  \int_{(0,1)^n} I^{pc}_{h_k} u_{h_k} I^{pc}_{h_k} \Delta_{h_k}^2 \varphi_h\ud x
= \int_{(0,1)^n} u\Delta^2 \varphi\ud x\,.
\]
Thus $\Delta^2 u = \delta_y$ in the sense of distributions. Since we also know that $u \in W_0^{2,2}((0,1)^n)$ we conclude that
$u(x) = G(x,y)$ as desired. 

To prove ii) note that the estimates in Theorem~\ref{t:mainthmh} show that the second discrete derivatives are bounded in 
$L^p$ for all $p < \infty$. Hence by the discrete Sobolev embedding theorem the discrete first deriatives are bounded in 
$C^{0, \alpha}$ for all $\alpha<1$. This implies that 
\begin{equation}  \label{e:osc_first_derivatives}
 |I^{pc}_h \nabla_h u -  \nabla J_h u_h| \le C h^\alpha.
 \end{equation} 
Moreover the $L^p$ bound on the discrete second derivatives and \eqref{e:interpolation1} give a bound of $J u_h$ in $W^{2,p}$.
Hence a subsequence of $J_h  u_h$ converges in $C^{1, \alpha}$ to $G(\cdot, y)$. Since the limit is unique, the whole sequence
converges in $C^{1, \alpha}$ to $G$. Together with  \eqref{e:osc_first_derivatives} this yields uniform convergence
of the discrete first derivatives. 

The local compactness argument in  Section~\ref{s:inner_special} (and a diagonalisation argument) shows that a subsequence
of $I^{pc}_h \nabla_h^2 u_h$ converges in $L^2_{\rm loc}( (0,1)^2 \setminus \{y\})$ to a function $v$. 
Since  $I^{pc}_h \nabla_h^2 u_h$ is also bounded in $L^q$ for some $q > 2$ we get strong convergence in $L^2( (0,1)^2)$. 
Using again the $L^q$ bound we get strong convergence in all $L^p$ with $p < q$. Since we  have $L^q$ bounds for all
$ q < \infty$ we get strong convergence for all $p< \infty$. It remains to show that $v = \nabla^2 G(\cdot, y)$. 
To obtain this identity we can use discrete integration by parts and pass to the limit on both sides, as in the proof that
$\Delta^2 u = \delta_y$.

The proof of (iii) is similar. Uniform boundedness of the discrete derivatives follows directly from   Theorem~\ref{t:mainthmh}.
This theorem also shows that the second discrete derivatives are uniformly bounded on the complement of any cube $Q_r(y)$. 
It follows that the functions $u_h$ are uniformly Lipschitz on the complement of any cube $Q_r(y)$ and we obtain locally 
uniform convergence of $I^{pc}_h \nabla_h u_h$ in the complement of those cubes as in the proof of (ii). Combined with the uniform boundedness we immediately conclude convergence of $I^{pc}_h \nabla_h u_h$ in $L^p$ for all $p<\infty$.

The proof of $L^p$ convergence of $I_h^{pc}\nabla_h^2u_h$ for $p<3$ is again analogous to the argument for $n=2$. 
\end{proof}

\section*{Acknowledgements}
The authors are very grateful to G. Dolzmann for many inspiring conversations and suggestions. 
Indeed,  G. Dolzmann and the first author derived in unpublished work estimates in the interior 
and near the regular boundary points using discrete Campanato estimates in the spirit of \cite{Dolzmann1993,Dolzmann1999}
and developed a version of the duality argument in Section~\ref{s:inner_outer_decay}. The authors are also very grateful to E. Süli for his insightful comments on a preliminary version of this manuscript.

This work has its roots in very interesting discussions with J.D. Deuschel, N. Kurt and E. Bolthausen in the framework
of the Berlin-Leipzig DFG Research Unit 718 'Analysis and stochastics in complex physical systems'. The authors would like to thank A. Cipriani for pointing out connections to sandpile models.

SM was supported by the DFG through the Hausdorff Center 
for Mathematics (EXC59) and the CRC 1060 'The mathematics of emergent effects', 
project A04. FS was supported by the Hausdorff Center 
for Mathematics through the Bonn International Graduate School for Mathematics (BIGS), and by the German National Academic Foundation. 
SM also strongly benefitted from the Trimester Programme 'Mathematical challenges
of materials science and condensed matter physics'
at the Hausdorff Research Institute for Mathematics (HIM).

\bibliographystyle{alpha_edited}
\bibliography{DiscreteBilaplacianArxivVersion}  

\end{document}